\documentclass[11pt]{article}
\usepackage[letterpaper, margin=1in]{geometry}

\usepackage{makecell}

\usepackage{amsfonts}
\usepackage{amsmath}
\usepackage{amsthm}
\usepackage{thm-restate}
\usepackage{graphicx}
\usepackage{tikz}
\usetikzlibrary{calc}
\usetikzlibrary{shapes.geometric, arrows, positioning}
\usepackage{xcolor}
\usepackage[colorlinks=true, allcolors=blue]{hyperref}

\usepackage{algorithm}
\usepackage{algpseudocode}
\usepackage{tcolorbox}
\usepackage{subcaption}

\newtheorem{defn}{Definition}
\newtheorem{fact}{Fact}

\newtheorem{prop}{Proposition}
\newtheorem{lem}{Lemma}

\newtheorem{thm}{Theorem}

\DeclareMathOperator{\poly}{poly}

\usepackage{cite}
\usepackage{algorithm}
\usepackage{algpseudocode}

\newcommand{\mc}[1]{\mathcal{#1}}
\newcommand{\ol}{\overline}

 \newcommand{\aaron}[1]{}

 \newcommand{\seri}[1]{}
\newcommand{\E}{\mathop{\textbf{E}}}

\title{Round Elimination via Self-Reduction: Closing Gaps for Distributed Maximal Matching}

\author{
  Seri Khoury\thanks{INSAIT, Sofia University “St. Kliment Ohridski”. \texttt{seri.khoury@insait.ai}} 
  \and 
  Aaron Schild\thanks{Google Research. \texttt{aaron.schild@gmail.com}}
}

\date{}

\begin{document}
\maketitle

\begin{abstract}
Maximal Independent Set (MIS) and Maximal Matching (MM) play a vital role in distributed symmetry breaking. Despite decades of research, the complexity of both problems in the standard \textsc{LOCAL} model remains unresolved, and several gaps between the best-known upper and lower bounds persist. For \(n\)-node graphs with maximum degree \(\Delta\), the best current upper bound for randomized algorithms is \(O(\log \Delta + \operatorname{poly}(\log \log n))\), shown by Barenboim, Elkin, Pettie, and Schneider for MM [FOCS'12, JACM'16], and by Ghaffari for MIS [SODA'16]. On the other hand, the best-known lower bound for the two problems is
\(
\Omega\left(\min\left\{\sqrt{\frac{\log n}{\log \log n}}, \frac{\log \Delta}{\log \log \Delta}\right\}\right),
\)
shown by Kuhn, Moscibroda, and Wattenhofer [PODC'04, JACM'16].

In this work, we present an \(\Omega\left(\min\{\log \Delta, \sqrt{\log n}\}\right)\) lower bound for MM in $\Delta$-ary trees against randomized algorithms. By a folklore reduction, the same lower bound applies to MIS, albeit not in trees. As a function of \(n\), this is the first advancement in our understanding of the randomized complexity of the two problems in more than two decades. As a function of \(\Delta\), this shows that the current upper bounds are optimal for a wide range of \(\Delta \in 2^{O(\sqrt{\log n})}\), answering an open question by Balliu, Brandt, Hirvonen, Olivetti, Rabie, and Suomela [FOCS'19, JACM'21].

Moreover, our result implies a surprising and counterintuitive separation between MIS and MM in trees, as it was very recently shown that MIS in trees can be solved in $o(\sqrt{\log n})$ rounds. While MIS can be used to find an MM in general graphs, the reduction does not preserve the tree structure when applied to trees. Our separation shows that this is not an artifact of the reduction, but a fundamental difference between the two problems in trees. This also implies that MIS is strictly harder in general graphs compared to trees. 

Our main technical contribution is a novel technique in which we show that there is a self-reduction from a matching problem in \(r\) rounds to the same matching problem in \(r - 1\) rounds (with slightly weaker probabilistic guarantees). Conceptually, this resembles the celebrated round elimination technique, which transforms an \(r\)-round algorithm for a problem \(\Pi\) into an \((r - 1)\)-round algorithm for a different problem \(\Pi'\). However, our proof differs significantly from the round elimination framework in several fundamental aspects. One of the key concepts we analyze in achieving our result is \emph{vertex survival probability}, where we show that after \(r \ll \min\{\log \Delta, \sqrt{\log n}\}\) rounds, any algorithm that finds a matching must leave two surviving unmatched nodes that are adjacent.



\end{abstract}


\newpage
\section{Introduction}

Among the most extensively studied problems in distributed symmetry breaking are the Maximal Independent Set (MIS) and Maximal Matching (MM) problems~\cite{Luby86,AlonBI86,RozhonG20,2013Barenboim,BalliuBHORS21,Ghaffari16,BarenboimEPS16,KuhnMW16,BarenboimE10,Linial92,Naor91,0001GHIR23,0001G23,GhaffariP19,ecomposition,BalliuBO22,Balliu0KO22,Balliu0KO21,LenzenW11,MetivierRSZ11,wattenhofer2020mastering,SchneiderW08}. Classical randomized algorithms find an MIS or MM in $O(\log n)$ rounds with high probability in the classical LOCAL model~\cite{Luby86,AlonBI86}.\footnote{We say that an algorithm succeeds with high probability if it succeeds with probability $1-1/n^c$ for an arbitrarily large constant $c>1$.} In this model, there is a network of $n$ nodes that can communicate via synchronized communication rounds, and their goal is to solve some task (e.g., find an MIS or MM) while minimizing the number of communication rounds. 

Remarkably, the classical $\approx 40$-year-old algorithms by~\cite{Luby86,AlonBI86} remain the best known for general graphs. Whether there exists a $o(\log n)$-round algorithm for MIS or MM remains one of the most fascinating mysteries in the field~\cite{2013Barenboim,Ghaffari16,BarenboimEPS16,HarrisSS18,wattenhofer_podc,abs-2406-19430}. On the other hand, the best-known lower bound as a function of $n$ remains the $21$-year-old $\Omega(\sqrt{\log n/\log\log n})$ lower bound by Kuhn, Moscibroda, and Wattenhofer~\cite{KuhnMW04,KuhnMW16}, which applies also to randomized algorithms. For deterministic ones, the goal is to approach as closely as possible the logarithmic state-of-the-art bound for randomized algorithms.\footnote{The current best upper bound for deterministic algorithms is roughly \(O(\log^{5/3} n)\) rounds~\cite{ecomposition}, and the current best lower bound is \(\Omega(\min\{\Delta, \log n / \log \log n\})\)~\cite{BalliuBHORS21}. For more on deterministic algorithms, we refer the reader to the excellent surveys by Suomela~\cite{Suomela13} and Rozhon~\cite{abs-2406-19430}. The focus of this work is on randomized algorithms.}

Unfortunately, our understanding of the randomized complexity of the two problems as a function of $n$ has not improved for more than two decades. Nevertheless, several breakthroughs have substantially advanced our understanding of their complexity as a function of the maximum degree $\Delta$. First, when expressed as a function of both $n$ and $\Delta$, the lower bound of~\cite{KuhnMW16} becomes $\Omega(\min\{\log\Delta/\log\log\Delta,\sqrt{\log n/\log\log n}\})$. A few years later, Barenboim, Elkin, Pettie, and Schneider~\cite{BarenboimEPS16} showed that MM can be solved in $O(\log\Delta + \poly(\log\log n))$ rounds, almost matching the lower bound of~\cite{KuhnMW16} as a function of $\Delta$. Moreover, they presented a similar algorithm for MIS, with a worse dependence on $\Delta$, namely $\log^2\Delta$ rather than $\log\Delta$. Subsequently, Ghaffari presented an improved algorithm for MIS that takes $O(\log\Delta + \poly(\log\log n))$ rounds, closing the gap between the state-of-the-art algorithms for MIS and MM.

Hence, we now know that the logarithmic dependence on $n$ from the classical algorithms of~\cite{AlonBI86,Luby86} can be replaced by a logarithmic dependence on $\Delta$. Interestingly, for graphs with small $\Delta \ll \log\log n$, this dependence is exponentially higher, due to the $\Omega(\min\{\Delta,\log\log n/\log\log\log n\})$ lower bound by Balliu, Brandt, Hirvonen, Olivetti, Rabie, and Suomela~\cite{BalliuBHORS21}, which is tight in that regime of $\Delta$~\cite{Kuhn09,BarenboimE09,BarenboimEK14}.

It is worth mentioning that at the time of the publication of the logarithmic-in-$\Delta$ algorithms for MIS and MM by~\cite{Ghaffari16} and~\cite{BarenboimEPS12}, the community believed that an $\Omega\left(\min\left\{\log\Delta,\sqrt{\log n}\right\}\right)$ lower bound had been established (see also~\cite{abs-1011-5470}). It was thought that such a lower bound held even against any constant-factor approximate minimum vertex cover (MVC). Due to a trivial reduction from MM to 2-approximate MVC (that simply adds all the nodes that are matched in the MM to the vertex cover), the same lower bound would apply to MM and MIS. However, it was later clarified in~\cite{KuhnMW16,Bar-YehudaCS17} that this belief was incorrect. In particular, a constant-approximate MVC can be found in $O\left(\log\Delta/\log\log\Delta\right)$~\cite{Bar-YehudaCS17} rounds, which is tight due to the original (correct) $\Omega(\min\{\log\Delta/\log\log\Delta,\sqrt{\log n/\log\log n}\})$ lower bound of~\cite{KuhnMW16}. 

While the upper bound of~\cite{Bar-YehudaCS17} settled the complexity of constant-approximate MVC as a function of~$\Delta$, it does not apply to MIS or MM. In particular, as discussed earlier, for very small values of~$\Delta \ll \log\log n$, it is known that MIS and MM require $\Omega(\Delta)$ rounds~\cite{Balliu0KO21}. However, this lower bound does not persist if we allow a small $\poly(\log\log n)$ dependence on~$n$. In fact, determining the correct dependence on~$\Delta$ when allowing a small $\poly(\log\log n)$ dependence on~$n$ remains a widely known open question in the field. This question is explicitly mentioned in the seminal paper by~\cite{BalliuBHORS21},\footnote{The significance of the paper by~\cite{BalliuBHORS21} was recognized with a best paper award at FOCS 2019.} which presented the lower bounds for small values of~$\Delta \ll \log\log n$ via the celebrated round elimination technique. The authors specifically state:

\begin{center}
    ``.. when we consider
algorithms that take $\poly(\log\log n)$ rounds as a function of $n$, there is a gap in the complexity as a function of $\Delta$: the
current lower bound is $\Omega\left(\log\Delta/\log\log\Delta\right)$, while the upper bound is $O(\log\Delta)$; it is wide open whether the speedup
simulation technique is applicable in the study of such questions."~\cite{BalliuBHORS21}
\end{center}

Moreover, this gap also translates to a $\Theta\left (\sqrt{\log\log n}\right)$ gap in trees. In particular, the algorithms of~\cite{Ghaffari16,BarenboimEPS16} imply an $O(\sqrt{\log n})$ upper bound in trees, and the $\Omega\left(\sqrt{\log n/\log\log n}\right)$ lower bound by~\cite{KuhnMW16} can be extended to trees~\cite{BarenboimEPS16,BalliuGKO23}.




\paragraph{Our Contribution:} In this work, we present an $\Omega(\min\{\log\Delta,\sqrt{\log n}\})$ lower bound against randomized algorithms for MM in $\Delta$-ary trees.\footnote{A $\Delta$-ary tree is a tree where all non-leaf nodes have degree $\Delta$.} By a folklore reduction, the same lower bound applies to MIS, albeit not in trees.\footnote{In this reduction, we find an MM by using an MIS algorithm on the line graph of the original graph. In the line graph, each node represents an edge of the original graph, and two nodes are connected if the two corresponding edges in the original graph intersect. Clearly, when applied on trees, this reduction doesn't maintain the tree structure.}

\begin{restatable}{thm}{thmmain}\label{thm:main}
Any $C_0\min\{\log\Delta, \sqrt{\log n}\}$-round randomized LOCAL algorithm cannot produce a maximal matching with probability greater than $\Delta^{-1/1000}$ in $n$-vertex $\Delta$-regular graphs, where $C_0 :=10^{-10}$. The same lower bound applies to $\Delta$-ary trees.
\end{restatable}

\paragraph{Implications:} 

\begin{enumerate}
    \renewcommand{\labelenumi}{(I\arabic{enumi})}
    \item We bring the first advancement in our understanding of the randomized complexity of MIS and MM as a function of $n$ in more than two decades.  
    \item We close the $\Theta(\log\log\Delta)$ gap discussed above for algorithms that are allowed to have a small $\poly(\log\log n)$ (or even $o(\sqrt{\log n}))$ dependence on $n$, showing the optimality of the $10$-year old algorithm of~\cite{Ghaffari16} and the $13$-year old algorithm of ~\cite{BarenboimEPS12,BarenboimEPS16} for MIS and MM, respectively, for a wide regime of $\Delta\in [\poly(\log n),2^{O(\sqrt{\log n})}]$. This answers the open question by~\cite{BalliuBHORS21} that we discussed above.
    \item We close the $\Theta(\sqrt{\log\log n})$ gap for MM in trees, settling its complexity at $\Theta(\sqrt{\log n})$.

    \item We show a surprising and counterintuitive separation between MIS and MM in trees. In very recent work~\cite{UBMIS}, a $o(\sqrt{\log n})$-round algorithm for MIS in trees was developed.\footnote{In particular, they show an $O(\sqrt{\log n/\log(\log^*(n))})$-round algorithm for MIS in trees (as well as an $O(\log\Delta/\log(\log^*\Delta)+\poly(\log\log n))$-round algorithm for graphs with girth 7).} That is, while MM can only be easier than MIS in general graphs, our lower bound shows that it is \emph{harder} in trees. 
    \item Together with the $O(\log\Delta/\log\log\Delta)$ algorithm for $O(1)$-approximate MVC of~\cite{Bar-YehudaCS17}, our lower bound implies a separation between MM and constant-approximate MVC for a wide range of $\Delta\in 2^{o(\sqrt{\log n} \cdot \log\log n)}$. Prior to our result, such a separation was only known for $\Delta\in \log^{o(1)} n$, due to the lower bound of~\cite{BalliuBHORS21}. 
    \item We exponentially improve the best-known lower bound as a function of $n$ for MIS and MM in regular graphs from $\Omega(\log\log n/\log\log\log n)$~\cite{BalliuBHORS21} to $\Omega(\sqrt{\log n})$. This also implies a higher separation between MM and fractional matching in regular graphs, as the latter is trivial.\footnote{By setting the fractional value of each edge to be $1/\Delta$.} 
\end{enumerate}





\paragraph{Roadmap:} In Section~\ref{sec:TechnicalOverview}, we begin with a technical overview of our result and explain the key differences compared to prior techniques. Then, in Section~\ref{sec:Prelims}, we introduce some definitions and preliminaries. The technical heart of the paper is given in Sections~\ref{sec:ProofBody} and~\ref{sec:VertexSurvivalProbability}.  Section~\ref{sec:ProofBody} contains the main structure of the proof, and Section~\ref{sec:VertexSurvivalProbability} contains a proof of a novel key lemma concerning \emph{vertex survival probability}, which paved the way towards our result.  

\section{Technical Overview}\label{sec:TechnicalOverview}

\textbf{Prior Techniques in a Nutshell:} There are essentially two techniques to obtain lower bounds for MM and MIS in the LOCAL model: indistinguishability and round elimination. The first was developed by~\cite{KuhnMW04} and was used to obtain a lower bound for constant-approximate MVC (even against polylogarithmic approximation). They showed a construction in which there is an unbalanced bipartite graph where the endpoints of each edge have the same view. Since each MVC should not include a large fraction of the bigger side in the bipartite graph, and since the nodes cannot distinguish the bigger side from the smaller one (as they have the same views), their construction implies a lower bound against approximate MVC. Moreover, the radius of the equal views is $\Omega\left(\min\left\{\frac{\log\Delta}{\log\log\Delta}, \sqrt{\frac{\log n}{\log\log n}}\right\}\right)$, implying the same lower bound on the number of rounds. As discussed earlier, this technique cannot be improved due to the $O\left(\frac{\log\Delta}{\log\log\Delta}\right)$-round $O(1)$-approximate MVC algorithm of \cite{Bar-YehudaCS17} (unless it would be developed in a very different way for MM or MIS, which would require novel and substantially different ideas).

The second technique, round elimination, has been used to obtain several lower bounds \cite{BalliuBHORS21,BO20,BalliuB0O24,Balliu0KO23,BalliuBO22,Balliu0KO21,Balliu0KO22,Brandt19,BrandtFHKLRSU16}, including for MM and MIS \cite{BalliuBHORS21}. This technique transforms an $r$-round algorithm for a problem $\Pi$ into an $(r-1)$-round algorithm for another problem $\Pi'$. While this approach has been very successful for several problems, including leading to the state-of-the-art lower bounds for deterministic algorithms, its performance is limited as a function of $n$ for randomized ones. In particular, this technique obtains only a $\tilde{\Omega}(\log\log n)$ lower bound against randomized algorithms, which is exponentially smaller than the lower bound obtained by~\cite{KuhnMW16} as a function of $n$. Furthermore, the approach is somewhat complicated, as deducing the problem $\Pi'$ benefits from the use of a computer program~\cite{Olivetti20}.

\paragraph{Our Technique:} We prove Theorem \ref{thm:main} by \emph{self-reduction} round elimination. We show that it is possible to make $\Pi = \Pi'$ with the help of a novel metric on which to assess randomized algorithms: \emph{vertex survival probability} (which will be discussed shortly). In particular, we do not need to use a computer program to derive $\Pi'$ in our proof.  


For simplicity and without loss of generality, consider a LOCAL algorithm that always outputs a matching. I.e., this algorithm does not suffer from a small probability failure event where the output contains two incident edges. In other words, we assume that the algorithm always outputs a matching, and we analyze the failure probability of that matching not being maximal. We formally deal with this assumption in the proof of Theorem~\ref{thm:main} (we simply add an additional round that deletes any intersecting edges from the output). We call an algorithm that always outputs a matching a \emph{matching-certified algorithm}.

\paragraph{Vertex Survival Probability:} The vertex survival probability of a matching-certified algorithm $f$ (denoted $P_f$) is the probability that a given node is unmatched in the matching found by $f$. On an infinite $\Delta$-regular tree (or a finite $\Delta$-regular graph with sufficiently high girth), all nodes have the same view, and thus have the same (but correlated) output distribution as a function of the algorithm's randomness.\footnote{As we discuss shortly, since we are interested in randomized algorithms, we may assume without loss of generality that the nodes do not have unique IDs to begin with. Indeed, the unique IDs can be simply obtained from the randomness tapes.} This means that $P_f$ is the same for any vertex and is thus well-defined. We show a $2^{-O(r)}$-lower bound on the vertex survival probability of $r$-round algorithms (Theorem \ref{thm:vertex-main}).


However, in general, algorithms for maximal matching need not be algorithms with low vertex survival probability. We mitigate this issue by carefully selecting the hard instance. Specifically, we use a $\Delta$-regular high girth graph $G$ in which any set $I$ of at least $n\Delta^{-1/1000}$ vertices cannot be an independent set. Then, for $r = c\log\Delta$ for some constant $c < 1$, Theorem \ref{thm:vertex-main} shows that any $r$-round matching-certified algorithm in $G$ has vertex survival probability at least $\Delta^{-1/1000}$. Thus, the set of unmatched vertices $I$ has expected size $n\Delta^{-1/1000}$, and therefore, with high probability, there is an edge between two unmatched vertices and the found matching is not maximal. To summarize, we prove Theorem \ref{thm:main} by a two-step strategy:

\begin{enumerate}
    \item We show that any $c \log\Delta$-round algorithm in $G$ has vertex survival probability at least $\Delta^{-1/1000}$, for any $\Delta \le 2^{O(\sqrt{\log n})}$\footnote{The upper bound on $\Delta$ is required to ensure that $r = \Theta(\log\Delta)$ is less than the girth $\Theta(\log_{\Delta} n)$ of the graph.} (Section \ref{sec:VertexSurvivalProbability}, Theorem \ref{thm:vertex-main}).
    \item We argue that with such large $\Delta^{-1/1000}$ vertex survival probability, there is an edge between two unmatched nodes in $G$ (Section \ref{sec:ProofBody}).
\end{enumerate}

Our $\Omega\left(\min\left\{\log\Delta,\sqrt{\log n}\right\}\right)$ lower bound is implied by the fact that the $\Omega(\log\Delta)$ lower bound applies for all $\Delta\leq 2^{O(\sqrt{\log n})}$.
To obtain the result for $\Delta$-ary trees as well, we show that an algorithm for $\Delta$-ary tree would imply an algorithm for our $\Delta$-regular graph $G$ (The extension of the proof to trees is given in Section~\ref{sec:LBTrees}).

\paragraph{Randomness and edge neighborhoods:} LOCAL algorithms are usually given IDs and randomness on vertices. IDs can be generated by simply taking the first $10\log n$ bits of the randomness tape and reserving them as IDs. Each vertex will have a unique ID with high probability. Furthermore, we may associate randomness with edges rather than vertices by splitting the randomness tape into $\Delta$ chunks (reserve every $\Delta$ bits), one per adjacent edge. Moreover, while $r$-round LOCAL algorithms are usually given the $r$-neighborhood of a vertex; in our proofs, it is easier to work with the $r$-neighborhood of an edge. That is, we assume that an $r$-round algorithm is associated with an edge (rather than a vertex) and it receives as input the union of the two $r$-neighborhoods of the endpoints of that edge (which we also refer to as \emph{$r$-flower}, as we discuss shortly). Observe that the difference between the complexity of $r$-flower and $r$-neighborhood algorithms is a single round.

\subsection{Round Elimination via Self-Reduction} \label{sec:TOVSP}

We now give an overview of the main technical result of the paper (Lemma~\ref{lem:vertex-round-elim}). In this result, we show that the existence of an $r$-round matching-certified algorithm $f$ implies the existence of an $(r-1)$-round matching-certified algorithm $g$ with slightly higher vertex survival probability. We break up the discussion into three parts:

\begin{enumerate}
    \item Strategy: Near-determinism of $f$ (Section \ref{sec:TOStrategy}).
    \item Warmup: Obtaining an $\Omega(\log\log\Delta)$ lower bound (Section \ref{sec:TOWarmup}).
    \item Refinement: Obtaining an $\Omega(\min(\log\Delta, \sqrt{\log n}))$ lower bound (Section \ref{sec:TORefinement}).
\end{enumerate}



\subsubsection{Strategy: Determinism on $r$-Flowers and $r$-Neighborhoods}\label{sec:TOStrategy}

Consider an $r$-round matching-certified algorithm $f$ with vertex survival probability $P_f$. We consider the performance of $f$ in $\Delta$-ary trees (and  $\Delta$-regular high girth graphs, which locally look like $\Delta$-ary trees). Intuitively, for such graphs, we are able to argue that $f$ is nearly deterministic in several ways that together enable the construction of an $(r-1)$-round algorithm $g$ with vertex survival probability $O(P_f)$. We now discuss this intuition in more detail:

\paragraph{$\Delta$-ary tree structures: $r$-neighborhoods and $r$-flowers:} Recall that the algorithm has access to a tape of randomness on each edge, which we can think of as a numerical label selected uniformly at random from $[0,1]$. An $r$-\emph{neighborhood} of some vertex $A$ is a list of all edge labels present on edges within distance $r$ from $A$. An $r$-\emph{flower} of some edge $\{A,B\}$ is the union of the $r$-neighborhoods centered at $A$ and $B$. We depict this in Figure \ref{fig:unlabeled-rN-rF}.

\begin{figure}[H]
    \begin{tabular}{c | c}
    \begin{subfigure}{0.5\textwidth}
        \begin{tikzpicture}[
    dot/.style={circle, fill, inner sep=1.5pt}, scale=1.0 
]

\coordinate (c) at (0,0);

\coordinate (n1) at (2, 0);    
\coordinate (n2) at (-1.5, 1.5); 
\coordinate (n3) at (-1.5, -1.5);

\coordinate (n1_end1) at ($(n1) + (1.8, -1.0)$);  
\coordinate (n1_end2) at ($(n1) + (1.8, 1.0)$); 

\coordinate (n2_end1) at ($(n2) + (-0.8, 1.5)$);    
\coordinate (n2_end2) at ($(n2) + (-1.8, 0)$);   

\coordinate (n3_end1) at ($(n3) + (-1.8, 0)$);    
\coordinate (n3_end2) at ($(n3) + (-0.8, -1.5)$); 

\node[dot] at (c) {};
\node[dot] at (n1) {};
\node[dot] at (n2) {};
\node[dot] at (n3) {};
\node[dot] at (n1_end1) {};
\node[dot] at (n1_end2) {};
\node[dot] at (n2_end1) {};
\node[dot] at (n2_end2) {};
\node[dot] at (n3_end1) {};
\node[dot] at (n3_end2) {};

\draw (c) -- (n1) node[midway, below, sloped, inner sep=1pt] {};
\draw (c) -- (n2) node[midway, above, sloped, inner sep=1pt] {};
\draw (c) -- (n3) node[midway, below, sloped, inner sep=1pt] {};

\draw (n1) -- (n1_end1) node[midway, below, sloped, inner sep=1pt] {};
\draw (n1) -- (n1_end2) node[midway, above, sloped, inner sep=1pt] {};

\draw (n2) -- (n2_end1) node[midway, above, sloped, inner sep=1pt] {};
\draw (n2) -- (n2_end2) node[midway, below, sloped, inner sep=1pt] {};

\draw (n3) -- (n3_end1) node[midway, above, sloped, inner sep=1pt] {};
\draw (n3) -- (n3_end2) node[midway, below, sloped, inner sep=1pt] {};

\end{tikzpicture}
        \caption{The 2-neighborhood when $\Delta = 3$.}
    \end{subfigure}
    &
    \begin{subfigure}{0.5\textwidth}
        \begin{tikzpicture}[
    dot/.style={circle, fill, inner sep=1.5pt}, scale=1.0 
]

\coordinate (c_left) at (-0.9, 0);
\coordinate (c_right) at (0.9, 0);

\coordinate (n2) at ($(c_left) + (135:1.8)$); 
\coordinate (n3) at ($(c_left) + (225:1.8)$); 
\coordinate (n2_end1) at ($(n2) + (-0.8, 1.5)$); 
\coordinate (n2_end2) at ($(n2) + (-1.8, 0)$);  
\coordinate (n3_end1) at ($(n3) + (-1.8, 0)$);  
\coordinate (n3_end2) at ($(n3) + (-0.8, -1.5)$); 

\coordinate (n2r) at ($(c_right) + (45:1.8)$); 
\coordinate (n3r) at ($(c_right) + (315:1.8)$);
\coordinate (n2r_end1) at ($(n2r) + (0.8, 1.5)$); 
\coordinate (n2r_end2) at ($(n2r) + (1.8, 0)$);  
\coordinate (n3r_end1) at ($(n3r) + (1.8, 0)$);  
\coordinate (n3r_end2) at ($(n3r) + (0.8, -1.5)$);

\node[dot, label=below right:$A$] at (c_left) {};
\node[dot, label=below left:$B$] at (c_right) {};
\node[dot] at (n2) {};
\node[dot] at (n3) {};
\node[dot] at (n2_end1) {};
\node[dot] at (n2_end2) {};
\node[dot] at (n3_end1) {};
\node[dot] at (n3_end2) {};
\node[dot] at (n2r) {};
\node[dot] at (n3r) {};
\node[dot] at (n2r_end1) {};
\node[dot] at (n2r_end2) {};
\node[dot] at (n3r_end1) {};
\node[dot] at (n3r_end2) {};

\draw (c_left) -- (c_right) node[midway, above, inner sep=2pt] {}; 

\draw (c_left) -- (n2) node[midway, above, sloped, inner sep=1pt] {};
\draw (c_left) -- (n3) node[midway, below, sloped, inner sep=1pt] {};
\draw (n2) -- (n2_end1) node[midway, above, sloped, inner sep=1pt] {};
\draw (n2) -- (n2_end2) node[midway, below, sloped, inner sep=1pt] {};
\draw (n3) -- (n3_end1) node[midway, above, sloped, inner sep=1pt] {};
\draw (n3) -- (n3_end2) node[midway, below, sloped, inner sep=1pt] {};

\draw (c_right) -- (n2r) node[midway, above, sloped, inner sep=1pt] {}; 
\draw (c_right) -- (n3r) node[midway, below, sloped, inner sep=1pt] {}; 
\draw (n2r) -- (n2r_end1) node[midway, above, sloped, inner sep=1pt] {}; 
\draw (n2r) -- (n2r_end2) node[midway, below, sloped, inner sep=1pt] {}; 
\draw (n3r) -- (n3r_end1) node[midway, above, sloped, inner sep=1pt] {}; 
\draw (n3r) -- (n3r_end2) node[midway, below, sloped, inner sep=1pt] {}; 

\end{tikzpicture}
        \caption{The 2-flower centered at $\{A,B\}$ when $\Delta = 3$.}
    \end{subfigure}
    \end{tabular}
    
    \caption{2-flowers and 2-neighborhoods when $\Delta=3$ without edge labels. For edge-labeled versions, see Figure \ref{fig:rN-rF}.}
    \label{fig:unlabeled-rN-rF}
\end{figure}
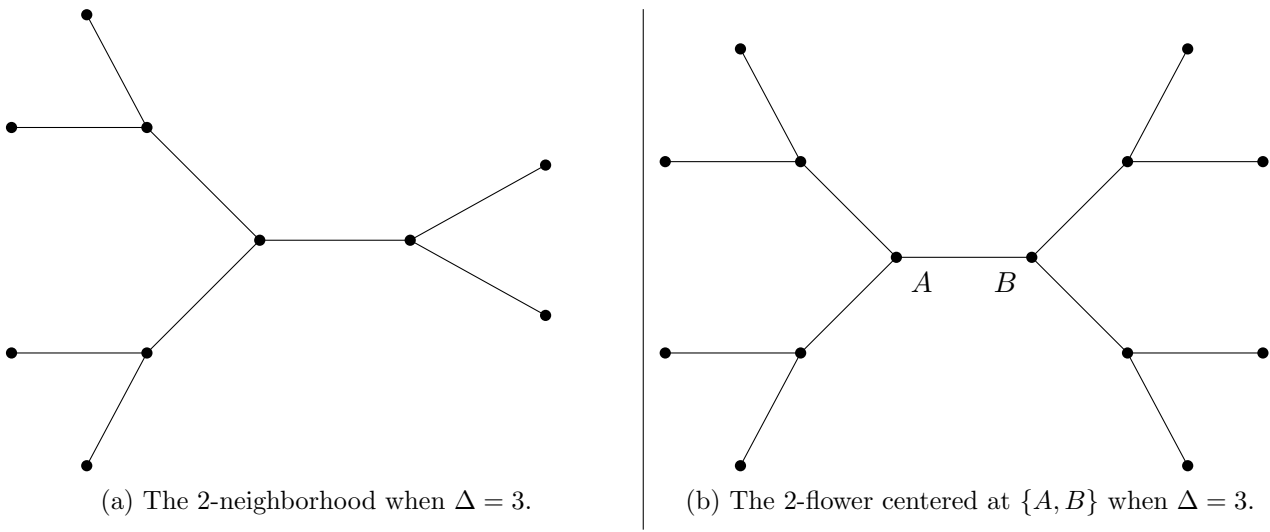



\paragraph{Determinism of $f$ on $r$-flowers:} As discussed earlier, we assume that an $r$-round algorithm $f$ takes as input an $r$-flower of an edge $\{A,B\}$, together with all the edge labels on the edges in that $r$-flower, and outputs $0$ or $1$ depending on whether it elects to add the edge $\{A,B\}$ to the matching. Since the edge labels already include the randomness tapes, the output of $f$ on $r$-flowers is deterministic. Interestingly, it is not very hard to see that $f$ is also deterministic on $r$-neighborhoods, as we discuss next.

\paragraph{Determinism of $f$ on $r$-neighborhoods:} Suppose that we only fix the labels in the $r$-neighborhood of some vertex $A$. This is not sufficient to deterministically specify the output of $f$, as $f$ receives $r$-flowers rather than $r$-neighborhoods as input. However, $f$ is deterministic on $r$-neighborhoods in another way: extensions to accepting $r$-flowers exist \emph{in only one direction}. An accepting $r$-flower is a flower for which the output of $f$ is 1. In other words, in a sense, only one direction is not deterministic. The statement is given formally in Proposition \ref{prop:r-ngbr-exclusive}.

In more detail, fix an $r$-neighborhood around a node $A$, let $B_1,B_2,\hdots,B_{\Delta}$ denote the neighbors of $A$, and consider the $\Delta$ different distributions of outputs obtained by applying $f$ to $r$-flowers centered at $\{A,B_i\}$ for various $i\in [\Delta]$. The following property holds: at most one $i$ exists for which there are $r$-flower extensions in direction $i$ on which $f$ outputs 1. This holds because if two such directions $i\neq j$ existed, the two directions could simultaneously cause $f$ to match $A$ with $B_i$ and with $B_j$, meaning that $f$ would not be matching-certified. We depict this reasoning in Figure \ref{fig:dir}. If such a value of $i$ exists for an $r$-neighborhood $x$ (i.e., if there is such an $i$ for which there are $r$-flower extensions to $x$ in direction $i$ on which $f$ outputs 1), we denote it by $\text{dir}(f,x)$. If $i$ does not exist, we define $\text{dir}(f,x)$ to be $0$.

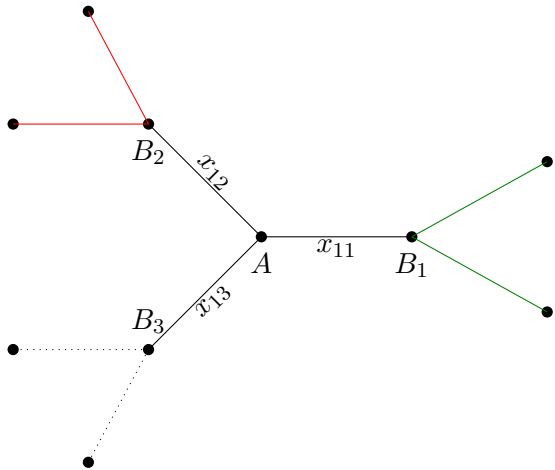
\begin{figure}[H]
    \centering
    \begin{tikzpicture}[
    dot/.style={circle, fill, inner sep=1.5pt}, scale=1.0 
]

\coordinate (c) at (0,0);

\coordinate (n1) at (2, 0);    
\coordinate (n2) at (-1.5, 1.5); 
\coordinate (n3) at (-1.5, -1.5);

\coordinate (n1_end1) at ($(n1) + (1.8, -1.0)$);  
\coordinate (n1_end2) at ($(n1) + (1.8, 1.0)$); 

\coordinate (n2_end1) at ($(n2) + (-0.8, 1.5)$);    
\coordinate (n2_end2) at ($(n2) + (-1.8, 0)$);   

\coordinate (n3_end1) at ($(n3) + (-1.8, 0)$);    
\coordinate (n3_end2) at ($(n3) + (-0.8, -1.5)$); 

\node[dot, label=below:$A$] at (c) {};
\node[dot, label=below:$B_1$] at (n1) {};
\node[dot, label=below:$B_2$] at (n2) {};
\node[dot, label=above:$B_3$] at (n3) {};
\node[dot] at (n1_end1) {};
\node[dot] at (n1_end2) {};
\node[dot] at (n2_end1) {};
\node[dot] at (n2_end2) {};
\node[dot] at (n3_end1) {};
\node[dot] at (n3_end2) {};

\draw (c) -- (n1) node[midway, below, sloped, inner sep=1pt] {$x_{11}$};
\draw (c) -- (n2) node[midway, above, sloped, inner sep=1pt] {$x_{12}$};
\draw (c) -- (n3) node[midway, below, sloped, inner sep=1pt] {$x_{13}$};

\draw[color=green!50!black] (n1) -- (n1_end1) node[midway, below, sloped, inner sep=1pt] {};
\draw[color=green!50!black] (n1) -- (n1_end2) node[midway, above, sloped, inner sep=1pt] {};

\draw[color=red] (n2) -- (n2_end1) node[midway, above, sloped, inner sep=1pt] {};
\draw[color=red] (n2) -- (n2_end2) node[midway, below, sloped, inner sep=1pt] {};

\draw[dotted] (n3) -- (n3_end1) node[midway, above, sloped, inner sep=1pt] {};
\draw[dotted] (n3) -- (n3_end2) node[midway, below, sloped, inner sep=1pt] {};

\end{tikzpicture}
    \caption{Depicting the scenario in which there are two different $\text{dir}(f,x)$-directions for the 1-neighborhood $x$ (for the sake of contradiction). $x$ consists of three labels $x_{11}$, $x_{12}$, and $x_{13}$ in directions 1, 2, and 3 respectively. Suppose that $f$ is a 1-round matching-certified algorithm, and suppose that 1 and 2 are valid possible values for $\text{dir}(f,x)$. This would imply that labels for the green edges exist for which $f$ applied to the 1-flower centered at $\{A,B_1\}$ returns 1. Separately, it would imply that labels for the red edges exist for which $f$ applied to the 1-flower centered at $\{A,B_2\}$ returns 1. In this case, the red and green edges do not overlap, so they can both be set. However, this would imply that both $\{A,B_1\}$ and $\{A,B_2\}$ are in the matching, contradicting the fact that $f$ is matching-certified.}
    \label{fig:dir}
\end{figure}

\paragraph{The behaviour of $f$ on $(r-1)$-flowers and $(r-1)$-neighborhoods:} $f$ is not deterministic given solely a labeling of an $(r-1)$-flower or $(r-1)$-neighborhood. However, robust analogues (in a probabilistic sense) of the above statements hold for $(r-1)$-flowers and $(r-1)$-neighborhoods, which enable us to prove Lemma \ref{lem:vertex-round-elim} (in which we show the round-elimination step). We discuss these robust analogues more in the next section; Section \ref{sec:TOWarmup}.

\subsubsection{Warmup: An $\Omega(\log\log\Delta)$-Round Lower Bound}\label{sec:TOWarmup}

Given an $r$-round matching-certified algorithm $f$ with vertex survival probability $P_f$, we now describe how to obtain an $(r-1)$-round matching-certified algorithm $g$ with $P_g \le O(\sqrt{P_f})$ as a warmup (as opposed to $O(P_f)$).\footnote{We do not give a formal version of this proof in this paper, as it is superseded by Lemma \ref{lem:vertex-round-elim}.} When integrated with the other components of our framework, this already implies an $\Omega(\log\log\Delta)$-round lower bound for MM, which we consider a good starting point towards understanding our $\Omega(\log\Delta)$ lower bound. We begin with the notion of \emph{$\delta$-good flowers}.

\paragraph{The set of $\delta$-good flowers $\mc X_{r-1}(f,\delta)$:} Without loss of generality, we may say that $f$ accepts an $r$-flower $w$ centered at an edge $\{A,B\}$ if and only if $\text{dir}(f,w_A)$ and $\text{dir}(f,w_B)$ point towards $w$, where $w_A$ and $w_B$ are the $r$-neighborhoods centered around $A$ and $B$ respectively.\footnote{Indeed, in the case where $\text{dir}(f,w_A)$ and $\text{dir}(f,w_B)$ point towards $w$, we can safely add the edge $\{A,B\}$ to the matching, as no neighboring edge can satisfy the same condition due to the uniqueness of $\text{dir}(f,x)$ for a neighborhood $x$.} Towards applying a single round elimination step, we are interested in the set of $(r-1)$-flowers that are part of many accepting $r$-flowers. Intuitively speaking, for such $(r-1)$-flowers, most of their $r$-neighborhood extensions $x$ from both sides have $\text{dir}(f,x)$ pointing towards $y$. More formally, let $\mc X_{r-1}(f,\delta)$ be the set of $(r-1)$-flowers $y$ for which at least a $(1-\delta)$-fraction of $r$-neighborhood extensions $x$ from the side of $A$ have $\text{dir}(f,x)$ pointing towards $B$ (i.e., in the direction of $y$), \emph{and} at least a $(1-\delta)$-fraction of $r$-neighborhood extensions $x$ from the side of $B$ have $\text{dir}(f,x)$ pointing towards $A$. Interestingly, we are able to show that $\mc X_{r-1}(f,\delta)$ has large probability mass, which already gives us a hint towards finding an $(r-1)$-round algorithm $g$.

\paragraph{$\mc X_{r-1}(f,\delta)$ has large probability mass (Proposition \ref{prop:r-1-flower-high-weak}):} We now discuss the proof of Proposition \ref{prop:r-1-flower-high-weak}, in which we show that $\mc X_{r-1}(f,\delta)$ has large probability mass. An $(r-1)$-flower $y$ centered at $\{A,B\}$ can be extended into an $r$-neighborhood on side $A$ and on side $B$. Let $P_A$ denote the probability that a random $r$-neighborhood extension $x$ on the $A$ side of $y$ has $\text{dir}(f,x)$ pointing towards $y$. Similarly, let $P_B$ the probability that a randomly sampled $r$-neighborhood extension $x$ on the $B$-side of $y$ has $\text{dir}(f,x)$ pointing towards $y$. We show that when we randomize over the $(r-1)$-flower $y$, the expected value of both $P_A$ and $P_B$ is $1/\Delta$ (Proposition \ref{prop:one-side}) and the expectation of the product $P_A\cdot P_B$ is $(1-P_f)/\Delta$ (Proposition \ref{prop:two-side}). Since $P_f$ is small, the expected product is close to the expectation of each value, meaning that a sufficiently large fraction of these probabilities (associated with different $(r-1)$-flowers $y$) must be close to 1. The $(r-1)$-flowers for which both $P_A$ and $P_B$ are close to 1 by definition lie in $\mc X_{r-1}(f,\delta)$. Quantitatively, we show that at least a $\frac{1-2P_f/\delta}{\Delta}$-fraction of $(r-1)$-flowers are good (see Proposition \ref{prop:r-1-flower-high-weak} for the formal statement).

\paragraph{Existence of a dominant direction for membership in $\mc X_{r-1}(f,\delta)$ (Proposition \ref{prop:r-1-neighborhood-low-weak}):} Since $\mc X_{r-1}(f,\delta)$ has large probability mass, we could try to define the algorithm $g$ as follows: $g(y) = 1$ if and only if $y\in \mc X_{r-1}(f,\delta)$ (for a sufficiently small value of $\delta$). However, it is possible for two incident $y$s to be in $\mc X_{r-1}(f,\delta)$ simultaneously. Thus, $g$ is not matching-certified. We prove an additional property about any $(r-1)$-neighborhood that allows us to transform $g$ into a matching-certified algorithm. In more detail, for an $(r-1)$-neighborhood $x$ and $i\in [\Delta]$, let $P_i(f,x,\delta)$ be the probability that the $(r-1)$-flower extension of $x$ in the $i$'th direction is $\delta$-good. We show that for any $(r-1)$-neighborhood $x$, there exists a direction $i^*\in [\Delta]$ for which $\sum_{i\ne i^*} P_i(f,x,\delta) = O(\delta)$ (Proposition \ref{prop:r-1-neighborhood-low-weak}). This is a robust analogue of the uniqueness of $\text{dir}$ that occurred for $r$-neighborhoods. We prove this proposition and a generalization of it (Lemma \ref{lem:r-1-neighborhood-low-common}) by contradiction: if no such direction exists, then we use a sampling argument to show that there is an $r$-neighborhood $x$ with two different directions $\text{dir}(f,x)$, contradicting the uniqueness of $\text{dir}(f,x)$. For more overview, see the beginning of Section \ref{subsec:r-1-neighborhood-low}. If there are two different possible values of $i^*$, we may set $i^* = 0$, as $\sum_{i=1}^{\Delta} P_i(f,x,\delta) = O(\delta)$ in this case and it is thus not necessary to ever match an $(r-1)$-flower incident with $x$ anyways. Thus, we may assume that $i^*$ exists and is unique.


\paragraph{Defining the algorithm $g$:} Now we are ready to define the algorithm $g$. For an $(r-1)$-flower $y$ centered at an edge $\{A,B\}$, $g(y) = 1$ if and only if (a) $y\in \mc X_{r-1}(f,\delta)$ (b) both of the $(r-1)$-neighborhoods in $y$ around $A$ and $B$ have their dominant directions $i^*$ pointing towards $y$.  The uniqueness of $i^*$ ensures that $g$ is matching-certified. 


\paragraph{Vertex survival probability analysis:} The existence of $i^*$ shows that the probability that $y\in \mc X_{r-1}(f,\delta)$ and that $i^*$ for $y_A$ does not point towards $y$ is at most $\frac{1}{\Delta}\sum_{i\ne i^*} P_i(f,y_A,\delta) \le O(\delta)/\Delta$, where $y_A$ is the $r$-neighborhood centered at $A$. The probability that $y\in \mc X_{r-1}(f,\delta)$ is at least $(1/\Delta)(1 - 2P_f/\delta)$ as discussed earlier. Thus, $$\Pr[g(y) = 1] \ge \frac{1}{\Delta}(1 - 2P_f/\delta - 2O(\delta)) > (1/\Delta)(1 - O(\sqrt{P_f}))$$ where the last inequality holds for an optimal choice of $\delta\approx \sqrt{P_f}$. This translates into a vertex survival probability bound of $O(\sqrt{P_f})$ for $g$, as a vertex survives if and only if $g(y) = 0$ for all $\Delta$ incident $(r-1)$-flowers.

\subsubsection{Refinement: An $\Omega\left(\min\left\{\log\Delta, \sqrt{\log n}\right\}\right)$-Round Lower Bound}\label{sec:TORefinement}

We now describe how to obtain an $(r-1)$-round matching-certified algorithm $g$ with $P_g \le O(P_f)$. Combining this with the other components in our proof, this implies the desired lower bound for MM (Theorem~\ref{thm:main}). We do this by considering multiple values of $\delta$, as follows. Recall that in the previous section we defined $g(y) = 1$ only for $\delta$-good $(r-1)$-flowers $y$, for an optimized choice of $\delta \approx \sqrt{P_f}$. In designing $g$, we are effectively trying to approximate $f$ given just the $(r-1)$-flower. Using one value of $\delta$ results in a crude approximation, hence the large loss in vertex survival probability found in the previous section. Instead of just considering good flowers for a single value of $\delta$, we instead consider good flowers for all values of $\delta \le C$ for some sufficiently small constant $C\in [0,1]$.

\paragraph{Key improved lemma (Lemma \ref{lem:r-1-neighborhood-low-strong}) and defining dominant neighborhoods $\mc R_{r-1}(f,\delta)$:} We show an improvement to Proposition \ref{prop:r-1-neighborhood-low-weak} that uncovers the existence of interaction between good flower sets for different $\delta$s. Let $\mc R_{r-1}(f,\delta)$ denote the set of $(r-1)$-neighborhoods $x$ with a dominant $\delta$-direction, where a dominant $\delta$-direction is a direction $i^*$ for which $P_{i^*}(f,x,\delta) \ge \hat{C}$ for some large constant $\hat{C}$. We show in Lemma \ref{lem:r-1-neighborhood-low-strong} that for any $x\in \mc R_{r-1}(f,\delta)$, we have $\sum_{i\ne i^*} P_i(f,x,C) \le O(\delta)$.\footnote{This is equivalent to the statement of Lemma \ref{lem:r-1-neighborhood-low-strong}, which uses $\delta_{\text{dom}}$.} This is much stronger than the conclusion obtained by Proposition \ref{prop:r-1-neighborhood-low-weak}, which only reasons about $P_i(f,x,\delta)$ rather than $P_i(f,x,C)$. 

\paragraph{Showing that most neighborhoods are in $\mc R_{r-1}(f,\delta)$ for very small $\delta$ (Proposition \ref{prop:complement-prob}):} For an $r$-neighborhood $x$, we define $\delta_{\text{dom}}(f,x)$ to be the smallest value of $\delta$ for which $x$ has a dominant $\delta$-direction; i.e. $x\in \mc R_{r-1}(f,\delta)$. We show that a uniformly random $(r-1)$-neighborhood $x$ has $\delta_{\text{dom}}(f,x) = O(P_f)$ in expectation (Proposition \ref{prop:complement-prob}). A more detailed overview of the proof of this proposition is given in Section \ref{sec:complement-prob}. 

\paragraph{Improved version of $g$:} We use the exact same algorithm as before, but with setting $\delta=C$ for a specific choice of a constant $C\in [0,1]$. Specifically, for an $(r-1)$-flower $y$, we define $g(y) = 1$ if and only if (a) $y$ is $C$-good, (b) both of the $(r-1)$-neighborhoods in $y$ around $A$ and $B$ have their dominant direction $i^*$ pointing towards $y$, where $i^*$ is the direction $i\in [\Delta]$ that maximizes $P_i(f,x,C)$ or 0 if the maximizer is not unique.

\paragraph{Improved vertex survival probability analysis with a  wishful thinking assumption:} Temporarily assume that $\delta_{\text{dom}}(f,x) = O(P_f)$ for \emph{all} $x\in \mc R_{r-1}$, not just in expectation. One can sample $y\in \mc F_{r-1}$ by sampling $y_A\in \mc R_{r-1}$, picking a random direction $i\in [\Delta]$, and sampling an extension of $y_A$ in direction $i$ (Proposition \ref{prop:r-to-f-sample}). This sampling process shows that the probability that $y$ is $C$-good and $i^*$ for $y_A$ does not point towards $y$ is at most\footnote{This probability bound is formally given in the final proof of Lemma \ref{lem:vertex-round-elim} (Section \ref{sec:prooflemmaRE}).}

$$\frac{1}{\Delta}\sum_{i\ne i^*} P_i(f,y_A,C) \le O(\delta_{\text{dom}}(f,y_A))/\Delta \le O(P_f)/\Delta$$

where the first inequality comes from applying both Lemma \ref{lem:r-1-neighborhood-low-strong} and the second comes from the wishful thinking. The probability that $y$ is $C$-good is at least $(1/\Delta)(1 - 2P_f/C)$ (as discussed earlier, this is shown in Proposition \ref{prop:r-1-flower-high-weak}). Thus, by a union bound, $$\Pr[g(y) = 1] \ge \frac{1}{\Delta}\left(1 - 2P_f/C - 2O(P_f)\right) = \frac{1}{\Delta}\left(1 - O(P_f)\right)$$ 
This translates into a vertex survival probability bound of $O(P_f)$ for $g$, as desired.

\paragraph{Overcoming the wishful thinking assumption:} Applying Markov's Inequality in conjunction with Proposition \ref{prop:complement-prob} shows that all but a small constant fraction of $(r-1)$-neighborhoods $x$ have $\delta_{\text{dom}}(f,x) = O(P_f)$.\footnote{In fact, this can be proven directly without going through Proposition \ref{prop:complement-prob}. We do this as a warmup to Proposition \ref{prop:complement-prob} (Proposition \ref{prop:weak-complement-prob}).} Unfortunately, the vertex survival probability increase caused by ignoring this small constant fraction of neighborhoods is substantial. A naive bound\footnote{This can be obtained using Proposition \ref{prop:crude-distrib}.} for the constant fraction $p$ of $x\in \mc R_{r-1}$s with $\delta_{\text{dom}}(f,x) \gg P_f$  shows that the probability that $y$ is $C$-good and $i^*$ for $y_A$ does not point towards $y$ is at most $\frac{1}{\Delta}\sum_{i\ne i^*} P_i(f,y_A,C) \le 2/\Delta$. Plugging this bound in along with the better bound for the other $(1-p)$-fraction of neighborhoods would lead to

$$\Pr[g(y) = 1] \ge \frac{1}{\Delta}(1 - 2P_f/C - 2(1-p)O(P_f) - 4p) = \frac{1}{\Delta}(1 - O(1))$$

which would only lead to a vertex survival probability bound of $O(1)$ since $p$ is constant. In particular, accounting for the fact that the high $\delta_{\text{dom}}$ $y_A$-s only occur with constant probability does not meaningfully improve the analysis.

\paragraph{Final vertex survival probability analysis:} We rectify this issue by simply exploiting the fact that the sampling process in the previous analysis only needs a bound on the expected value of $\delta_{\text{dom}}(f,y_A)$, not the maximum value. In particular, the probability that $y$ is $C$-good and $i^*$ for $y_A$ does not point towards $y$ is at most

$$\E\left[\frac{1}{\Delta}\sum_{i\ne i^*} P_i(f,y_A,C)\right] \le O\left(\E[\delta_{\text{dom}}(f,y_A)]\right)/\Delta \le O(P_f)/\Delta$$

Plugging this into the previous analysis yields $P_g \le O(P_f)$ as desired. Intuitively, some $y_A$s may lead to higher-than-desired probability that $y$ is $C$-good and $i^*$ points away from $y$, but that high probability is counterbalanced by the low fraction of these $y_A$s. It is also worth noting that unlike the $\Omega(\log\log\Delta)$-round analysis from the previous section, which used Proposition \ref{prop:r-1-neighborhood-low-weak} for the single value of $\delta = \sqrt{P_f}$, the new analysis uses the improved version (Lemma \ref{lem:r-1-neighborhood-low-strong}), which implictly works with a \emph{continuous range} of $\delta$s (for each possible value of $\delta_{\text{dom}}(f,y_A)$).

\section{Preliminaries}\label{sec:Prelims}

In this section we provide some preliminaries and formal definitions. In this work, we show hardness results for matching on graphs that locally look like $\Delta$-ary trees for some positive integer $\Delta$. An $r$-round LOCAL algorithm for (any variant of) matching takes the $r$-hop neighborhood of a node and all the random tapes associated with all the vertices in this neighborhood and outputs for each incident edge 1 or 0 depending on whether it joins the matching. As discussed earlier, there is a simple transformation that allows us to assume that the randomness tapes are associated with edges instead of vertices. Furthermore, each randomness tape can be thought of as a single real number selected uniformly at random between $0$ and $1$, sampled independently from the others. In this section, we introduce notation for describing and manipulating $\Delta$-ary trees that are edge-labeled with these numbers. We give our definitions over the course of four subsections, divided as follows:

\begin{enumerate}
    \item Description: defining $r$-neighborhoods and $r$-flowers (Section \ref{sec:PrelimDescription})
    \item Manipulation: transforming neighborhoods into flowers and vice-versa (Section \ref{sec:PrelimManipulation})
    \item Algorithms: processing $r$-flowers (Section \ref{sec:PrelimAlgos})
    \item Distributions: (conditional) sampling from distributions over flowers and neighborhoods (Section \ref{sec:PrelimDistributions})
\end{enumerate}

We have an additional subsection (Section \ref{sec:PrelimOther}) for theorems unrelated to $\Delta$-ary trees that we use in this work.

\subsection{Description}\label{sec:PrelimDescription}

We define two key concepts related to $\Delta$-ary trees: $r$-\emph{neighborhoods} and $r$-\emph{flowers}, where $r$-flowers consist of an edge and all edges within distance $r$ of that edge's endpoints. We are interested in the sets of possible labelings of these subgraphs, which we denote by $\mc R_r$ and $\mc F_r$ respectively.

\paragraph{Some Notation:} We consider $\Delta$-ary trees in this work, for some integer $\Delta\ge 2$.\footnote{$\Delta = 1$ only allows for connected components consisting of a single isolated edge. $\Delta = 2$ leads to the simplest arbitrarily large graph (a path).} All of the $r$-flowers in this work for various values of $r$ are centered around the edge $\{A,B\}$. For a positive integer $r$, let $[r] := \{1,2,\hdots,r\}$ and $[[r]] := \{0\}\cup [r] = \{0,1,\hdots,\Delta\}$.

\begin{defn}[$r$-Flowers $\mathcal{F}_r$]\label{def:rF}
    Let $\mc S_0 := [0,1]$ and for each $r\ge 1$, let $\mc S_r := \mc S_{r-1}^{\Delta-1}$.\footnote{Powers of sets represent Cartesian products, i.e. $\mc S_{r-1}\times\mc  S_{r-1}\times \hdots\times\mc S_{r-1}$ where $\mc S_{r-1}$ appears $\Delta-1$ times.} For $r\ge 0$, let $\mc F_r := \mc S_0\times \mc S_1^2\times \mc S_2^2\times \hdots\times \mc S_{r-1}^2\times \mc S_r^2$. Index members of $w\in \mc F_r$ as a matrix, where $w_{sv}$ is the $v$th copy of $\mc S_s$ for $s\in [r]$ and $v\in \{A,B\}$, and $w_0$ the unique element in $\mc S_0$.
\end{defn}

To clarify, when we say that $w$ is a matrix, it is not a matrix of real numbers, but instead a list of lists of objects, where the 0th row consists of just one element ($w_0$) and the $s$th row for $s \ge 0$ consists of a column indexed by the symbols $\{A,B\}$ rather than integers. Furthermore, the entries in the $s$th row are members of $\mc S_s$ rather than real numbers. Now, we introduce neighborhoods:

\begin{defn}[$r$-Neighborhoods $\mathcal{R}_r$]\label{def:rN}
    For $r\ge 1$, let $\mc R_r := \mc S_0^{\Delta}\times \mc S_1^{\Delta}\times\hdots\times \mc S_{r-1}^{\Delta}$, with $z\in \mc R_r$ indexed as a matrix, with $z_{si}$ in the $i$th copy of $\mc S_{s-1}$\footnote{The discrepancy between $s$ and $s-1$ may be confusing here, but we chose this to align with the $r$-flower definition, specifically the fact that both $z_{si}$ and $w_{si}$ are labels on edges that are from distance $s-1$ to distance $s$ from the center or the edge $\{A,B\}$.} for $s\in [r]$ and $i\in [\Delta]$. For $r = 0$, define the trivial empty neighborhood $\emptyset$ and define $\mc R_0 := \{\emptyset\}$.
\end{defn}

In this definition, the matrix is integer-indexed, but the values are vectors of different dimensions (i.e. from different $\mc S_s$s). We now depict how the reader should think of these labels as being arranged on the $\Delta$-ary tree (Figure \ref{fig:rN-rF}):

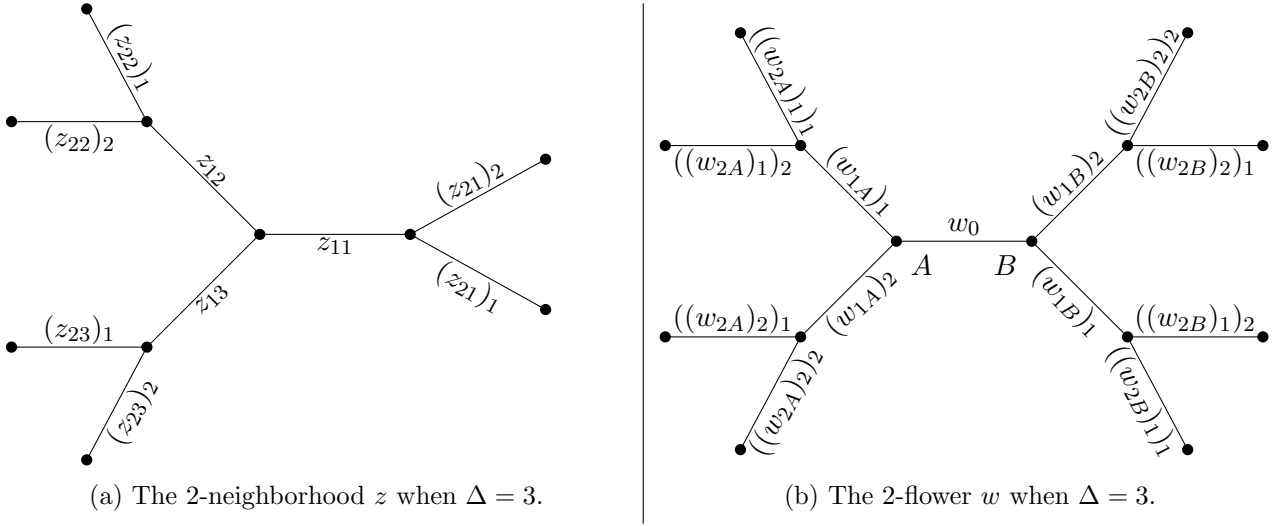
\begin{figure}[H]
    \begin{tabular}{c | c}
    \begin{subfigure}{0.5\textwidth}
        \begin{tikzpicture}[
    dot/.style={circle, fill, inner sep=1.5pt}, scale=1.0 
]

\coordinate (c) at (0,0);

\coordinate (n1) at (2, 0);    
\coordinate (n2) at (-1.5, 1.5); 
\coordinate (n3) at (-1.5, -1.5);

\coordinate (n1_end1) at ($(n1) + (1.8, -1.0)$);  
\coordinate (n1_end2) at ($(n1) + (1.8, 1.0)$); 

\coordinate (n2_end1) at ($(n2) + (-0.8, 1.5)$);    
\coordinate (n2_end2) at ($(n2) + (-1.8, 0)$);   

\coordinate (n3_end1) at ($(n3) + (-1.8, 0)$);    
\coordinate (n3_end2) at ($(n3) + (-0.8, -1.5)$); 

\node[dot] at (c) {};
\node[dot] at (n1) {};
\node[dot] at (n2) {};
\node[dot] at (n3) {};
\node[dot] at (n1_end1) {};
\node[dot] at (n1_end2) {};
\node[dot] at (n2_end1) {};
\node[dot] at (n2_end2) {};
\node[dot] at (n3_end1) {};
\node[dot] at (n3_end2) {};

\draw (c) -- (n1) node[midway, below, sloped, inner sep=1pt] {$z_{11}$};
\draw (c) -- (n2) node[midway, above, sloped, inner sep=1pt] {$z_{12}$};
\draw (c) -- (n3) node[midway, below, sloped, inner sep=1pt] {$z_{13}$};

\draw (n1) -- (n1_end1) node[midway, below, sloped, inner sep=1pt] {$(z_{21})_1$};
\draw (n1) -- (n1_end2) node[midway, above, sloped, inner sep=1pt] {$(z_{21})_2$};

\draw (n2) -- (n2_end1) node[midway, above, sloped, inner sep=1pt] {$(z_{22})_1$};
\draw (n2) -- (n2_end2) node[midway, below, sloped, inner sep=1pt] {$(z_{22})_2$};

\draw (n3) -- (n3_end1) node[midway, above, sloped, inner sep=1pt] {$(z_{23})_1$};
\draw (n3) -- (n3_end2) node[midway, below, sloped, inner sep=1pt] {$(z_{23})_2$};

\end{tikzpicture}
        \caption{The 2-neighborhood $z$ when $\Delta = 3$.}
    \end{subfigure}
    &
    \begin{subfigure}{0.5\textwidth}
        \begin{tikzpicture}[
    dot/.style={circle, fill, inner sep=1.5pt}, scale=1.0 
]

\coordinate (c_left) at (-0.9, 0);
\coordinate (c_right) at (0.9, 0);

\coordinate (n2) at ($(c_left) + (135:1.8)$); 
\coordinate (n3) at ($(c_left) + (225:1.8)$); 
\coordinate (n2_end1) at ($(n2) + (-0.8, 1.5)$); 
\coordinate (n2_end2) at ($(n2) + (-1.8, 0)$);  
\coordinate (n3_end1) at ($(n3) + (-1.8, 0)$);  
\coordinate (n3_end2) at ($(n3) + (-0.8, -1.5)$); 

\coordinate (n2r) at ($(c_right) + (45:1.8)$); 
\coordinate (n3r) at ($(c_right) + (315:1.8)$);
\coordinate (n2r_end1) at ($(n2r) + (0.8, 1.5)$); 
\coordinate (n2r_end2) at ($(n2r) + (1.8, 0)$);  
\coordinate (n3r_end1) at ($(n3r) + (1.8, 0)$);  
\coordinate (n3r_end2) at ($(n3r) + (0.8, -1.5)$);

\node[dot, label=below right:$A$] at (c_left) {};
\node[dot, label=below left:$B$] at (c_right) {};
\node[dot] at (n2) {};
\node[dot] at (n3) {};
\node[dot] at (n2_end1) {};
\node[dot] at (n2_end2) {};
\node[dot] at (n3_end1) {};
\node[dot] at (n3_end2) {};
\node[dot] at (n2r) {};
\node[dot] at (n3r) {};
\node[dot] at (n2r_end1) {};
\node[dot] at (n2r_end2) {};
\node[dot] at (n3r_end1) {};
\node[dot] at (n3r_end2) {};

\draw (c_left) -- (c_right) node[midway, above, inner sep=2pt] {$w_{0}$}; 

\draw (c_left) -- (n2) node[midway, above, sloped, inner sep=1pt] {$(w_{1A})_1$};
\draw (c_left) -- (n3) node[midway, below, sloped, inner sep=1pt] {$(w_{1A})_2$};
\draw (n2) -- (n2_end1) node[midway, above, sloped, inner sep=1pt] {$((w_{2A})_1)_1$};
\draw (n2) -- (n2_end2) node[midway, below, sloped, inner sep=1pt] {$((w_{2A})_1)_2$};
\draw (n3) -- (n3_end1) node[midway, above, sloped, inner sep=1pt] {$((w_{2A})_2)_1$};
\draw (n3) -- (n3_end2) node[midway, below, sloped, inner sep=1pt] {$((w_{2A})_2)_2$};

\draw (c_right) -- (n2r) node[midway, above, sloped, inner sep=1pt] {$(w_{1B})_2$}; 
\draw (c_right) -- (n3r) node[midway, below, sloped, inner sep=1pt] {$(w_{1B})_1$}; 
\draw (n2r) -- (n2r_end1) node[midway, above, sloped, inner sep=1pt] {$((w_{2B})_2)_2$}; 
\draw (n2r) -- (n2r_end2) node[midway, below, sloped, inner sep=1pt] {$((w_{2B})_2)_1$}; 
\draw (n3r) -- (n3r_end1) node[midway, above, sloped, inner sep=1pt] {$((w_{2B})_1)_2$}; 
\draw (n3r) -- (n3r_end2) node[midway, below, sloped, inner sep=1pt] {$((w_{2B})_1)_1$}; 

\end{tikzpicture}
        \caption{The 2-flower $w$ when $\Delta = 3$.}
    \end{subfigure}
    \end{tabular}
    
    \caption{The assignment of scalars within $z$ and $w$ to edges in a $\Delta$-ary tree. The nested indices are never actually used in our proofs -- we manipulate flowers and neighborhoods entirely using the functions $\text{end}_v$ and $\text{res}_i$ described later in this section. We state the properties of $\text{end}_v$ and $\text{res}_i$ that we require in Section \ref{sec:VertexSurvivalProbability} as propositions and lemmas in Section \ref{sec:PrelimDistributions}.}
    \label{fig:rN-rF}
\end{figure}




\subsection{Manipulation}\label{sec:PrelimManipulation}

We could have instead defined $r$-neighborhoods and $r$-flowers as labelings of certain graphs. We chose not to do this in order to more easily facilitate transformations between labelings of $r$-flowers and $r$-neighborhoods. These transformations are critical to obtaining a proof of our main technical result (Theorem \ref{thm:vertex-main}). First, we define $\sigma$-shuffles, which intuitively speaking are transformations on $r$-neighborhoods that permute the $\Delta$ different directions. Observe that in the following definition we are slightly abusing notation by not just letting $\sigma$ be a function over integers $[k]$, but also letting it be a function over $R_r$ and $[[k]]$:

\begin{defn}[$\sigma$-Shuffles]
For a positive integer $k$, let $\Pi_k$ denote the set of permutations on $[k]$; that is the set of bijective functions $\sigma:[k]\rightarrow[k]$. For any $\sigma\in \Pi_k$, define $\sigma(0) := 0$ (recall that $0\notin [k]$ so has not already been assigned). For a permutation $\sigma \in \Pi_{\Delta}$, $r\ge 0$, and $z\in \mc R_r$, define $\sigma(z)\in \mc R_r$ as follows: $(\sigma(z))_{s\sigma(i)} := z_{si}$ for all $s\in [r]$, $i\in [\Delta]$. 
\end{defn}

We often need to work with the specific permutation $\sigma_i\in \Pi_{\Delta}$ that switches 1 and $i$ for some direction $i\in [\Delta]$:

\begin{defn}[The special permutations $\sigma_i$]
For each $i\in [\Delta]$, let $\sigma_i\in \Pi_{\Delta}$ denote the permutation that does the following: $\sigma_i(1) = i$, $\sigma_i(i) = 1$, and $\sigma_i(j) = j$ for all other $j\in [\Delta]$.
\end{defn}

Next, we discuss the function $\text{end}_v(w)$. Intuitively speaking, $\text{end}_v(w)$ takes an $r$-flower $w$ centered around the edge $\{A,B\}$ and $v\in \{A,B\}$ and outputs the $r$-neighborhood centered around $v$, with $w$ being an extension of $\text{end}_v(w)$ in the 1-direction of $v$. We depict $\text{end}_A(w)$ in Figure \ref{fig:endpoint-of-flower}.

\begin{defn}[Endpoint $end_v$ of an $r$-Flower]\label{def:endpoint-of-flower}
    For $r\ge 1$, $v\in \{A,B\}$, and $w\in \mc F_r$, define the $v$-\emph{endpoint} $\text{end}_v(w) \in \mc R_r$ as follows:

\begin{enumerate}
    \item $(\text{end}_v(w))_{11} := w_0$
    \item $(\text{end}_v(w))_{s1} := w_{(s-1)u}$ for $s\in \{2,3,\hdots,r\}$, where $u$ is the unique member of $\{A,B\}\setminus \{v\}$
    \item $(\text{end}_v(w))_{s(i+1)} := (w_{sv})_i$ for any $i\in [\Delta-1]$ and $s\in [r]$
\end{enumerate}

For $r = 0$, define $\text{end}_v(w) := \emptyset$.
\end{defn}

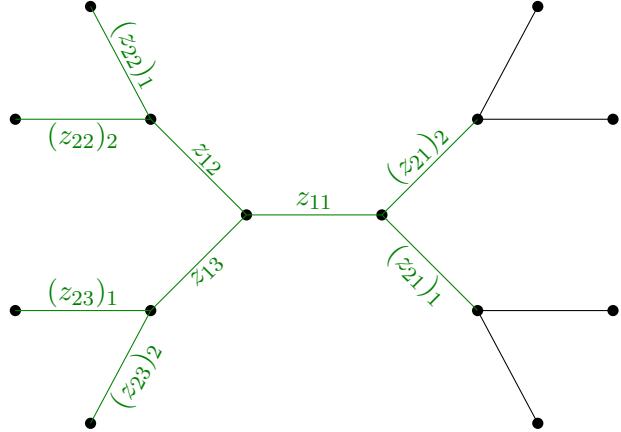
\begin{figure}[H]
    \centering
    \begin{tikzpicture}[
    dot/.style={circle, fill, inner sep=1.5pt}, scale=1.0 
]

\coordinate (c_left) at (-0.9, 0);
\coordinate (c_right) at (0.9, 0);

\coordinate (n2) at ($(c_left) + (135:1.8)$); 
\coordinate (n3) at ($(c_left) + (225:1.8)$); 
\coordinate (n2_end1) at ($(n2) + (-0.8, 1.5)$); 
\coordinate (n2_end2) at ($(n2) + (-1.8, 0)$);  
\coordinate (n3_end1) at ($(n3) + (-1.8, 0)$);  
\coordinate (n3_end2) at ($(n3) + (-0.8, -1.5)$); 

\coordinate (n2r) at ($(c_right) + (45:1.8)$); 
\coordinate (n3r) at ($(c_right) + (315:1.8)$);
\coordinate (n2r_end1) at ($(n2r) + (0.8, 1.5)$); 
\coordinate (n2r_end2) at ($(n2r) + (1.8, 0)$);  
\coordinate (n3r_end1) at ($(n3r) + (1.8, 0)$);  
\coordinate (n3r_end2) at ($(n3r) + (0.8, -1.5)$);

\node[dot] at (c_left) {};
\node[dot] at (c_right) {};
\node[dot] at (n2) {};
\node[dot] at (n3) {};
\node[dot] at (n2_end1) {};
\node[dot] at (n2_end2) {};
\node[dot] at (n3_end1) {};
\node[dot] at (n3_end2) {};
\node[dot] at (n2r) {};
\node[dot] at (n3r) {};
\node[dot] at (n2r_end1) {};
\node[dot] at (n2r_end2) {};
\node[dot] at (n3r_end1) {};
\node[dot] at (n3r_end2) {};

\draw[color=green!50!black] (c_left) -- (c_right) node[midway, above, inner sep=2pt] {$z_{11}$}; 

\draw[color=green!50!black] (c_left) -- (n2) node[midway, above, sloped, inner sep=1pt] {$z_{12}$};
\draw[color=green!50!black] (c_left) -- (n3) node[midway, below, sloped, inner sep=1pt] {$z_{13}$};
\draw[color=green!50!black] (n2) -- (n2_end1) node[midway, above, sloped, inner sep=1pt] {$(z_{22})_1$};
\draw[color=green!50!black] (n2) -- (n2_end2) node[midway, below, sloped, inner sep=1pt] {$(z_{22})_2$};
\draw[color=green!50!black] (n3) -- (n3_end1) node[midway, above, sloped, inner sep=1pt] {$(z_{23})_1$};
\draw[color=green!50!black] (n3) -- (n3_end2) node[midway, below, sloped, inner sep=1pt] {$(z_{23})_2$};

\draw[color=green!50!black] (c_right) -- (n2r) node[midway, above, sloped, inner sep=1pt] {$(z_{21})_2$};
\draw[color=green!50!black] (c_right) -- (n3r) node[midway, below, sloped, inner sep=1pt] {$(z_{21})_1$};
\draw (n2r) -- (n2r_end1) node[midway, above, sloped, inner sep=1pt] {};
\draw (n2r) -- (n2r_end2) node[midway, below, sloped, inner sep=1pt] {};
\draw (n3r) -- (n3r_end1) node[midway, above, sloped, inner sep=1pt] {};
\draw (n3r) -- (n3r_end2) node[midway, below, sloped, inner sep=1pt] {};

\end{tikzpicture}
    \caption{The 2-neighborhood $z = \text{end}_A(w)$ with $\Delta = 3$ and the 2-flower $w$ given in Figure \ref{fig:rN-rF}. The $\{A,B\}$ edge became the edge pointing in the 1 direction.}
    \label{fig:endpoint-of-flower}
\end{figure}

It is worth noting that conceptually speaking, the set of $r$-flowers $y\in \mc F_r$ that extend the $r$-neighborhood $x\in \mc R_r$ in direction $i\in [\Delta]$ can be formally expressed as the set of $y\in \mc F_r$ for which $\text{end}_A(y) = \sigma_i(x)$. Intuitively, this is because $y$ extends the 1-direction of $\text{end}_A(y)$, which is in turn the $i$-direction of $x$.

Next, we introduce the function $\text{res}_i(z)$ that, intuitively speaking, takes an $r$-neighborhood $z$ and $i\in [\Delta]$ and outputs the $r$-flower in the direction $i$, with the center vertex of $z$ mapped to $A$. We depict $\text{res}_i(z)$ in Figure \ref{fig:restriction-of-neighborhood}.

\begin{defn}[Restrictions $\text{res}_i$ of $r$-Neighborhoods to $(r-1)$-Flowers]\label{def:rest-rN-rF}

For $r\ge 1$ and $z\in \mc R_r$, define the $1$-\emph{restriction} $\text{res}_1(z) \in \mc F_{r-1}$ as follows:

\begin{enumerate}
    \item $(\text{res}_1(z))_0 := z_{11}$
    \item $(\text{res}_1(z))_{sB} := z_{(s+1)1}$ for all $s\in [r-1]$
    \item $((\text{res}_1(z))_{sA})_i := z_{s(i+1)}$ for all $s\in [r-1]$, $i\in [\Delta-1]$
\end{enumerate}

For $i\in [\Delta]$ and $z\in \mc R_r$, define $\text{res}_i(z) := \text{res}_1(\sigma_i(z))$.
\end{defn}

\begin{figure}[H]
    \centering
    \begin{tikzpicture}[
    dot/.style={circle, fill, inner sep=1.5pt}, scale=1.0 
]

\coordinate (c) at (0,0);

\coordinate (n1) at (2, 0);    
\coordinate (n2) at (-1.5, 1.5); 
\coordinate (n3) at (-1.5, -1.5);

\coordinate (n1_end1) at ($(n1) + (1.8, -1.0)$);  
\coordinate (n1_end2) at ($(n1) + (1.8, 1.0)$); 

\coordinate (n2_end1) at ($(n2) + (-0.8, 1.5)$);    
\coordinate (n2_end2) at ($(n2) + (-1.8, 0)$);   

\coordinate (n3_end1) at ($(n3) + (-1.8, 0)$);    
\coordinate (n3_end2) at ($(n3) + (-0.8, -1.5)$); 

\node[dot, label=left:$A$] at (c) {};
\node[dot] at (n1) {};
\node[dot, label=below:$B$] at (n2) {};
\node[dot] at (n3) {};
\node[dot] at (n1_end1) {};
\node[dot] at (n1_end2) {};
\node[dot] at (n2_end1) {};
\node[dot] at (n2_end2) {};
\node[dot] at (n3_end1) {};
\node[dot] at (n3_end2) {};

\draw[color=green!50!black] (c) -- (n1) node[midway, below, sloped, inner sep=1pt] {$(y_{1A})_1$};
\draw[color=green!50!black] (c) -- (n2) node[midway, above, sloped, inner sep=1pt] {$y_0$};
\draw[color=green!50!black] (c) -- (n3) node[midway, below, sloped, inner sep=1pt] {$(y_{1A})_2$};

\draw (n1) -- (n1_end1) node[midway, below, sloped, inner sep=1pt] {};
\draw (n1) -- (n1_end2) node[midway, above, sloped, inner sep=1pt] {};

\draw[color=green!50!black] (n2) -- (n2_end1) node[midway, above, sloped, inner sep=1pt] {$(y_{1B})_1$};
\draw[color=green!50!black] (n2) -- (n2_end2) node[midway, below, sloped, inner sep=1pt] {$(y_{1B})_2$};

\draw (n3) -- (n3_end1) node[midway, above, sloped, inner sep=1pt] {};
\draw (n3) -- (n3_end2) node[midway, below, sloped, inner sep=1pt] {};

\end{tikzpicture}
    \caption{The 1-flower $y = \text{res}_2(z)$ with $\Delta = 3$ and the 2-neighborhood $z$ given in Figure \ref{fig:rN-rF}.}
    \label{fig:restriction-of-neighborhood}
\end{figure}
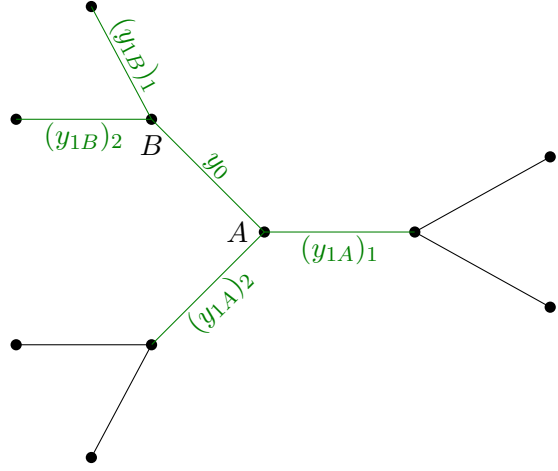

Next, we define reversals of flowers, obtained by simply reflecting the labels about the center edge $\{A,B\}$:
    
\begin{defn}[Reversals $\ol{w}$ of $r$-Flowers]\label{def:rev-rF}
For $r\ge 0$ and $w\in \mc F_r$, define the \emph{reversal} $\ol{w} \in \mc F_r$ as follows:

\begin{enumerate}
    \item $\ol{w}_0 := w_0$
    \item For any $s\in [r]$, $\ol{w}_{sA} := w_{sB}$
    \item For any $s\in [r]$, $\ol{w}_{sB} := w_{sA}$
\end{enumerate}
\end{defn}

Finally, we define projections of $r$-neighborhoods, which simply refers to removing all edges from distance $r-1$ to distance $r$:

\begin{defn}[Projection]
For $r\ge 1$ and $y\in \mc R_r$, define the \emph{projection} $\text{proj}(y) \in \mc R_{r-1}$ to be $\text{proj}(y) := \text{end}_A(\text{res}_1(y))$.
\end{defn}

\subsection{Algorithms}\label{sec:PrelimAlgos}

We now define the matching-certified algorithms that we analyze in Section \ref{sec:VertexSurvivalProbability}. Any LOCAL algorithm can be easily transformed into this algorithms with one additional round as discussed in Section \ref{sec:ProofBody}. This is discussed in the proof of Theorem \ref{thm:main}. Matching-certified algorithms have the property that they always output a matching, but not necessarily one that is maximal. To define said algorithms, we start by defining incidence between $r$-flowers:

\begin{defn}[Incidence of $r$-Flowers]\label{def:inc-rF}
For $r\ge 0$, two $w,w'\in \mc F_r$ are said to be \emph{incident} iff there exist $v,v'\in \{A,B\}$ and a permutation $\sigma\in \Pi_{\Delta}$ for which both

\begin{enumerate}
    \item $\text{end}_v(w) = \sigma(\text{end}_{v'}(w'))$
    \item $\sigma(1)\ne 1$\footnote{This constraint on $\sigma$ merely ensures that $w$ and $w'$ are extensions in different directions from $\text{end}_v(w)$. 1 is used here rather than some other direction due to the fact that $w$ and $w'$ are extensions of $\text{end}_v(w)$ and $\text{end}_{v'}(w')$ respectively in the 1 directions.}
\end{enumerate}
\item 

An \emph{incidence permutation} for the pair $w,w'$ is a (non-unique) $\sigma\in \Pi_{\Delta}$ for which the above conditions hold.
\end{defn}

Now, we are ready to define the algorithms of interest:

\begin{defn}[$r$-Round Matching-Certified Algorithms]\label{def:match-certified}
For $r \ge 0$, an $r$-\emph{round matching-certified algorithm} is a measurable\footnote{The set $\mc F_r$ is equal to $[0,1]^d$ for some positive integer $d$. In particular, it is a continuous space, so probabilities and events are only sensible to define for measurable functions. This is a technicality, as (for example) any $r$-round algorithm that only uses a finite amount of random bits at each vertex is measurable.} function $f:\mc F_r\rightarrow\{0,1\}$ with the following property:

\begin{enumerate}
    \item (Matching-Certified) For any incident $w,w'\in \mc F_r$, $f(w)f(w') = 0$.
\end{enumerate}
\end{defn}

In particular, note that an $r$-round matching-certified algorithm $f$ outputs a matching with probability 1 as defined above, but that said matching may not be maximal. We handle algorithms that merely output a matching with high probability in Section \ref{sec:ProofBody}.

We assess $r$-round matching certified algorithms by one key metric: vertex survival probability. Intuitively, the vertex survival probability $P_f$ of a matching-certified algorithm $f$ is simply the probability that any given vertex is unmatched in the matching found by $f$:

\begin{defn}[Sampling From Distributions and Vertex Survival Probability]\label{def:sampling-and-vert-survival}
For $r\ge 0$, let $z\sim \mc R_r$ and $w\sim \mc F_r$ denote sampling $z$ and $w$ from the uniform distributions over $\mc R_r$ and $\mc F_r$ respectively. More generally, for a set $\Omega$, let $w\sim \Omega$ denote a sample $w$ taken uniformly from the set $\Omega$. The \emph{vertex survival probability} of an $r$-round matching-certified algorithm $f$, denoted $P_f$, is defined to be

$$P_f := \Pr_{z\sim \mc R_{r+1}}[\text{for all } i\in [\Delta], f(\text{res}_i(z)) = 0]$$
\end{defn}

\subsection{Distributions}\label{sec:PrelimDistributions}

The definitions of $\text{end}_v$ and $\text{res}_i$ involved direct manipulations of indices. We do not do any indexing into flowers and neighborhoods in Section \ref{sec:VertexSurvivalProbability}. Instead, we only use facts about them that we prove in this section. Specifically, we need facts about equivalence of sampling from multiple pairs of distributions.

In this paper, we work with continuous distributions and conditional distributions derived from those distributions. When we condition, we often condition on a measure 0 event. This is often not well-defined due to the Borel-Kolmogorov paradox, in which conditioning on membership in different great circles of a sphere can lead to distributions with different probability density functions. However, our conditional distributions are very simple and can be defined by fixing one coordinate system and doing substitutions within that coordinate system. We handle these technicalities in Appendix \ref{sec:AppDistributions}, where we prove the results stated in this section.\\

\textbf{Definitions}: First, we state the conditional probabilities that we define (via coordinate substitution). These definitions take place in Appendix \ref{sec:ConditionalDefns}. This proposition defines the idea of sampling an $r$-neighborhood conditioned on the $(r-1)$-flower in direction $i$ being fixed to some value $x$:

\begin{restatable}{prop}{propcondresdefn}\label{prop:cond-res-defn}
For any positive integer $r$, $i\in [\Delta]$, random variable (measurable function) $X:\mc R_r\rightarrow \mathbb R$, and $x\in \mc F_{r-1}$, the conditional expectation

$$\E_{y\sim \mc R_r}[X(y) \mid \text{res}_i(y) = x]$$

is well-defined.
\end{restatable}

This proposition defines the idea of sampling an $r$-flower conditioned on the $r$-neighborhood centered around one endpoint (either $A$ or $B$) being fixed to some value $x$:

\begin{restatable}{prop}{propcondenddefn}\label{prop:cond-end-defn}
For any nonnegative integer $r$, $v\in \{A,B\}$, random variable $X:\mc F_r\rightarrow \mathbb R$, and $x\in \mc R_r$, the conditional expectation

$$\E_{y\sim \mc F_r}[X(y) \mid \text{end}_v(y) = x]$$

is well-defined.
\end{restatable}

\textbf{Permutation invariance}: We now state the properties that we need involving these conditional expectations. First, we need permutation invariance of the uniform distributions over $\mc R_r$ and $\mc F_r$. This proposition shows that sampling an $r$-flower and flipping it is equivalent to sampling a fresh $r$-flower:

\begin{restatable}{prop}{propedgeflip}\label{prop:edge-flip}
For any nonnegative integer $r$ and random variable $X:\mc F_r\rightarrow \mathbb R$, 

$$\E_{y\sim \mc F_r}[X(\ol{y})] = \E_{y\sim \mc F_r}[X(y)]$$
\end{restatable}

This proposition shows that sampling an $r$-neighborhood and permuting the directions is equivalent to sampling a fresh $r$-neighborhood:

\begin{restatable}{prop}{propvertpermute}\label{prop:vert-permute}
For any nonnegative integer $r$, distribution $\mc D$ over the finite set of permutations $\Pi_{\Delta}$, and random variable $X:\mc R_r\rightarrow \mathbb R$, 

$$\E_{\sigma\sim \mc D, y\sim \mc R_r}[X(\sigma(y))] = \E_{y\sim \mc R_r}[X(y)]$$
\end{restatable}

This permutation invariance extends to the setting where sampling is conditioned on fixing the $(r-1)$-flower in some direction, with coordinates permuted accordingly:

\begin{restatable}{prop}{propcondvertpermute}\label{prop:cond-vert-permute}
For any positive integer $r$, $i\in [\Delta]$, $x\in \mc F_{r-1}$, and random variable $X:\mc R_r\rightarrow \mathbb R$, 

$$\E_{y\sim \mc R_r}[X(\sigma_i(y)) \mid \text{res}_i(y) = x] = \E_{z\sim \mc R_r}[X(z) \mid \text{res}_1(z) = x]$$
\end{restatable}

\textbf{Multi-step sampling equivalences}: We now state the four conditional sampling equivalences that we need. The first covers the equivalence between sampling $y$ uniformly from $\mc R_r$ and sampling from $\mc R_r$ via the following two-step procedure for an arbitrary $i\in [\Delta]$:

\begin{enumerate}
    \item Sample $x\sim \mc F_{r-1}$.
    \item Sample $y\sim \mc R_r$ conditioned on $\text{res}_i(y) = x$.
\end{enumerate}

Proving this and the other propositions in this section is very simple conceptually. For this proposition, for instance, we merely break up the coordinates of the random variable $y$ sampled uniformly from $\mc R_r$ into two groups: the coordinates kept in $\text{res}_i(y)$, and all other coordinates. The proofs of these results given in the appendix do this, but are somewhat cumbersome due to the necessity of working with the coordinate-based definitions of $\text{end}_v$ and $\text{res}_i$:

\begin{restatable}{prop}{propftorsample}\label{prop:f-to-r-sample}
For any positive integer $r$, $i\in [\Delta]$, and random variable $X:\mc R_r\rightarrow \mathbb R$,

$$\E_{x\sim \mc F_{r-1}}\left[\E_{y\sim \mc R_r}\left[X(y) \middle| \text{ }\text{res}_i(y) = x\right]\right] = \E_{y\sim \mc R_r}[X(y)]$$
\end{restatable}

The second covers the equivalence between sampling $y$ uniformly from $\mc F_{r-1}$ and sampling from $\mc F_{r-1}$ via the following three-step procedure:

\begin{enumerate}
    \item Sample $x\sim \mc R_{r-1}$.
    \item Sample $i\sim [\Delta]$.
    \item Sample $y\sim \mc F_{r-1}$ conditioned on $\text{end}_A(y) = \sigma_i(x)$.
\end{enumerate}

\begin{restatable}{prop}{proprtofsample}\label{prop:r-to-f-sample}
For any positive integer $r$ and random variable $X:\mc F_{r-1}\rightarrow \mathbb R$,

$$\E_{x\sim \mc R_{r-1}}\left[\E_{i\sim [\Delta]}\left[\E_{y\sim \mc F_{r-1}}\left[X(y) \middle|\text{ }\text{end}_A(y) = \sigma_i(x)\right]\right]\right] = \E_{y\sim \mc F_{r-1}}[X(y)]$$
\end{restatable}

The third covers the equivalence between sampling $y$ uniformly from $\mc F_r$ and sampling from $\mc F_r$ via the following four-step procedure that samples the middle and extensions on two sides separately:

\begin{enumerate}
    \item Sample $x\sim \mc F_{r-1}$
    \item Sample $z_A\sim \mc R_r$ conditioned on $\text{res}_1(z_A) = x$.
    \item Sample $z_B\sim \mc R_r$ conditioned on $\text{res}_1(z_B) = \ol{x}$.
    \item There is a unique $r$-flower $y\in \mc F_r$ for which $\text{end}_A(y) = z_A$ and $\text{end}_B(y) = z_B$. Return it.
\end{enumerate}

\begin{restatable}{prop}{propftofsample}\label{prop:f-to-f-sample}
For any positive integer $r$ and pair of random variables $X_A, X_B :\mc R_r\rightarrow \mathbb R$,

$$\E_{x\sim \mc F_{r-1}}\left[\E_{z_A\sim \mc R_r}[X_A(z_A) \mid \text{res}_1(z_A) = x]\E_{z_B\sim \mc R_r}[X_B(z_B) \mid \text{res}_1(z_B) = \ol{x}]\right] = \E_{y\sim \mc F_r}[X_A(\text{end}_A(y))X_B(\text{end}_B(y))]$$
\end{restatable}

The last covers the equivalence between sampling $y$ uniformly from $\mc R_r$ conditioned on $\text{proj}(y) = x$ for some fixed $x\in \mc R_{r-1}$ and the following procedure:

\begin{enumerate}
    \item Given $x\in \mc R_{r-1}$:
    \begin{enumerate}
        \item (Independently) for each $i\in [\Delta]$, sample $z_i\sim \mc F_{r-1}$ conditioned on $\text{end}_A(z_i) = \sigma_i(x)$.
        \item There is a unique $r$-neighborhood $y\in \mc R_r$ for which $\text{res}_i(y) = z_i$ for all $i\in [\Delta]$. Return it.
    \end{enumerate}
\end{enumerate}

\begin{restatable}{prop}{proprtorsample}\label{prop:r-to-r-sample}
For any positive integer $r$, $x\in \mc R_{r-1}$, and any collection $z_1,z_2,\hdots,z_{\Delta} \in \mc F_{r-1}$ satisfying the constraint that $\text{end}_A(z_i) = \sigma_i(x)$ for all $i\in [\Delta]$, there exists a unique $r$-neighborhood $\text{glue}(z_1,z_2,\hdots,z_{\Delta})\in \mc R_r$ with the property that

$$\text{res}_i(\text{glue}(z_1,z_2,\hdots,z_{\Delta})) = z_i$$

for all $i\in [\Delta]$. Furthermore, for any random variable $X:\mc R_r\rightarrow \mathbb R$, any $i\in [\Delta]$, and any $\zeta\in \mc F_{r-1}$ satisfying $\text{end}_A(\zeta) = \sigma_i(x)$,

\begin{align*}
&\E_{y\sim \mc R_r}[X(y) \mid \text{res}_i(y) = \zeta]\\
&= \E_{z_j\sim \mc F_{r-1} \text{ }\forall j\in [\Delta]}\left[X(\text{glue}(z_1,z_2,\hdots,z_{\Delta})) \mid (z_i = \zeta) \text{ and } (\text{end}_A(z_j) = \sigma_j(x) \text{ for all } j\in [\Delta])\right]
\end{align*}
\end{restatable}

\subsection{Other Results: Chernoff with Bounded Dependence}\label{sec:PrelimOther}

In our proofs, we also use Chernoff bounds with bounded dependence:

\begin{thm}[Inequality (3) of \cite{LampertRZ18}, derived from Theorem 2.1 of \cite{Janson04}]\label{thm:dep-chernoff}
Consider $n$ random variables $\mc A = \{X_1,X_2,\hdots,X_n\}$ with the property that $0\le X_i\le 1$ for all $i\in [n]$ almost surely. Then

$$\Pr\left[\sum_{i=1}^n X_i - \sum_{i=1}^n \E[X_i] > \lambda\right] \le \exp\left(-\frac{2\lambda^2}{n\chi(\mc A)}\right)$$

where $\chi(\mc A)$ is the chromatic number of the dependency graph of $\mc A$.
\end{thm}

\section{Proof of Theorem~\ref{thm:main}}\label{sec:ProofBody}

Our main technical result is the following result about vertex survival probability of $r$-round matching-certified algorithms, which we prove in Section \ref{sec:VertexSurvivalProbability}:

\begin{restatable}{thm}{thmvertexmain}\label{thm:vertex-main}
For any $r\ge 0$ and $r$-round matching-certified algorithm $f$, $P_f \ge C_1^{-r}$, where $C_1 := 10^{80}$.
\end{restatable}

We use this result to prove Theorem \ref{thm:main}. The hard instances of the lower bound on the survival probability from Theorem \ref{thm:vertex-main} are graphs where the $r$-neighborhood of each vertex is a $\Delta$-regular tree (or at least for the vast majority of vertices). In order to directly show that an algorithm does not produce a maximal matching, we need to have a lower bound on the algorithm's \emph{edge} survival probability. This would require an extension of Theorem \ref{thm:vertex-main}. Instead, we design the hard graph to make it so that lower bounding vertex survival probability suffices. In particular, random regular graphs can be sampled in such a way that they have a high girth and a small maximum independent set. We take advantage of these two properties as follows:

\begin{enumerate}
    \item High girth enables the application of Theorem \ref{thm:vertex-main} to show that any $r$-round algorithm for small enough $r$ has high vertex survival probability.
    \item Applying Chernoff with bounded dependence (Theorem \ref{thm:dep-chernoff}) shows that, with high probability, a large set of vertices is unmatched by the algorithm.
    \item Since the lower bound instance graph only has small independent sets, the unmatched set of vertices cannot be an independent set, meaning that there is an edge between some pair of unmatched vertices. Thus, the algorithm can only produce a maximal matching with very low probability.
\end{enumerate}
   
The following result describes our family of hard instances. Note that we state a stronger property than an upper bound on the size of any independent set -- we ask for all sets with size larger than the independence threshold to have their induced subgraph sizes concentrate around their expectations. This result was effectively shown by \cite{Bollobas80, Bollobas88, BrandtCGGRV22}, but we prove it in the appendix to be self-contained since we need stronger guarantees:

\begin{restatable}[Adapted from \cite{Bollobas80, Bollobas88, BrandtCGGRV22}]{lem}{lemconfigmodel}\label{lem:config-model}
For any two even positive integers $n,\Delta$, there exists a graph $G_{n,\Delta}$ with the following properties:
\begin{enumerate}
    \item (Regularity) $G_{n,\Delta}$ is $\Delta$-regular.
    \item (Girth) $G_{n,\Delta}$ has girth at least $(\log_\Delta n)/1000$.
    \item (Subgraph size) For any $S\subseteq V(G_{n,\Delta})$, the following properties hold:
    \begin{enumerate}
        \item If $\sqrt{n} \le |S| \le \frac{10^6n\ln\Delta}{\Delta}$, then $|E(G_{n,\Delta}[S])| < 10^8(\ln(n/|S|))|S|$.
        \item If $|S| > \frac{10^6n\ln\Delta}{\Delta}$, then $\left||E(G_{n,\Delta}[S])| - \frac{|S|^2\Delta}{2n}\right| < \frac{|S|^2\Delta}{50n}$.
    \end{enumerate}
\end{enumerate}
\end{restatable}

Now that we are equipped with a proof outline and a family of hard instances, we proceed towards proving our main results.

\subsection{A Lower Bound for MM in High-Girth Regular Graphs}\label{sec:LBHighGirthMM}

First, we reason about maximal matchings in regular graphs. Our proof conforms exactly to the outline given earlier. In particular, we do use the subgraph size guarantee of Lemma \ref{lem:config-model}.

\thmmain*

\begin{proof}[Proof of Theorem \ref{thm:main} for regular graphs]
Assume for the sake of contradiction that there exists an $r = C_0\min(\log\Delta, \sqrt{\log n})$-round LOCAL algorithm $\mc A$ that, when given any unique fixed IDs, outputs a maximal matching with probability at least $\Delta^{-1/1000}$. Let $f_{\mc A}$ be the $(r+1)$-round matching-certified algorithm (which only receives randomness tapes for each edge, \emph{not} an ID) obtained by deleting incident edges in the solution produced by $\mc A$, reassigning randomness from vertices to edges, and reserving the first $10\log n$ bits of each randomness tape as a (with high probability unique) ID to supply to the algorithm $\mc A$. Note that the probability that $f_{\mc A}$ finds a maximal matching is greater than or equal to the probability that $\mc A$ finds a maximal matching, as the postprocessing step introduced to obtain $f_{\mc A}$ does not modify the obtained matching in any case in which $\mc A$ finds a matching. Thus, it suffices to upper bound the probability that $f_{\mc A}$ finds a maximal matching.

$f_{\mc A}$ is an $(r+1)$-round matching-certified algorithm. We now discuss which graph to apply $f_{\mc A}$ to. Let $\Delta := 2^{(r+1)/C_0}$ and let $n := \Delta^{(r+1)/C_0}$. We apply $f_{\mc A}$ to the graph $G := G_{n,\Delta}$ given by Lemma \ref{lem:config-model}. Note that $r+1 \le 2C_0\min(\log\Delta, \sqrt{\log n})$. By the girth guarantee, $G$ has girth at least $(\log_{\Delta} n) / 1000 > (r+1)/(1000C_0) > 10(r+1)$. Thus, the regularity guarantee implies that the $(r+1)$-neighborhood of any vertex is a $\Delta$-ary tree. We may therefore apply Theorem \ref{thm:vertex-main}. This shows that

$$P_{f_{\mc A}} > C_1^{-r-1} \ge C_1^{-2C_0\log \Delta} > \Delta^{-1/1000}$$

which in turn shows that applying $f_{\mc A}$ to $G$ results in any given vertex being unmatched with probability at least $\Delta^{-1/1000}$. For each $v\in V(G)$, let $X_v$ denote the indicator variable of the event that $v$ is unmatched by the algorithm $f_{\mc A}$. By linearity of expectation, $\E\left[\sum_{v\in V(G)} X_v\right] >  \Delta^{-1/1000} n$. These random variables are dependent, but the dependency graph on them has chromatic number at most its degree, which is at most $\chi(\{X_v\}_{v\in V(G)}) \le \Delta^{4r} < n^{1/1000}$. By Theorem \ref{thm:dep-chernoff} with $\lambda \gets \Delta^{-1/1000} n/2$ (applied to $-X_v$),

$$\Pr\left[\sum_{v\in V(G)} X_v < \Delta^{-1/1000} n/2\right] \le \exp\left(-\frac{2(\Delta^{-1/1000} n/2)^2}{n^{1 + 1/1000}}\right) \le \exp(-\sqrt{n})$$

Thus, the probability over the algorithm's randomness that the number of unmatched vertices is less than $\Delta^{-1/1000} n/2$ is at most $e^{-\sqrt{n}}$. By the subgraph size guarantee of Lemma \ref{lem:config-model}, part (b), any independent set in $G$ must have at most $10^6(n\log\Delta)/\Delta < \Delta^{-1/1000} n/2$ vertices. Thus, the set of unmatched vertices cannot form an independent set, meaning that the matching found by $f_{\mc A}$ is not maximal if the number of unmatched vertices is at least $\Delta^{-1/1000} n/2$. Thus, $f_{\mc A}$ and thus $\mc A$ outputs a maximal matching with probability at most $\exp(-\sqrt{n}) < \Delta^{-1/1000}$ on $G$, a contradiction to $\mc A$ having success probability at least $\Delta^{-1/1000}$ for any IDs, as desired.
\end{proof}

\subsection{A Lower Bound for MM in Trees}\label{sec:LBTrees}

Now, we reason about the ability for fast algorithms to find maximal matchings in trees. One can use a similar strategy to the proof of Theorem \ref{thm:main} for regular graphs (Section \ref{sec:LBHighGirthMM}) to show impossibility results for algorithms that have success probability higher than $1 - 1/(2n)$. We can do better though with a more nuanced approach. When applying an algorithm to $G_{n,\Delta}$, implicitly the algorithm is being run on $n$ different trees sampled from the same distribution. The number of vertices that are far away from any unmatched edge is thus proportional to the probability that the algorithm finds a matching when applied to any given tree. Since there are many surviving edges, the subgraph size property of $G_{n,\Delta}$ implies an expansion property as well. This expansion property can be used to show that few vertices in $G_{n,\Delta}$ are far away from surviving edges, which implies that any fast algorithm for trees must succeed with low probability, as desired.

Before proving the theorem, we prove a proposition linking the subgraph size guarantee to expansion. For a set of vertices $S$, let $N(S) \subseteq V(G)$ denote $S$ together with its neighbors in $G$:

\begin{prop}\label{prop:expansion}
For any two positive even integers $n,\Delta$ and any $S\subseteq G_{n,\Delta}$ with $n/e^{\Delta^{1/4}/10^9} \le |S| \le 4n/5$,

$$\frac{|N(S)|}{|S|} > \sqrt{\min\left(\Delta, \frac{n}{|S|}\right)}$$
\end{prop}

\begin{proof}
Let $T := N(S)$ and $G := G_{n,\Delta}$ to simplify notation. We split the proof into two cases depending on the size of $S$. If $|S| \le \frac{10^6n\ln\Delta}{\Delta}$, then $|E(G[S])| < \frac{|S|\Delta^{1/4}}{10} < \frac{9|S|\Delta}{20}$ by the first part of the subgraph size guarantee. If $|S| > \frac{10^6n\ln \Delta}{\Delta}$, then $|E(G[S])| \le \frac{13|S|^2\Delta}{25n} \le \frac{13}{25}\frac{4}{5}|S|\Delta < \frac{9|S|\Delta}{20}$. The regularity of $G$ applied to the vertices in $S$ implies that

$$|E(G[T])| \ge \Delta |S| - |E(G[S])| > \frac{11\Delta |S|}{20}$$

By the subgraph size guarantee applied to $T$,

$$|E(G[T])| \le \max\left(10^8\ln\left(\frac{n}{|T|}\right)|T|,\frac{13|T|^2\Delta}{25n}\right) \le \max\left(\frac{\Delta^{1/4}|T|}{10}, \frac{13|T|^2\Delta}{25n}\right)$$

When the first term is the maximizer, combining the two inequalities yields

$$\frac{|T|}{|S|} > 11\Delta^{3/4}/2 > \sqrt{\Delta}$$

When the second term is the maximizer, combining the two inequalities yields

$$\frac{|T|^2}{|S|^2} > \frac{55n}{52|S|} > \frac{n}{|S|}$$

Thus,

$$\frac{|T|}{|S|} > \sqrt{\min\left(\Delta, \frac{n}{|S|}\right)}$$

as desired.
\end{proof}

Now, we are ready to prove the desired result about hardness on trees:

\thmmain*

\begin{proof}[Proof of Theorem \ref{thm:main} for $n$-vertex $\Delta$-ary trees]
Assume for the sake of contradiction that there exists an $r = C_0\min(\log\Delta, \sqrt{\log n})$-round LOCAL algorithm $\mc A$ that, when given a tree and any unique fixed IDs on the vertices of that tree, outputs a maximal matching with probability at least $\Delta^{-1/1000}$. This algorithm is not always required to output a matching. Let $f_{\mc A}$ be the $(r+1)$-round matching-certified algorithm (which only receives randomness tapes for each edge, \emph{not} an ID) obtained by deleting pairs of incident edges in the output solution, reassigning randomness from vertices to edges, and reserving the first $10\log n$ bits of each randomness tape as an ID to supply to the LOCAL algorithm $\mc A$. $f_{\mc A}$ always outputs a matching and has success probability at least as high as $\mc A$, so it suffices to reason about $f_{\mc A}$. Let $\Delta := 2^{(r+1)/C_0}$ and $n := \Delta^{(r+1)/C_0}$. Note that $r+1 \le 2C_0\min(\log\Delta,\sqrt{\log n})$. Let $T_v$ be a $\Delta$-ary tree with depth $(r+1)/C_0$ rooted at a vertex $v$ (i.e. so that $|V(T_v)| = n$). In particular, $T_v$ contains vertices without degree $\Delta$ (leaves). For a graph $H$, let $f_{\mc A}(H,\sigma)$ denote the matching obtained by applying the algorithm $f_{\mc A}$ to $H$ given randomness $\sigma$. We just need to upper bound the following (success) probability:

$$\Pr_{\sigma}[f_{\mc A}(T_v,\sigma) \text{ is maximal}]$$

It suffices to bound the probability of surviving for central vertices only:

$$\Pr_{\sigma}[f_{\mc A}(T_v,\sigma) \text{ is maximal}] \le \Pr_{\sigma}[f_{\mc A}(T_v,\sigma) \text{ is maximal in distance $\le (r+1)/(2C_0)$ from $v$}]$$

where maximality in distance $R$ refers to maximality on the subgraph induced by the $R$-radius neighborhood of $v$. We upper bound this probability by exploiting the indistinguishability of this graph from a certain high girth graph. Let $N := n^{1/C_0}$. We apply $f_{\mc A}$ to the graph $G := G_{N,\Delta}$ given by Lemma \ref{lem:config-model}. By the girth guarantee, $G$ has girth at least $(\log_\Delta N)/1000 > 100(r+1)/C_0$. Thus, the regularity guarantee implies that the $100(r+1)/C_0$-neighborhood of any vertex is a $\Delta$-ary tree, meaning that the vertices within distance $R := (r+1)/(2C_0)$ of a vertex $v\in V(G)$ cannot distinguish betweeen $T_v$ and $G$ in $r+1$ rounds. Thus, for any $v\in V(G)$,

$$\Pr_{\sigma}[f_{\mc A}(T_v,\sigma) \text{ is maximal in distance $\le R$ from $v$}] \le \Pr_{\sigma}[f_{\mc A}(G,\sigma) \text{ is maximal in distance $\le R$ from $v$}]$$

For randomness $\sigma$, let $X_{\sigma}\subseteq V(G)$ denote the set of vertices $v\in V(G)$ for which $f_{\mc A}(G,\sigma)$ is maximal within distance $\le R$ from $v$. By linearity of expectation,

$$\Pr_{\sigma}[f_{\mc A}(G,\sigma) \text{ is maximal in distance $\le R$ from $v$}] = \frac{\E_{\sigma}[|X_{\sigma}]|}{N}$$

Thus, it suffices to bound the expected size of $X_{\sigma}$. For randomness $\sigma$, let $Y_{\sigma}$ denote the set of vertices $v\in V(G)$ that are not incident to an edge in $f_{\mc A}(G,\sigma)$. By Theorem \ref{thm:vertex-main},

$$\E_{\sigma}[|Y_{\sigma}|] \ge P_{f_{\mc A}}N \ge C_1^{-r-1}N \ge C_1^{-2C_0\log\Delta}N > \Delta^{-1/1000}N$$

By Theorem \ref{thm:dep-chernoff},

$$\Pr_{\sigma}[|Y_{\sigma}| \le \Delta^{-1/1000}N/2] \le \exp\left(-\frac{2(\Delta^{-1/1000}N/2)^2}{N\Delta^{4r}}\right) \le \exp(-\sqrt{N})$$

Let $Z_{\sigma}$ denote the set of $v\in V(G)$ for which $v\in Y_{\sigma}$ and there exists a neighbor $u\in Y_{\sigma}$ of $v$ in $G$. When $|Y_{\sigma}| > \Delta^{-1/1000}N/2$, $|Y_{\sigma}| > \Delta^{-1/1000}N/2 > 2(10^6 N\log\Delta)/\Delta$. By definition, $Y_{\sigma}\setminus Z_{\sigma}$ is an independent set in $G$. Part (b) of the subgraph size property of $G$ implies that $|Y_{\sigma}\setminus Z_{\sigma}| \le \frac{10^6 N\ln \Delta}{\Delta}$. Thus, $|Z_{\sigma}| = |Y_{\sigma}| - |Y_{\sigma}\setminus Z_{\sigma}| > \Delta^{-1/1000}N/4$. If, for the sake of contradiction, $|X_{\sigma}| \ge \Delta^{-1/1000}N/2$, then applying Proposition \ref{prop:expansion} $\ell = \log\log\Delta$ times shows that the $\ell$-neighborhoods of both $X_{\sigma}$ and $Z_{\sigma}$ contain at least $4N/5$ vertices. Since $|V(G)| = N$ and $4N/5 + 4N/5 > N$, these neighborhoods must overlap, meaning that there exist $u\in X_{\sigma}$ and $v\in Z_{\sigma}$ within distance $2\ell < R$ of each other. This contradicts the definition of $X_{\sigma}$, so $|Y_{\sigma}| > \Delta^{-1/1000}N/2$ implies $|X_{\sigma}| < \Delta^{-1/1000}N/2$, meaning that

$$\Pr_{\sigma}[|X_{\sigma}| \ge \Delta^{-1/1000}N/2] \le \Pr_{\sigma}[|Y_{\sigma}| \le \Delta^{-1/1000}N/2] \le \exp(-\sqrt{N})$$

Thus,

$$\E_{\sigma}[|X_{\sigma}|] \le \Delta^{-1/1000}N/2 + N\exp(-\sqrt{N}) \le \Delta^{-1/1000}N$$

and the desired probability is at most $\Delta^{-1/1000}$. This contradicts the fact that $f_{\mc A}$ had success probability at least $\Delta^{-1/1000}$, as desired.
\end{proof}

\section{Vertex Survival Probability (Proof of Theorem \ref{thm:vertex-main})}\label{sec:VertexSurvivalProbability}

    


In this section we prove Theorem \ref{thm:vertex-main}, which recall is the following:

\thmvertexmain*

\noindent Theorem \ref{thm:vertex-main} is implied by the following self-reduction round elimination result:

\begin{restatable}{lem}{lemvertexroundelim}\label{lem:vertex-round-elim}
For any $r\ge 1$ and $r$-round matching-certified algorithm $f$, there exists an $(r-1)$-round matching-certified algorithm $g$ with $P_g \le C_1 P_f$.
\end{restatable}

\begin{proof}[Proof of Theorem \ref{thm:vertex-main} given Lemma \ref{lem:vertex-round-elim}]
Obtain a sequence of matching-certified algorithms $f_0,f_1,\hdots,f_r$ as follows. Let $f_r := f$ and for each $s\in \{r-1,r-2,\hdots,1,0\}$, let $f_s$ be the algorithm $g$ obtained by applying Lemma \ref{lem:vertex-round-elim} to $f_{s+1}$. By Lemma \ref{lem:vertex-round-elim}, $f_s$ is an $s$-round matching-certified algorithm with $P_{f_s} \le P_f C_1^{r-s}$. Since $f_0$ is matching-certified and 0-round, $f_0(w) = 0$ for any $w\in \mc F_0$, as any two 0-flowers are incident. Thus, $P_{f_0} = 1$. As a result, $1 \le P_f C_1^r$, which means that $P_f \ge C_1^{-r}$ as desired.
\end{proof}

The rest of this section is dedicated to proving Lemma~\ref{lem:vertex-round-elim}. To this end, we first obtain several important helper results that allow us to define the desired $(r{-}1)$-round algorithm~$g$. Some of these results concern $(r{-}1)$-flowers that can be extended to many accepting $r$-flowers (Proposition~\ref{prop:r-1-flower-high-weak} and Lemma~\ref{lem:r-1-flower-high-strong}, which are given in Section~\ref{sec:goodflowers}), while others concern $(r{-}1)$-neighborhoods that have a dominant direction that can be extended to many such $(r{-}1)$-flowers (Lemma~\ref{lem:r-1-neighborhood-low-common}, Proposition~\ref{prop:r-1-neighborhood-low-weak}, Lemma~\ref{lem:r-1-neighborhood-low-strong}, and Proposition~\ref{prop:complement-prob}, which are given in Section~\ref{sec:dominantneighborhoods}). The description of the $(r{-}1)$-round algorithm~$g$ is then given in Section~\ref{sec:prooflemmaRE}, along with the proof of Lemma~\ref{lem:vertex-round-elim}. The proofs of these various propositions and lemmas are given in their respective sections, except for Proposition~\ref{prop:r-1-flower-high-weak} and Lemma~\ref{lem:r-1-flower-high-strong}, which are proved in Section~\ref{subsec:r-1-flower-high}, Lemma~\ref{lem:r-1-neighborhood-low-common}, which is proved in Section~\ref{subsec:r-1-neighborhood-low}, and Proposition~\ref{prop:complement-prob}, which is proven in Section~\ref{sec:complement-prob}.

Let us start with the following simple proposition and definition (Proposition~\ref{prop:r-ngbr-exclusive} which is followed by the definition of $\text{dir}$). Observe that an $r$-round matching-certified algorithm $f$ has deterministic behavior when given $w\in \mc F_r$ by virtue of it being an $r$-round algorithm and the fact that $w$ already contains all the randomness tapes on all the edges it contains. In the following proposition we show that giving $f$ just an $r$-neighborhood rather than an $r$-flower also leads to deterministic behavior:\footnote{Recall that for $x\in \mc R_r$, $\sigma_i(x)$ is the $r$-neighborhood obtained by switching the 1 and $i$ directions. The constraint $\text{end}_A(w) = \sigma_i(y)$ is encoding that $w$ is an extension of $y$ in the $i$-th direction. Figure \ref{fig:dir} depicts the proof of Proposition \ref{prop:r-ngbr-exclusive}.}

\begin{prop}\label{prop:r-ngbr-exclusive}
For any $r$-round matching algorithm $f$ and any $y\in \mc R_r$, there is at most one $i\in [\Delta]$ for which there exists a $w\in \mc F_r$ for which both of the following properties are satisfied:
\begin{enumerate}
\item $f(w) = 1$
\item $\text{end}_A(w) = \sigma_i(y)$.
\end{enumerate}
\end{prop}

\begin{proof}
Suppose, for the sake of contradiction, that there exists $y\in \mc R_r$ and $i\ne i'\in [\Delta]$ along with $w,w'\in \mc F_r$ for which $f(w) = f(w') = 1$, $\text{end}_A(w) = \sigma_i(y)$, and $\text{end}_A(w') = \sigma_{i'}(y)$. The latter two constraints imply that

$$\text{end}_A(w) = \sigma_i(\sigma_{i'}(\text{end}_A(w')))$$

We now show that $\sigma_i(\sigma_{i'}(1)) \ne 1$. If $i' = 1$, then $i \ne 1$, which implies that $\sigma_i(\sigma_{i'}(1)) = \sigma_i(i') = \sigma_i(1) = i\ne 1$ as desired. If $i'\ne 1$, then $\sigma_i(\sigma_{i'}(1)) = \sigma_i(i') = i' \ne 1$ as desired, where the last equality uses both $i'\ne i$ and $i'\ne 1$. Thus, in both cases, $\sigma_i(\sigma_{i'}(1)) \ne 1$. Therefore, $w$ and $w'$ are incident. But $f(w)f(w') = 1$, so $f$ cannot be matching-certified. This is a contradiction. Thus, the desired result holds.
\end{proof}

\paragraph{Defining $\text{dir}(f,y)$:} For $r\ge 1$ and $y\in \mc R_r$, let $\text{dir}(f,y) \in [[\Delta]]$\footnote{Recall that $[[\Delta]] = \{0,1,\hdots,\Delta\}$ and that any permutation $\sigma$ that we use in this paper satisfies $\sigma(0) = 0$.} be the value of $i$ obtained from Proposition \ref{prop:r-ngbr-exclusive}. If $i$ does not exist (no extension of $y$ is matched by $f$), then define $\text{dir}(f,y) := 0$.\footnote{$\text{dir}(f,y)$ could be defined to be an arbitrary direction in this case, but a distributed algorithm cannot necessarily symmetry-break to return an arbitrary index, but it can always return 0.} 

\subsection{Good $(r-1)$-Flowers}\label{sec:goodflowers}
So far, we have learned about the determinism of an $r$-round algorithm $f$ when restricted to both $r$-flowers (immediate) and $r$-neighborhoods (Proposition \ref{prop:r-ngbr-exclusive}). $f$ does not behave deterministically on $(r-1)$-flowers, but it turns out that $f$ does behave \emph{nearly} deterministically on $(r-1)$-flowers if it has low enough vertex survival probability. In this section we prove helper results towards illustrating this phenomenon. We begin with the following definition of $\delta$-good $(r-1)$-flowers. Intuitively speaking, these are $(r-1)$-flowers $x$ that can be extended to many $r$-neighborhoods $y$ (from either side) that have their $\text{dir}(f,y)$ pointing towards $x$.

\begin{defn}[\textbf{Good $(r-1)$-Flowers}]\label{def:goodFlowers}
    Suppose that $f$ is an $r$-round matching-certified algorithm and let $\delta\in (0,1]$. We define $\mc X_{r-1}(f,\delta) \subseteq \mc F_{r-1}$ to be the set of $(r-1)$-flowers $x$ for which \emph{both}

\begin{enumerate}
\item $\displaystyle\Pr_{y\sim \mc R_r} [\text{dir}(f,y) = 1 \mid \text{res}_1(y) = x] \ge (1 - \delta)$
\item $\displaystyle\Pr_{y\sim \mc R_r} [\text{dir}(f,y) = 1 \mid \text{res}_1(y) = \ol{x}] \ge (1 - \delta)$
\end{enumerate}
\end{defn}

The 1s that are present in this definition may look confusing -- why should 1 be a privileged direction? It is not important to this definition. The important thing is that $\text{dir}(f,y)$ is pointing in the \emph{same} direction for which $y$ is an extension of $x$ and stating 1 here is solely for convenience:

\begin{fact}\label{fact:same-dir}
For any $r\ge 1$, $i\in [\Delta]$, and $x\in \mc F_{r-1}$,

$$\Pr_{y\sim \mc R_r}[\text{dir}(f,y) = i \mid \text{res}_i(y) = x] = \Pr_{y\sim \mc R_r}[\text{dir}(f,y) = 1 \mid \text{res}_1(y) = x]$$
\end{fact}

\begin{proof}
Follows immediately from Proposition \ref{prop:cond-vert-permute} (symmetry of $\mc R_r$ under coordinate permutation), the fact that $\text{res}_i(y)$ is defined to be $\text{res}_1(\sigma_i(y))$, and the fact that $\text{dir}(f,\sigma_i(y)) = \sigma_i(\text{dir}(f,y))$ by definition of $\sigma_i(y)$:

\begin{align*}
\Pr_{y\sim \mc R_r}[\text{dir}(f,y) = 1 \mid \text{res}_1(y) = x] &= \Pr_{y\sim \mc R_r}[\text{dir}(f,\sigma_i(y)) = 1 \mid \text{res}_1(\sigma_i(y)) = x]\\
&= \Pr_{y\sim \mc R_r}[\text{dir}(f,y) = i \mid \text{res}_i(y) = x]
\end{align*}
\end{proof}

Consider $w\in \mc F_r$. If $\text{dir}(f,\text{end}_A(w))$ and $\text{dir}(f,\text{end}_B(w))$ both point towards $w$, we may without loss of generality assume\footnote{We do not actually use this assumption in the proof.} that $f(w) = 1$, as the definition of $\text{dir}$ ensures that no other $r$-flower $w'$ incident with $w$ can have $f(w') = 1$. Thus, a $\delta$-good $(r-1)$-flower $x$ also has the property that a uniformly random extension of $x$ to an $r$-flower $w$ has $f(w) = 1$ with probability at least $(1 - \delta)^2$. In other words, the behavior of $f$ on good $(r-1)$-flowers is nearly deterministic; in particular, the algorithm nearly always accepts when solely given a good $(r-1)$-flower.

In our analysis, it is essential to know the probability that a random $(r-1)$-flower is $\delta$-good. Intuitively, for any matching-certified algorithm $f$, the fraction of $(r-1)$-flowers that are $\delta$-good is at most $1/((1-\delta)^2\Delta) \approx 1/\Delta$, no matter its vertex survival probability. This is because for any $(r+1)$-neighborhood, only 1 out of every $\Delta$ of its incident $r$-flowers can be in the matching, and the previous paragraph asserted that for every $\delta$-good $(r-1)$-flower, a $(1-\delta)^2$-fraction of extensions must accept. So naturally, it is helpful to know how close we can get to matching that bound. If $f$ has low vertex survival probability, then we can get quite close:

\begin{restatable}{prop}{proproneflowerhighweak}\label{prop:r-1-flower-high-weak}
Let $f$ be an $r$-round matching-certified algorithm and let $\delta\in (0,1]$. It holds that:

$$\Pr_{x\sim \mc F_{r-1}}[x\in \mc X_{r-1}(f,\delta)] \ge \frac{1 - 2P_f/\delta}{\Delta}$$
\end{restatable}

Intuitively, the probability mass of good $(r-1)$-flowers is so close to the upper bound that we can ignore the behavior of $f$ on all other $(r-1)$-flowers when constructing the desired $(r-1)$-round algorithm $g$ for Lemma \ref{lem:vertex-round-elim}. Nonetheless, this lower bound is not quite strong enough for our purposes. In the following stronger lemma, instead of working with a given specific value of $\delta$, we work with all values of $\delta$ together via an expectation:

\begin{restatable}{lem}{lemroneflowerhighstrong}\label{lem:r-1-flower-high-strong}
For any $r$-round matching-certified algorithm $f$ and $\xi\in (0,1)$,

$$\E_{\tau\sim [0,1]}\left[\Pr_{x\sim \mc F_{r-1}}\left[x\in \mc X_{r-1}(f,\tau)\right] \middle| \tau\le\xi\right] \ge \frac{1 - 7P_f/\xi}{\Delta}$$
\end{restatable}

The proofs of Proposition \ref{prop:r-1-flower-high-weak} and Lemma \ref{lem:r-1-flower-high-strong} are both given in Section \ref{subsec:r-1-flower-high} (which is solely dedicated to proving these two results). Following these two results, membership in $\mc X_{r-1}(f,\delta)$ (for some choice of a small $\delta$) makes for a promising candidate to define the $(r-1)$-round algorithm $g$ that is needed to prove Lemma~\ref{lem:vertex-round-elim} (i.e., for the round-elimination step). Unfortunately, this algorithm is not necessarily a matching-certified algorithm, as nothing prevents $\mc X_{r-1}(f,\delta)$ from containing two incident $(r-1)$-flowers. Towards making $g$ matching-certified, in the following section, we analyze the performance of an $r$-round algorithm $f$ when given just an $(r-1)$-neighborhood. 

\subsection{Dominant Directions of $(r-1)$-Neighborhoods}\label{sec:dominantneighborhoods}

In this section, we prove several key properties about the performance of an $r$-round algorithm $f$ on $(r-1)$-neighborhoods. We seek to understand how often extensions of $(r-1)$-neighborhoods are good $(r-1)$ flowers. We now define some quantities involving these extensions:

\begin{defn}[$(r-1)$-Neighborhood Extension Probabilities]
For any $r\ge 1$, $r$-round matching-certified algorithm $f$, $x\in \mc R_{r-1}$, $i\in [\Delta]$, and $\delta\in [0,1]$, define

$$P_i(f,x,\delta) := \Pr_{y\sim \mc F_{r-1}}[y\in \mc X_{r-1}(f,\delta) \mid \text{end}_A(y) = \sigma_i(x)]$$

Also define

$$P_{\max}(f,x,\delta) := \max_{i\in [\Delta]} P_i(f,x,\delta)$$

and

$$i_{\max}(f,x,\delta) := \arg\max_{i\in [\Delta]} P_i(f,x,\delta)$$

with $i_{\max}(f,x,\delta) := 0$ if there is a tie.\footnote{In particular, $i_{\max}(f,x,\delta) \in [[\Delta]] = \{0,1,\hdots,\Delta\}$. As with the definition of $\text{dir}$, $i_{\max}$ could also be defined arbitrarily in the tiebreaking case, but returning 0 makes it so that the algorithm does not have to break symmetry.} Moreover, define 

$$P(f,x,\delta) := \sum_{i=1}^{\Delta} P_i(f,x,\delta)$$

and

$$P_{\text{comp}}(f,x,\delta) := P(f,x,\delta) - P_{\max}(f,x,\delta)$$
\end{defn}

Note that for any $i\in [\Delta]$, $P_i(f,x,\delta)$ is a nondecreasing function of $\delta$. As a result, $P_{\max}(f,x,\delta)$ is also a nondecreasing function of $\delta$. It turns out that $P_{\text{comp}}(f,x,\delta)$ is also a nondecreasing function of $\delta$:

\begin{prop}\label{prop:comp-monotone}
For any $r$-round matching-certified algorithm $f$, $x\in \mc R_{r-1}$, and $\delta\le \tau\in [0,1]$,

$$P_{\text{comp}}(f,x,\delta) \le P_{\text{comp}}(f,x,\tau)$$
\end{prop}

\begin{proof}
Let $i$ and $j$ be indices for which $P_i(f,x,\delta) = P_{\max}(f,x,\delta)$ and $P_j(f,x,\tau) = P_{\max}(f,x,\tau)$ respectively\footnote{We did not define these based on $i_{\max}$ because $i_{\max} = 0$ when the maximizer is not unique.}. If $i = j$, the result immediately follows from the fact that $P_k(f,x,\delta) \le P_k(f,x,\tau)$ for all $k\in [\Delta]$, so assume that $i\ne j$. Then

\begin{align*}
P_{\text{comp}}(f,x,\delta) &= P_j(f,x,\delta) + \sum_{k\in [\Delta], k\ne i,j} P_k(f,x,\delta)\\
&\le P_i(f,x,\delta) + \sum_{k\in [\Delta], k\ne i,j} P_k(f,x,\delta)\\
&\le P_i(f,x,\tau) + \sum_{k\in [\Delta], k\ne i,j} P_k(f,x,\tau)\\
&= P_{\text{comp}}(f,x,\tau)
\end{align*}

as desired.
\end{proof}

Intuitively speaking, we show that any $(r-1)$-neighborhood has a dominant direction. This direction has the property that all other directions combined have a limited number of good $(r-1)$-flowers that they can be extended to. The following result formalizes this intuition, where $i^*$ is the dominant direction:

\begin{restatable}{lem}{lemroneneighborhoodlowcommon}\label{lem:r-1-neighborhood-low-common}
For any $r\ge 1$, $r$-round matching-certified algorithm $f$, $x\in \mc R_{r-1}$, and $\delta_1,\delta_2,\hdots,\delta_{\Delta}\in [0,1/2]$, there exists an $i^*\in [\Delta]$ for which

$$\sum_{i\in [\Delta], i\ne i^*} P_i(f,x,\delta_i) \le F_1(\delta_{i^*})$$

where $F_1(\delta) := 6e^4\delta$.
\end{restatable}

We prove Lemma~\ref{lem:r-1-neighborhood-low-common} in Section~\ref{subsec:r-1-neighborhood-low} (which is solely dedicated to proving this lemma). We do not directly use Lemma \ref{lem:r-1-neighborhood-low-common}. Instead, we apply two immediate consequences of Lemma~\ref{lem:r-1-neighborhood-low-common}: Proposition \ref{prop:r-1-neighborhood-low-weak} and Lemma \ref{lem:r-1-neighborhood-low-strong}. We derive the first consequence of Lemma \ref{lem:r-1-neighborhood-low-common} by setting all $\delta_i$s equal:

\begin{prop}\label{prop:r-1-neighborhood-low-weak}
For any $r$-round matching-certified algorithm $f$, $x\in \mc R_{r-1}$, and $\delta\in [0,1/2]$,

$$P_{\text{comp}}(f,x,\delta) \le F_1(\delta)=6e^4\delta$$
\end{prop}

\begin{proof}
Follows immediately from Lemma \ref{lem:r-1-neighborhood-low-common}, with $\delta_i$ set to $\delta$ for all $i\in [\Delta]$ and using the fact that $P_{\text{comp}}(f,x,\delta) \le P(f,x,\delta) - P_{i^*}(f,x,\delta)$.
\end{proof}

To our knowledge, this lemma is only strong enough to obtain an $\Omega(\log\log\Delta)$-round hardness result. To obtain the desired $\Omega(\log\Delta)$-round hardness result, we need a neighborhood-dependent bound. This bound depends on just how dominant the dominant direction is:

\begin{defn}[Very Dominant $(r-1)$-Neighborhoods]
\seri{Handling constant here is very important.}For any $r\ge 1$, $r$-round matching-certified algorithm $f$, and $\delta\in [0,1]$, define $\mc R_{r-1}(f,\delta)$ to be the set of all $x\in \mc R_{r-1}$ for which

$$P_{\max}(f,x,\delta) \ge C_4$$

where $C_4 := 99/100$. Furthermore, for $x\in \mc R_{r-1}$, define $\delta_{\text{dom}}(f,x) \in [0,1]$ to be the infimum\footnote{We use infimum instead of minimum here because $P_{\max}(f,x,\delta)$ is not necessarily a continuous function of $\delta$. Note that this infimum is well-defined as it is taken over a nonempty set, because $P_{\max}(f,x,1) = 1$ for any $x\in \mc R_{r-1}$.} of the set of values $\delta \in [0,1]$ for which

$$x\in \mc R_{r-1}(f,\delta)$$

\end{defn}

We are able to obtain a bound on $P_{\text{comp}}(f,x,C_5)$ for some constant $C_5$ in terms of $\delta_{\text{dom}}(f,x)$. Note that this is much smaller than the $O(C_5)$ bound that would be obtained by applying Proposition \ref{prop:r-1-neighborhood-low-weak}. We obtain this from Lemma \ref{lem:r-1-neighborhood-low-common} by using distinct values of $\delta_i$ for different $i$:

\begin{lem}\label{lem:r-1-neighborhood-low-strong}
For any $r$-round matching-certified algorithm $f$ and $x\in \mc R_{r-1}$

$$P_{\text{comp}}(f,x,C_5) \le F_1(\delta_{\text{dom}}(f,x))=6e^4\delta_{\text{dom}}(f,x)$$

where $C_5 := e^{-3}/1000$.
\end{lem}

\begin{proof}
It suffices to show that

$$P_{\text{comp}}(f,x,C_5) \le F_1(\delta)$$

for any $\delta > \delta_{\text{dom}}(f,x)$. Let $i' \in [\Delta]$ be an arbitrary index for which $P_{i'}(f,x,\delta) \ge C_4$. Define $\delta_{i'} := \delta$ and $\delta_j := C_5$ for all $j\ne i'$. By Lemma \ref{lem:r-1-neighborhood-low-common}, there exists an $i^*\in [\Delta]$ for which

$$\sum_{j\in [\Delta], j\ne i^*} P_j(f,x,\delta_j) \le F_1(\delta_{i^*})$$

First, we show that $i^* = i'$. Indeed, if $i^* \ne i'$, then the sum includes $j = i'$ and $\delta_{i^*} = C_5$, meaning that

$$C_4 > F_1(C_5) = F_1(\delta_{i^*}) \ge \sum_{j\in [\Delta], j\ne i^*} P_j(f,x,\delta_j) \ge P_{i'}(f,x,\delta) \ge C_4$$

which is a contradiction. Thus, $i^* = i'$, which implies that

$$P_{\text{comp}}(f,x,C_5) \le \sum_{j\ne i^*} P_j(f,x,C_5) \le F_1(\delta_{i^*}) = F_1(\delta_{i'}) = F_1(\delta)$$

as desired.
\end{proof}

In Section \ref{sec:complement-prob}, we show that most neighborhoods have a very dominant direction:

\begin{restatable}{prop}{propcomplementprob}\label{prop:complement-prob}
For $r\ge 1$ and any $r$-round matching-certified algorithm $f$,

$$\E_{x\sim \mc R_{r-1}}[\delta_{\text{dom}}(f,x)] \le C_{11} P_f$$

where $C_{11} := 168(10)^4e^4$.
\end{restatable}

\subsection{Completing the Round Elimination Result (Proof of Lemma~\ref{lem:vertex-round-elim})}\label{sec:prooflemmaRE}

Now, we are ready to define the $(r-1)$-round matching-certified algorithm $g$ given the $r$-round matching-certified algorithm $f$. For $y\in \mc F_{r-1}$, let $g(y) = 1$ if and only if all of the following hold:

\begin{enumerate}
    \item (First Condition) $y\in \mc X_{r-1}(f,C_5)$
    \item (Second Condition)\footnote{The 1s are important here, unlike in the definition of good flowers. They cannot be replaced by arbitrary members of $[\Delta]$. They appear solely due to the definitions of $\text{end}_A(y)$ and $\text{end}_B(y)$. Specifically, $y$ is an extension of the $(r-1)$-neighborhood $\text{end}_A(y)$ in the 1 direction. If they were extensions in other directions, different numbers would appear here.} $i_{\max}(f,\text{end}_A(y),C_5) = 1$ and $i_{\max}(f,\text{end}_B(y),C_5) = 1$.
\end{enumerate}

The second condition ensures that $g$ is an $(r-1)$-round matching-certified algorithm by only letting $g(y) = 1$ if the $i_{\max}$s on both sides point towards $y$. Thus, it suffices to upper bound the vertex survival probability of $g$; i.e. $P_g$. The key step is to show that the second condition of $g$ does not noticeably degrade the high inclusion probability (Proposition \ref{prop:r-1-flower-high-weak}) of the first condition. We can upper bound the degradation (probability that an $(r-1)$-flower is good and that $i_{\max}$ from one side does not point towards it) by sampling the $(r-1)$ flower a different way: sample an $(r-1)$-neighborhood, pick a random direction $i$, and sample an extension in that direction (Proposition \ref{prop:r-to-f-sample}). Given an $(r-1)$-neighborhood $x$, the probability that an extension in a random direction is $C_5$-good and not in the $i_{\max}$ direction (first condition but not second) is at most $\frac{1}{\Delta}P_{\text{comp}}(f,x,C_5)$ (average of non-$i_{\max}$ directional extension probabilities). Thus, randomizing over $x$ yields an overall probability by taking the expectation, which is bounded using Proposition \ref{prop:complement-prob}. We are now ready to prove Lemma \ref{lem:vertex-round-elim}, which we restate for convenience:

\lemvertexroundelim*

\begin{proof}

\textbf{Matching-certified:}

Consider two incident $(r-1)$-flowers $y,y'\in \mc F_{r-1}$. By definition, there exist $v,v'\in \{A,B\}$ and a permutation $\sigma$ for which $\text{end}_v(y) = \sigma(\text{end}_{v'}(y'))$ and $\sigma(1)\ne 1$. Suppose for the sake of contradiction that $g(y) = g(y') = 1$. By the second condition, $i_{\max}(f,\text{end}_{v'}(y'),C_5) = 1$ and $i_{\max}(f,\sigma(\text{end}_{v'}(y')),C_5) = i_{\max}(f,\text{end}_v(y),C_5) = 1$. But this means that $i_{\max}(f,\text{end}_{v'}(y'),C_5) = \sigma^{-1}(i_{\max}(f,\sigma(\text{end}_{v'}(y')),C_5)) = \sigma^{-1}(1) \ne 1$, a contradiction. Thus, $g$ is matching-certified.\\

\textbf{Vertex survival probability:}
We start by showing that

$$\Pr_{y\sim \mc F_{r-1}}\left[y\in \mc X_{r-1}(f,C_5) \text{ and } i_{\max}(f,\text{end}_A(y),C_5) \ne 1 \right] \le \frac{12e^4C_{11} P_f}{\Delta}$$

To do this, start by applying Proposition \ref{prop:r-to-f-sample} to the random variable $\textbf{1}_{y\in \mc X_{r-1}(f,C_5)}\textbf{1}_{i_{\max}(f,\text{end}_A(y),C_5)\ne 1}$ and apply the definition of $P_i(f,x,C_5)$ to show that

\begin{align*}
&\Pr_{y\sim \mc F_{r-1}}\left[y\in \mc X_{r-1}(f,C_5) \text{ and } i_{\max}(f,\text{end}_A(y),C_5) \ne 1 \right]\\
&= \E_{x\sim \mc R_{r-1}}\left[\E_{i\sim [\Delta]}\left[\Pr_{y\sim \mc F_{r-1}}\left[y\in \mc X_{r-1}(f,C_5) \text{ and } i_{\max}(f,\text{end}_A(y),C_5) \ne 1\mid \text{end}_A(y) = \sigma_i(x)\right]\right]\right]\\
&= \E_{x\sim \mc R_{r-1}}\left[\E_{i\sim [\Delta]}\left[\textbf{1}_{i_{\max}(f,\sigma_i(x),C_5)\ne 1}\Pr_{y\sim \mc F_{r-1}}\left[y\in \mc X_{r-1}(f,C_5)\mid \text{end}_A(y) = \sigma_i(x)\right]\right]\right]\\
&= \E_{x\sim \mc R_{r-1}}\left[\E_{i\sim [\Delta]}\left[\textbf{1}_{i_{\max}(f,x,C_5)\ne i}P_i(f,x,C_5)\right]\right]\\
\end{align*}

When $i_{\max}(f,x,C_5) \ne 0$, $P_{i_{\max}(f,x,C_5)}(f,x,C_5) = P_{\max}(f,x,C_5)$ so

$$\E_{i\sim [\Delta]}[\textbf{1}_{i_{\max}(f,x,C_5) \ne i}P_i(f,x,C_5)] = \frac{1}{\Delta}\left(\left(\sum_{i=1}^{\Delta} P_i(f,x,C_5)\right) - P_{\max}(f,x,C_5)\right) = \frac{1}{\Delta}P_{\text{comp}}(f,x,C_5)$$

When $i_{\max}(f,x,C_5) = 0$, the maximizing index is not unique, so $P_{\max}(f,x,C_5) \le P_{\text{comp}}(f,x,C_5)$ and

$$\E_{i\sim [\Delta]}[\textbf{1}_{i_{\max}(f,x,C_5) \ne i}P_i(f,x,C_5)] = \frac{1}{\Delta}P(f,x,C_5) \le \frac{2}{\Delta}P_{\text{comp}}(f,x,C_5)$$

Thus, combining both cases yields

$$\E_{x\sim \mc R_{r-1}}\left[\E_{i\sim [\Delta]}\left[\textbf{1}_{i_{\max}(f,x,C_5)\ne i}P_i(f,x,C_5)\right]\right] \le \frac{2}{\Delta}\E_{x\sim \mc R_{r-1}}\left[P_{\text{comp}}(f,x,C_5)\right]$$

By Lemma \ref{lem:r-1-neighborhood-low-strong}, $P_{\text{comp}}(f,x,C_5) \le 6e^4\delta_{\text{dom}}(f,x)$ for any $x\in \mc R_{r-1}$. Substitution and Proposition \ref{prop:complement-prob} show that

\begin{align*}
\Pr_{y\sim \mc F_{r-1}}\left[y\in \mc X_{r-1}(f,C_5) \text{ and } i_{\max}(f,\text{end}_A(y),C_5) \ne 1 \right] &\le \frac{2}{\Delta}\E_{x\sim \mc R_{r-1}}\left[P_{\text{comp}}(f,x,C_5)\right]\\
&\le \frac{12e^4}{\Delta}\E_{x\sim \mc R_{r-1}}\left[\delta_{\text{dom}}(f,x)\right]\\
&\le \frac{12e^4C_{11} P_f}{\Delta}
\end{align*}

This is the desired first step. Now, we use this to obtain an inclusion probability lower bound for $g$ via a union bound, the observation that $A$ and $B$ are interchangeable, and Proposition \ref{prop:r-1-flower-high-weak}:

\begin{align*}
\Pr_{y\sim \mc F_{r-1}}[g(y) = 1] &\ge \Pr_{y\sim \mc F_{r-1}}[y\in \mc X_{r-1}(f,C_5)]\\
&- \Pr_{y\sim \mc F_{r-1}}[y\in \mc X_{r-1}(f,C_5) \text{ and } i_{\max}(f,\text{end}_A(y),C_5) \ne 1]\\
&- \Pr_{y\sim \mc F_{r-1}}[y\in \mc X_{r-1}(f,C_5) \text{ and } i_{\max}(f,\text{end}_B(y),C_5) \ne 1]\\
&= \Pr_{y\sim \mc F_{r-1}}[y\in \mc X_{r-1}(f,C_5)]\\
&- 2\Pr_{y\sim \mc F_{r-1}}[y\in \mc X_{r-1}(f,C_5) \text{ and } i_{\max}(f,\text{end}_A(y),C_5) \ne 1]\\
&\ge \frac{1 - 2P_f/C_5}{\Delta} - \frac{24e^4C_{11}P_f}{\Delta}\\
&\ge \frac{1 - 48e^4C_{11}P_f}{\Delta} 
\end{align*}

Thus,

$$P_g = 1 - \Delta \Pr_{x\sim \mc F_{r-1}}[g(x) = 1] \le 48e^4C_{11}P_f \le C_1 P_f$$

as desired.
\end{proof}

\subsection{Good $(r-1)$-Flower Proofs (Proofs of Proposition \ref{prop:r-1-flower-high-weak} and Lemma \ref{lem:r-1-flower-high-strong})}\label{subsec:r-1-flower-high}

In this section, we show that the set of good $(r-1)$-flowers has high probability mass. We do this by rewriting the vertex survival probability of the algorithm $f$ in terms of the expectation of functions of a certain random variable $Q(f,x)$, and using a Markov-like argument on related random variables to show that many $(r-1)$-flowers are good. We prove the two desired results of this section over the course of a few subsections:

\begin{enumerate}
    \item In Section \ref{sec:FlowerQ}, we introduce a random variable $Q(f,x)$ for each $x\in \mc F_{r-1}$. $Q(f,x)$ is equal to the $A$-side probability that is used to define the set of good flowers $\mc X_{r-1}(f,\delta)$. We then prove two moment-like bounds on $Q(f,x)$ -- Propositions \ref{prop:one-side} and \ref{prop:two-side} -- which extract all of the information that we need from matching-certified algorithms in order to prove the desired results (Proposition \ref{prop:r-1-flower-high-weak} and Lemma \ref{lem:r-1-flower-high-strong}).
    \item Intuitively, the moment-like bounds show that $Q(f,x)$ is almost always close to 0 or 1. In Section \ref{sec:FlowerWeak}, we give a Markov-like argument to formalize this intuition and prove Proposition \ref{prop:r-1-flower-high-weak}.
    \item In Section \ref{sec:FlowerStrong}, we do a similar argument to show Lemma \ref{lem:r-1-flower-high-strong}.
\end{enumerate}

We provide a more detailed overview at the beginning of each subsection.

\subsubsection{Defining $Q(f,x)$ And Proving Moment-Like Bounds}\label{sec:FlowerQ}

For $x\in \mc F_{r-1}$ and an $r$-round matching-certified algorithm $f$, let

$$Q(f,x) := \Pr_{y\sim \mc R_r}[\text{dir}(f,y) = 1 \mid \text{res}_1(y) = x]$$

That is, $Q(f,x)$ is the probability that an $r$-neighborhood extension $y$ of the $(r-1)$-flower $x$ has its $\text{dir}(f,y)$ pointing towards $x$.\footnote{Recall that the choice of 1 in this definition is not important due to Fact 1 -- it can be replaced with any $i\in [\Delta]$.} This definition allows us to state the definition of good flowers more succinctly. In particular, for $\delta\in [0,1]$, $\mc X_{r-1}(f,\delta)$ is the set of $x\in \mc F_{r-1}$ for which $Q(f,x) \ge 1-\delta$ and $Q(f,\ol{x}) \ge 1-\delta$.

We prove two moment-like bounds on the random variable $Q(f,x)$: Propositions \ref{prop:one-side} and \ref{prop:two-side}. These bounds aim to abstract away all that we need to know about matching-certified algorithms for the purposes of proving the main results (Proposition \ref{prop:r-1-flower-high-weak} and Lemma \ref{lem:r-1-flower-high-strong}). Furthermore, note that one of these moment-like bounds is an upper bound while the other is a lower bound -- this has the effect of sandwiching $Q(f,x)$ strongly enough to show that it is generally close to either 0 or 1.

The first moment-like bound states that the expected probability that an $r$-neighborhood extension points towards $x$ is at most $1/\Delta$, which holds due to the exclusivity of dir (Proposition \ref{prop:r-ngbr-exclusive}):

\begin{prop}\label{prop:one-side}
For any $r$-round matching-certified algorithm $f$,

$$\E_{x\sim \mc F_{r-1}}[Q(f,x)] \le \frac{1}{\Delta}$$
\end{prop}

\begin{proof}
By definition of $Q(f,x)$ and Proposition \ref{prop:f-to-r-sample} for the random variable $\textbf{1}_{\text{dir}(f,y) = 1}$ and $i\gets 1$,

$$\E_{x\sim \mc F_{r-1}}[Q(f,x)] = \E_{x\sim \mc F_{r-1}}\left[\Pr_{y\sim \mc R_r}[\text{dir}(f,y) = 1 \mid \text{res}_1(y) = x]\right] = \Pr_{y\sim \mc R_r}[\text{dir}(f,y) = 1]$$

By Proposition \ref{prop:vert-permute} applied to the same random variable and the uniform distribution over permutations $\sigma_i$ for $i\sim [\Delta]$,

$$\Pr_{y\sim \mc R_r}[\text{dir}(f,y) = 1] = \E_{y\sim \mc R_r}\left[\Pr_{i\sim [\Delta]}[\text{dir}(f,\sigma_i(y)) = 1]\right] = \E_{y\sim \mc R_r}\left[\Pr_{i\sim [\Delta]}[\text{dir}(f,y) = i]\right] \le \E_{y\sim \mc R_r}[1/\Delta] = 1/\Delta$$ as desired.
\end{proof}

Our second moment-like bound lower bounds the product of the two directional extension probabilities for typical $(r-1)$-flowers. This complements the first moment-like bound that we had, which was an upper bound. We prove this lower bound in two steps: first, show that an $f$ accepts an $r$-flower $y$ with probability governed by its vertex survival probability; then, upper bound $f$'s $r$-flower acceptance probability in terms of the probability that dir on both sides points towards $y$, which can be bounded because the two dir-pointing probabilities are independent due to being functions of randomness on non-overlapping edges sets (two different sides):

\begin{prop}\label{prop:two-side}
For any $r$-round matching-certified algorithm $f$,

$$\E_{x\sim \mc F_{r-1}}[Q(f,x)Q(f,\ol{x})] \ge \frac{1-P_f}{\Delta}$$
\end{prop}

\begin{proof}
For any $y\in \mc F_r$, $f(y) = 1$ implies that $\text{dir}(f,\text{end}_A(y)) = 1$ and $\text{dir}(f,\text{end}_B(y)) = 1$ by definition of $\text{dir}$. Applying Proposition \ref{prop:f-to-f-sample} for the random variables $X_A(z) := X_B(z) := \textbf{1}_{\text{dir}(f,z) = 1}$ shows that

\begin{align*}
\Pr_{y\sim \mc F_r}[f(y) = 1] &\le \Pr_{y\sim \mc F_r}[\text{dir}(f,\text{end}_A(y)) = 1 \text{ and } \text{dir}(f,\text{end}_B(y)) = 1]\\
&= \E_{x\sim \mc F_{r-1}}\left[\Pr_{z_A\sim \mc R_r}[\text{dir}(f,z_A) = 1 \mid \text{res}_1(z_A) = x]\Pr_{z_B\sim \mc R_r}[\text{dir}(f,z_B) = 1 \mid \text{res}_1(z_B) = \ol{x}]\right]\\
&= \E_{x\sim \mc F_{r-1}}[Q(f,x)Q(f,\ol{x})]
\end{align*}

$\displaystyle\Delta\Pr_{y\sim \mc F_r}[f(y) = 1] = 1 - P_f$ because the $\Delta$ different matching events in $\mc R_{r+1}$ are mutually exclusive. Thus, we have the desired result.
\end{proof}

\subsubsection{A Lower Bound on the Fraction of Good Flowers (Proposition \ref{prop:r-1-flower-high-weak})}\label{sec:FlowerWeak}

We now use the moment-like bounds from the previous section to prove Proposition \ref{prop:r-1-flower-high-weak} via a Markov-inequality-like strategy, which we now overview. Recall that the desired result is as follows:

\proproneflowerhighweak*

We now give intuition for the proof in stages:\\

\textbf{Wishful Thinking}: Temporarily assume that $Q(f,x) = Q(f,\ol{x})$ for all $x\in \mc F_{r-1}$. We discuss how to get rid of this assumption at the end of this intuition discussion.\\

\textbf{Attempt 1: Actually Use Markov's Inequality}: We now discuss a proof attempt to illustrate how Markov's Inequality is relevant to the desired result. Let $\mc D$ denote the non-uniform distribution over $\mc F_{r-1}$ obtained by sampling $x\in \mc F_{r-1}$ with probability proportional to $Q(f,x)$. Dividing the Proposition \ref{prop:one-side} and \ref{prop:two-side} bounds, complementing, and applying the wishful assumption shows that

$$\E_{x\sim \mc D}[1-Q(f,x)] = \frac{\displaystyle\E_{x\sim \mc F_{r-1}}[Q(f,x)(1 - Q(f,x))]}{\displaystyle\E_{x\sim \mc F_{r-1}}[Q(f,x)]} \le 1 - \frac{(1-P_f)/\Delta}{1/\Delta} = P_f$$

Since $Q(f,x)$ is a probability, it is always at most 1, meaning that $1 - Q(f,x)$ is a nonnegative random variable to which Markov's Inequality can be applied. Thus,

$$\Pr_{x\sim \mc D}[1 - Q(f,x) > \delta] \le \frac{P_f}{\delta}$$

Rearranging, using the wishful assumption, and using the definition of $\mc X_{r-1}(f,\delta)$ shows that

$$\Pr_{x\sim \mc D}[x\in \mc X_{r-1}(f,\delta)] \ge 1 - \frac{P_f}{\delta}$$

Thus, a $(1 - P_f/\delta)$-fraction of this conditional distribution has the desired property. By Proposition \ref{prop:two-side} and the wishful assumption, $\E_{x\sim \mc F_{r-1}}[Q(f,x)] \ge \E_{x\sim \mc F_{r-1}}[Q(f,x)^2] \ge (1-P_f)/\Delta$. Intuitively speaking, a $(1 - P_f/\delta)$-fraction of the support of the distribution $\mc D$ is good, and a $(1 - P_f)/\Delta$-fraction of the uniform distribution is in the support of $\mc D$, so taking the product would lead a bound of $\Pr_{x\sim \mc F_{r-1}}[x\in \mc X_{r-1}(f,\delta)] \ge (1 - P_f/\delta)(1 - P_f)/\Delta \ge \frac{1 - 2P_f/\delta}{\Delta}$.\\

\textbf{Attempt 2: A Direct Approach}: Attempt 1 has two issues: (a) it is not immediately clear how to make the product reasoning at the end work and (b) we need to get rid of the wishful assumption. We start by addressing (a). We lower bound the desired probability by means of a proxy and the wishful assumption:

$$\Pr_{x\sim \mc F_{r-1}}[x \in \mc X_{r-1}(f,\delta)] \ge \E_{x\sim \mc F_{r-1}}[Q(f,x)\textbf{1}_{x\in \mc X_{r-1}(f,\delta)}] = \E_{x\sim \mc F_{r-1}}[Q(f,x)\textbf{1}_{Q(f,x) \ge 1-\delta}]$$

\noindent Define the quantities $$s_{\ge} := \E_{x\sim \mc F_{r-1}}[Q(f,x)\textbf{1}_{Q(f,x) \ge 1-\delta}]$$ and $$s_{<} := \E_{x\sim \mc F_{r-1}}[Q(f,x)\textbf{1}_{Q(f,x) < 1-\delta}]$$ We just need to obtain a lower bound on $s_{\ge}$. We start doing that by noticing two constraints that follow from these definitions and Propositions \ref{prop:one-side}, \ref{prop:two-side}:

$$s_{\ge} + s_{<} = \E_{x\sim \mc F_{r-1}}[Q(f,x)] \le \frac{1}{\Delta}$$

\begin{align*}
s_{\ge} + (1 - \delta)s_{<} &= \E_{x\sim \mc F_{r-1}}[Q(f,x)\textbf{1}_{Q(f,x)\ge 1-\delta}] + (1 - \delta)\E_{x\sim \mc F_{r-1}}[Q(f,x)\textbf{1}_{Q(f,x)< 1-\delta}]\\
&\ge \E_{x\sim \mc F_{r-1}}[Q(f,x)^2\textbf{1}_{Q(f,x)^2\ge 1-\delta}] + \E_{x\sim \mc F_{r-1}}[Q(f,x)^2\textbf{1}_{Q(f,x)< 1-\delta}]\\
&= \E_{x\sim \mc F_{r-1}}[Q(f,x)^2]\\
&\ge \frac{1-P_f}{\Delta}
\end{align*}

The larger contribution of $s_{\ge}$ than $s_{<}$ to the second constraint makes it possible to obtain a lower bound on $s_{\ge}$:

$$s_{\ge} = \frac{s_{\ge} + (1-\delta)s_{<} - (1-\delta)(s_{\ge} + s_{<})}{\delta} \ge \frac{(1-P_f)/\Delta - (1-\delta)(1/\Delta)}{\delta} = \frac{1 - P_f/\delta}{\Delta}$$

Therefore, $\displaystyle\Pr_{x\sim \mc F_{r-1}}[x \in \mc X_{r-1}(f,\delta)] \ge s_{\ge} \ge \frac{1 - P_f/\delta}{\Delta}$ as desired.\\

\textbf{Removing the Wishful Assumption}: Attempt 2 dealt with issue (a), so now we deal with issue (b). Attempt 2 developed constraints on the desired quantity by writing the expectations in Propositions \ref{prop:one-side} and \ref{prop:two-side} as sums of two quantities. To remove the wishful assumption, we instead write those expectations as sums of \emph{four} quantities $s_{\ge\ge}$, $s_{\ge<}$, $s_{<\ge}$, and $s_{<<}$, depending on whether or not each of $Q(f,x)$ \emph{and} $Q(f,\ol{x})$ rather than just $Q(f,x)$ are above or below the threshold $1-\delta$. To prove the desired result, we need to lower bound $s_{\ge\ge}$. The analogue of the second constraint from Attempt 2 ends up being

$$s_{\ge\ge} + s_{<\ge} + (1-\delta)s_{\ge<} + (1-\delta)s_{<<} \ge \frac{1-P_f}{\Delta}$$

Unlike in the two-quantity case, the term-to-bound $s_{\ge\ge}$ does not contribute more to this sum than any other term; in particular $s_{<\ge}$. We deal with this complication by noticing that $s_{<\ge} \le s_{\ge<}$. Intuitively, one would expect this because each of these quantities are expectations over thresholded $Q(f,x)$s, and the $s_{<\ge}$ terms have $Q(f,x) < 1-\delta$ while the $s_{\ge<}$ terms have $Q(f,x) \ge 1-\delta$, a larger number. Thus, the lower coefficient on $s_{\ge<}$ compensates for the higher coefficient on $s_{<
\ge}$ to obtain a more useful second constraint:

\begin{align*}
&s_{\ge\ge} + (1-\delta/2)s_{<\ge} + (1-\delta/2)s_{\ge<} + (1-\delta)s_{<<}\\
&= s_{\ge\ge} + s_{<\ge} + (1-\delta)s_{\ge<} + (1-\delta)s_{<<} + (\delta/2)(s_{\ge<} - s_{<\ge})\\
&\ge s_{\ge\ge} + s_{<\ge} + (1-\delta)s_{\ge<} + (1-\delta)s_{<<}\\
&\ge \frac{1-P_f}{\Delta}
\end{align*}

Now, the term-to-bound $s_{\ge\ge}$'s coefficient dominates all other coefficients, so we can get the desired bound using a similar proof to Attempt 2. This completes the intuition discussion. We now restate the result once more for convenience and prove it:

\proproneflowerhighweak*

\begin{proof}
For any pair of symbols $\sigma\in \{\ge, <\}$ and $\tau\in \{\ge, <\}$, let $\mc Y_{\sigma\tau} := \{x\in \mc F_{r-1} : Q(f,x)\text{ }\sigma\text{ }(1 - \delta) \text{ and } Q(f,\ol{x})\text{ }\tau\text{ }(1 - \delta)\}$ and let $$s_{\sigma\tau} := \E_{x\sim \mc F_{r-1}}[Q(f,x)\textbf{1}_{x\in \mc Y_{\sigma\tau}}]$$ 

We just need to lower bound $s_{\ge\ge}$, as

$$\Pr_{x\sim \mc F_{r-1}}[x\in \mc X_{r-1}(f,\delta)] = \E_{x\sim \mc F_{r-1}}[\textbf{1}_{x\in \mc Y_{\ge\ge}}] \ge \E_{x\sim \mc F_{r-1}}[Q(f,x)\textbf{1}_{x\in \mc Y_{\ge\ge}}] = s_{\ge\ge}$$

We now derive three constraints on the $s$ variables. First, Proposition \ref{prop:one-side} shows that

\begin{align*}
s_{\ge\ge} + s_{<\ge} + s_{\ge<} + s_{<<} &= \E_{x\sim \mc F_{r-1}}[Q(f,x)\textbf{1}_{x\in \mc Y_{\ge\ge}}] + \E_{x\sim \mc F_{r-1}}[Q(f,x)\textbf{1}_{x\in \mc Y_{<\ge}}] + \\
&\E_{x\sim \mc F_{r-1}}[Q(f,x)\textbf{1}_{x\in \mc Y_{\ge<}}] + \E_{x\sim \mc F_{r-1}}[Q(f,x)\textbf{1}_{x\in \mc Y_{<<}}]\\
&= \E_{x\sim \mc F_{r-1}}[Q(f,x)]\\
&\le \frac{1}{\Delta}
\end{align*}

Secondly, Proposition \ref{prop:two-side} shows that

\begin{align*}
s_{\ge\ge} + s_{<\ge} + (1 - \delta)s_{\ge<} + (1 - \delta)s_{<<} &= \E_{x\sim \mc F_{r-1}}[Q(f,x)\textbf{1}_{x\in \mc Y_{\ge\ge}}] + \E_{x\sim \mc F_{r-1}}[Q(f,x)\textbf{1}_{x\in \mc Y_{<\ge}}] + \\
&(1-\delta)\E_{x\sim \mc F_{r-1}}[Q(f,x)\textbf{1}_{x\in \mc Y_{\ge<}}] + (1-\delta)\E_{x\sim \mc F_{r-1}}[Q(f,x)\textbf{1}_{x\in \mc Y_{<<}}]\\
&\ge \E_{x\sim \mc F_{r-1}}[Q(f,x)Q(f,\ol{x})\textbf{1}_{x\in \mc Y_{\ge\ge}}] + \E_{x\sim \mc F_{r-1}}[Q(f,x)Q(f,\ol{x})\textbf{1}_{x\in \mc Y_{<\ge}}] + \\
&\E_{x\sim \mc F_{r-1}}[Q(f,x)Q(f,\ol{x})\textbf{1}_{x\in \mc Y_{\ge<}}] + \E_{x\sim \mc F_{r-1}}[Q(f,x)Q(f,\ol{x})\textbf{1}_{x\in \mc Y_{<<}}]\\
&= \E_{x\sim \mc F_{r-1}}[Q(f,x)Q(f,\ol{x})]\\
&\ge \frac{1 - P_f}{\Delta}
\end{align*}

Thirdly, Proposition \ref{prop:edge-flip} applied to the random variable $Q(f,x)\textbf{1}_{x\in \mc Y_{<\ge}}$ (third equality) shows that

\begin{align*}
s_{\ge<} &= \E_{x\sim \mc F_{r-1}}[Q(f,x)\textbf{1}_{x\in \mc Y_{\ge<}}]\\
&\ge \E_{x\sim \mc F_{r-1}}[Q(f,\ol{x})\textbf{1}_{x\in \mc Y_{\ge<}}]\\
&= \E_{x\sim \mc F_{r-1}}[Q(f,\ol{x})\textbf{1}_{\ol{x}\in \mc Y_{<\ge}}]\\
&= \E_{x\sim \mc F_{r-1}}[Q(f,x)\textbf{1}_{x\in \mc Y_{<\ge}}]\\
&= s_{<\ge}
\end{align*}

Combining the second and third constraints produces a more useful second constraint, as described in the intuition:

\begin{align*}
&s_{\ge\ge} + (1-\delta/2)s_{<\ge} + (1-\delta/2)s_{\ge<} + (1-\delta)s_{<<}\\
&= s_{\ge\ge} + s_{<\ge} + (1-\delta)s_{\ge<} + (1-\delta)s_{<<} + (\delta/2)(s_{\ge<} - s_{<\ge})\\
&\ge s_{\ge\ge} + s_{<\ge} + (1-\delta)s_{\ge<} + (1-\delta)s_{<<}\\
&\ge \frac{1-P_f}{\Delta}
\end{align*}

In particular,

$$s_{\ge\ge} + (1-\delta/2)(s_{<\ge} + s_{\ge<} + s_{<<}) \ge \frac{1-P_f}{\Delta}$$

Thus,

\begin{align*}
s_{\ge\ge} &= \frac{(s_{\ge\ge} + (1-\delta/2)(s_{<\ge} + s_{\ge<} + s_{<<})) - (1-\delta/2)(s_{\ge\ge} + s_{<\ge} + s_{\ge<} + s_{<<})}{\delta/2}\\
&\ge \frac{(1-P_f)/\Delta - (1-\delta/2)/\Delta}{\delta/2}\\
&= \frac{1-2P_f/\delta}{\Delta}
\end{align*}

and $\displaystyle\Pr_{x\sim \mc F_{r-1}}[x\in \mc X_{r-1}(f,\delta)] \ge s_{\ge\ge} \ge \frac{1 - 2P_f/\delta}{\Delta}$ as desired.
\end{proof}

\subsubsection{Proof of Lemma \ref{lem:r-1-flower-high-strong}}\label{sec:FlowerStrong}

In this section, we prove a stronger version of Proposition \ref{prop:r-1-flower-high-weak} (up to the constant factor 7 vs 2):

\lemroneflowerhighstrong*

We now provide some intuition for this result by comparing it to Proposition \ref{prop:r-1-flower-high-weak}, giving a tight algorithm $f$ for Proposition \ref{prop:r-1-flower-high-weak}, pointing out the slack in the proof of Proposition \ref{prop:r-1-flower-high-weak}, and describing the way that we improve the argument to prove Lemma \ref{lem:r-1-flower-high-strong}.\\

\textbf{Deriving (an inconsequentially-weaker) Proposition \ref{prop:r-1-flower-high-weak} from Lemma \ref{lem:r-1-flower-high-strong}}: The function $h(\tau) := \displaystyle\Pr_{x\sim \mc F_{r-1}}[x\in \mc X_{r-1}(f,\tau)]$ is nondecreasing as a function of $\tau$. Therefore, Proposition \ref{prop:r-1-flower-high-weak} follows immediately from using the fact that maximum is greater than or equal to expectation (with $(1-2P_f/\xi)/\Delta$ replaced by the slightly lower $(1-7P_f/\xi)/\Delta$).\\

\textbf{Deriving a consequentially-weaker version of Lemma \ref{lem:r-1-flower-high-strong} from Proposition \ref{prop:r-1-flower-high-weak}}: Substituting Proposition \ref{prop:r-1-flower-high-weak} into the statement of Lemma \ref{lem:r-1-flower-high-strong} yields\footnote{The first approximation comes from the idea that $\xi$ smaller than the typical (from Proposition \ref{prop:two-side}) value $P_f$ of $1-Q(f,x)$ do not need to be considered and the second approximation follows from integration.}

$$\E_{\tau\sim [0,\xi]}\left[\Pr_{x\sim \mc F_{r-1}}[x\in \mc X_{r-1}(f,\tau)]\right] \ge \E_{\tau\sim [0,\xi]}\left[\frac{1-2P_f/\tau}{\Delta}\right] \approx \E_{\tau\sim [P_f,\xi]}\left[\frac{1-2P_f/\tau}{\Delta}\right] \approx \frac{1 - P_f\log(\xi/P_f)}{\Delta}$$

Plugging this into the rest of the argument towards Lemma \ref{lem:vertex-round-elim} would lead to a weaker vertex survival probability guarantee of $P_g \le O(P_f\log(\xi/P_f)) \approx O(P_f\log(1/P_f))$, which would yield an $\Omega\left(\min\left(\frac{\log\Delta}{\log\log\Delta}, \sqrt{\frac{\log n}{\log\log n}}\right)\right)$-round lower bound, merely matching KMW \cite{KuhnMW04} rather than exceeding it.\\

\textbf{Tight instances for Proposition \ref{prop:r-1-flower-high-weak}}: Towards actually proving Lemma \ref{lem:r-1-flower-high-strong}, we examine situations in which Proposition \ref{prop:r-1-flower-high-weak} is tight, which will help illustrate why the inequality from the $\log(1/P_f)$ proof is not tight. For some $\delta\in (0,1)$, consider an algorithm $f$ that satisfies the following properties:\footnote{An algorithm with these properties (and vertex survival probability $P_f$) may not exist, but that would only make an argument easier.}

\begin{enumerate}
    \item With probability $\frac{1-P_f/(6\delta)}{\Delta(1-P_f/3)}$, $x\sim \mc R_{r-1}$ has $Q(f,x) = Q(f,\ol{x}) = 1-P_f/3$.
    \item With probability $\frac{P_f/(6\delta)}{\Delta(1-2\delta)}$, $x\sim \mc R_{r-1}$ has $Q(f,x) = Q(f,\ol{x}) = 1-2\delta$.
    \item Otherwise, $Q(f,x) = Q(f,\ol{x}) = 0$.
\end{enumerate}

Note that $f$ satisfies the conclusions of Propositions \ref{prop:one-side} and \ref{prop:two-side}:

\begin{align*}
\E_{x\sim \mc F_{r-1}}[Q(f,x)] &= \frac{1-P_f/(6\delta)}{\Delta(1-P_f/3)}(1-P_f/3) + \frac{P_f/(6\delta)}{\Delta(1-2\delta)}(1-2\delta) \le \frac{1}{\Delta}\\
\E_{x\sim \mc F_{r-1}}[Q(f,x)Q(f,\ol{x})] &= \frac{1-P_f/(6\delta)}{\Delta(1-P_f/3)}(1-P_f/3)^2 + \frac{P_f/(6\delta)}{\Delta(1-2\delta)}(1-2\delta)^2\\
&= \frac{(1 - P_f/(6\delta) - P_f/3 + P_f^2/(18\delta)) + (P_f/(6\delta) - P_f/3)}{\Delta}\\
&\ge \frac{1-P_f}{\Delta}
\end{align*}

But it is also the case that $f$ nearly tightly satisfies the conclusion of Proposition \ref{prop:r-1-flower-high-weak} for small enough constant $\delta$:

$$\Pr_{x\sim \mc F_{r-1}}[x\in \mc X_{r-1}(f,\delta)] = \frac{1-P_f/(6\delta)}{\Delta(1-P_f/3)} < \frac{1-(1/7)(P_f/\delta)}{\Delta}$$

Thus, $f$, which we rename to $f_{\delta}$ for the next part, could be problematic for Lemma \ref{lem:r-1-flower-high-strong}.\\

\textbf{The slack in proving the weaker version of Lemma \ref{lem:r-1-flower-high-strong}}: When we proved Lemma \ref{lem:r-1-flower-high-strong} with an additional $\log(1/P_f)$ factor, we invoked Proposition \ref{prop:r-1-flower-high-weak} for all $\tau\in [P_f,\xi]$. However, the tight instances $f_{\delta}$ that we just constructed yield tight $\tau$-good probability bounds only for $\tau \approx \delta$, not for arbitrary $\tau$. In fact, by proving Lemma \ref{lem:r-1-flower-high-strong}, we show that there is no $f$ that is simultaneously tight for all values of $\delta$. Indeed, we can see this lack of tightness for $f_{\delta}$ by computing the probability that a flower is $\tau$-good according to the algorithm $f_{\delta}$, for $\tau \ll \delta$:

$$\Pr_{x\sim \mc F_{r-1}}[x\in \mc X_{r-1}(f_{\delta},\tau)] \ge \frac{1 - P_{f_{\delta}}/(6\delta)}{\Delta(1-P_{f_{\delta}}/3)} > \frac{1 - 2P_{f_{\delta}}/\tau}{\Delta}$$

This slack comes from the derivation in the second constraint in the proof of Proposition \ref{prop:r-1-flower-high-weak}; specifically, when $1 - \delta > Q(f_{\delta},\ol{x})$ is used for $x\in \mc Y_{\ge<}$ and $x\in \mc Y_{<<}$. While this improvement may look small, it is actually quite substantial -- plugging in this improved value into the weak version analysis eliminates the $\log(1/P_{f_{\delta}})$ factor.\\

\textbf{Closing the slack in general}: The previous part suggests that grouping $(r-1)$-flowers by their value of $Q(f,x)$ rather than grouping $x$s by membership in $\mc X_{r-1}(f,\tau)$ could lead to a more efficient analysis. This is indeed the case and can be formalized simply by interchanging the expectation over $[0,\xi]$ on the outside with the expectation over $x\in \mc F_{r-1}$ on the inside.\\

This completes the intuition behind Lemma \ref{lem:r-1-flower-high-strong}. We now restate and prove the result:

\lemroneflowerhighstrong*

\begin{proof}
Start by interchanging expectations and simplifying:

\begin{align*}
\E_{\tau\sim [0,1]}\left[\Pr_{x\sim\mc F_{r-1}}\left[x\in \mc X_{r-1}(f,\tau)\right]\middle|\tau\le\xi\right] &= \E_{\tau\sim [0,\xi]}\left[\E_{x\sim\mc F_{r-1}}\left[\textbf{1}_{Q(f,x) \ge 1-\tau}\textbf{1}_{Q(f,\ol{x}) \ge 1-\tau}\right]\right]\\
&= \E_{x\sim\mc F_{r-1}}\left[\E_{\tau\sim [0,\xi]}\left[\textbf{1}_{Q(f,x) \ge 1-\tau}\textbf{1}_{Q(f,\ol{x}) \ge 1-\tau}\right]\right]\\
&= \E_{x\sim\mc F_{r-1}}\left[\E_{\tau\sim [0,\xi]}\left[\textbf{1}_{1-\min(Q(f,x),Q(f,\ol{x})) \le \tau}\right]\right]\\
&= \E_{x\sim\mc F_{r-1}}\left[\max\left(\frac{\xi - (1 - \min\left(Q(f,x),Q(f,\ol{x})\right))}{\xi}, 0\right)\right]\\
&\ge \E_{x\sim\mc F_{r-1}}\left[\max\left(\frac{Q(f,x)Q(f,\ol{x}) - (1-\xi)}{\xi}, 0\right)\right]\\
&\ge \E_{x\sim\mc F_{r-1}}\left[\textbf{1}_{x\in \mc X_{r-1}(f,\xi/2)}\max\left(\frac{Q(f,x)Q(f,\ol{x}) - (1-\xi)}{\xi}, 0\right)\right]\\
&= \E_{x\sim\mc F_{r-1}}\left[\textbf{1}_{x\in \mc X_{r-1}(f,\xi/2)}\left(\frac{Q(f,x)Q(f,\ol{x}) - (1-\xi)}{\xi}\right)\right]\\
&= \E_{x\sim\mc F_{r-1}}\left[\textbf{1}_{x\in \mc X_{r-1}(f,\xi/2)}\left(1 - \frac{1 - Q(f,x)Q(f,\ol{x})}{\xi}\right)\right]\\
\end{align*}

Applying Proposition \ref{prop:r-1-flower-high-weak} with $\delta\gets\xi/2$ shows that

$$\E_{x\sim\mc F_{r-1}}\left[\textbf{1}_{x\in \mc X_{r-1}(f,\xi/2)}\right] = \Pr_{x\sim \mc F_{r-1}}[x\in \mc X_{r-1}(f,\xi/2)] \ge \frac{1-4 P_f/\xi}{\Delta}$$

By Propositions \ref{prop:one-side} and \ref{prop:two-side},

\begin{align*}
\E_{x\sim\mc F_{r-1}}\left[\textbf{1}_{x\in \mc X_{r-1}(f,\xi/2)}Q(f,x)\left(1 -Q(f,\ol{x})\right)\right] &\le \E_{x\sim\mc F_{r-1}}\left[Q(f,x)\left(1 -Q(f,\ol{x})\right)\right]\\
&\le \frac{1}{\Delta} - \frac{1-P_f}{\Delta}\\
&\le \frac{P_f}{\Delta}
\end{align*}

This means, by Proposition \ref{prop:edge-flip} applied to the random variable $\textbf{1}_{x\in \mc X_{r-1}(f,\xi/2)}Q(f,\ol{x})\left(1-Q(f,x)\right)$, that

\begin{align*}
(1-\xi/2)\E_{x\sim\mc F_{r-1}}\left[\textbf{1}_{x\in \mc X_{r-1}(f,\xi/2)}\left(1 -Q(f,x)\right)\right] &\le \E_{x\sim\mc F_{r-1}}\left[\textbf{1}_{x\in \mc X_{r-1}(f,\xi/2)}Q(f,\ol{x})\left(1 - Q(f,x)\right)\right]\\
&= \E_{x\sim\mc F_{r-1}}\left[\textbf{1}_{x\in \mc X_{r-1}(f,\xi/2)}Q(f,x)\left(1 - Q(f,\ol{x})\right)\right]\\
&\le \frac{P_f}{\Delta}
\end{align*}

Rearranging and summing both inequalities shows that

$$\E_{x\sim\mc F_{r-1}}\left[\textbf{1}_{x\in \mc X_{r-1}(f,\xi/2)}\left(1-Q(f,x)Q(f,\ol{x})\right)\right] \le \frac{P_f}{\Delta} + \frac{P_f}{(1-\xi/2)\Delta}\le \frac{3P_f}{\Delta}$$

Dividing by $\xi$ and substituting shows that

$$\E_{\tau\sim [0,1]}\left[\Pr_{x\sim\mc F_{r-1}}\left[x\in \mc X_{r-1}(f,\tau)\right]\middle|\tau\le\xi\right] \ge \frac{1-4P_f/\xi}{\Delta} - \frac{3P_f}{\xi\Delta} \ge \frac{1 - 7P_f/\xi}{\Delta}$$

as desired.
\end{proof}

\subsection{Dominant Directions in $(r-1)$-Neighborhoods (Proof of Lemma \ref{lem:r-1-neighborhood-low-common})}\label{subsec:r-1-neighborhood-low}

In this section, we show the existence of a dominant direction (Lemma \ref{lem:r-1-neighborhood-low-common}), which we restate here for convenience:

\lemroneneighborhoodlowcommon*

\textbf{Intuition}: We prove the desired result by contradiction. Suppose for the sake of contradiction that

$$\sum_{j\in [\Delta], j\ne i} P_j(f,x,\delta_j) > F_1(\delta_i)$$

for all $i\in [\Delta]$. In Lemma \ref{lem:r-1-neighborhood-low-common}, we are extending an $(r-1)$-neighborhood $x$ in $\Delta$ different directions to obtain an $r$-neighborhood. One thing to note about these $\Delta$ different extensions is that they are taking place on disjoint sets of edges (they only overlap in $x$) and are thus independent. We obtain a contradiction by showing that $\text{dir}(f,x)$ actually has to point in two different directions if all $\Delta$ sums are high. We obtain this contradiction as follows:

\begin{enumerate}
    \item Key Question: Does there exist an $r$-neighborhood $z$ that extends $x$ for which both
    \begin{enumerate}
        \item at least two distinct $i\in [\Delta]$ have the property that the $(r-1)$-flower in direction $i$ ($\text{res}_i(z)$) is $\delta_i$-good
        \item all $i\in [\Delta]$ for which $\text{res}_i(z)$ is $\delta_i$-good have the property that $\text{dir}(f,z)$ points in direction $i$ ($\text{dir}(f,z) = i$)?
    \end{enumerate}
    \item If this happens, we immediately get a contradiction -- combining the two properties shows that $\text{dir}(f,z)$ points in at least two distinct directions, contradicting the definition of $\text{dir}$/Proposition \ref{prop:r-ngbr-exclusive}.
    \item We can prove the existence of such a $z$ using the probabilistic method:
    \begin{enumerate}
        \item $\Pr[\text{(a) holds}] \approx \sum_{j\ne i\in [\Delta]} P_i(f,x,\delta_i)P_j(f,x,\delta_j)$, as we can roughly speaking lower bound the probability that at least 2 directions are good by summing over all possible pairs of directions. Independence of the $\Delta$ different directions is used to lower bound each of these terms.\footnote{Terms involving 3 or more indices do matter, making the actual probability calculation more complicated, but this is the rough idea.}
        \item By the contradiction assumption, $\sum_{j\ne i\in [\Delta]} P_i(f,x,\delta_i)P_j(f,x,\delta_j) \ge \sum_{i\in [\Delta]} P_i(f,x,\delta_i)F_1(\delta_i)$.
        \item $\Pr[\text{(b) does not hold} | \text{(a) holds}] \le \frac{\Pr[\text{(b) does not hold}]}{\Pr[\text{(a) holds}]}$. The previous step lower bounded the denominator. It suffices to show that this probability is less than 1.
        \item To upper bound the numerator, we union bound: \begin{align*}\Pr[\text{(b) does not hold}] &\le \sum_{i=1}^{\Delta} \Pr[\text{(b) does not hold for direction } i]\\ &= \sum_{i=1}^{\Delta} \Pr[\text{dir}(f,z) \ne i \mid \text{res}_i(z) \text{ is $\delta_i$-good}]\Pr[\text{res}_i(z) \text{ is $\delta_i$-good}]\end{align*}
        \item We bound each of the two terms separately:
        \begin{enumerate}
            \item Left Term: Here, we are extending the $i$-direction $(r-1)$-flower to an $r$-neighborhood. Since the $(r-1)$-flower in that direction is $\delta_i$-good, $$\Pr[\text{dir}(f,z) \ne i \mid \text{res}_i(z) \text{ is $\delta_i$-good}] \le \delta_i$$ by definition of $\delta_i$-goodness.
            \item Right Term: Here, we are extending the $(r-1)$-neighborhood $x$ in direction $i$. This flower is $\delta_i$-good by definition with probability $P_i(f,x,\delta_i)$.
            \item Thus, $$\Pr[\text{(b) does not hold}] \le \sum_{i=1}^{\Delta} \delta_i P_i(f,x,\delta_i)$$
        \end{enumerate}
        \item Therefore, $$\Pr[\text{(b) does not hold} | \text{(a) holds}] \le \frac{\sum_{i=1}^{\Delta} \delta_i P_i(f,x,\delta_i)}{\sum_{i=1}^{\Delta}F_1(\delta_i)P_i(f,x,\delta_i)} < 1$$ as desired.
        \item \textbf{To summarize}, for each $i\in [\Delta]$ for which the $i$ direction $(r-1)$-flower is $\delta_i$-good, the set of $r$-neighborhood extensions $z$ that have at least one other good direction (probability mass $>F_1(\delta_i)$) overlaps with the set of extensions $z$ for which $\text{dir}(f,z) = i$ (probability mass $>1-\delta_i$).
    \end{enumerate}
\end{enumerate}

\textbf{Roadmap}: We break the proof up into three parts, which are done in three separate subsubsections:

\begin{enumerate}
    \item Abstract away problem-specific details of Lemma \ref{lem:r-1-neighborhood-low-common} to obtain a problem that only deals with coordinate tuples (Proposition \ref{prop:distrib}, Section \ref{sec:DistribReduce}).
    \item Prove a crude bound on the sum of all $P_i$s first (Proposition \ref{prop:crude-distrib}) to make further bounds easier later on. This part uses the fact that the $\delta_i$s are at most 1/2 and does not use independence of the directions. (Section \ref{sec:DistribCrude})
    \item Prove the desired conditional probability upper bound to obtain the desired result (Proposition \ref{prop:distrib}) using independence of the directions. (Section \ref{sec:DistribFinal})
\end{enumerate}

\subsubsection{Abstracting Away Neighborhoods and Flowers (Proposition \ref{prop:distrib} Statement, Lemma \ref{lem:r-1-neighborhood-low-common} Proof)}\label{sec:DistribReduce}

We simplify Lemma \ref{lem:r-1-neighborhood-low-common} down to the following proposition, where $\Omega_i$ represents the set of $(r-1)$-flower extensions of $x$ in direction $i$, $\Theta_i$ represents (a slight superset of) the set of $\delta_i$-good flowers in direction $i$, and $h(z) = \text{dir}(f,z)$:

\begin{prop}\label{prop:distrib}
Let $\Omega_1,\hdots,\Omega_{\Delta}$ be arbitrary probability spaces and let $\Omega = \Omega_1\times \Omega_2\times \hdots\Omega_{\Delta}$. Let $h:\Omega\rightarrow [[\Delta]]$\footnote{Recall that $[[\Delta]] = \{0,1,\hdots,\Delta\}$.} and $\delta_1,\delta_2,\hdots,\delta_{\Delta} \in [0,1/2]$. Define

$$\Theta_i := \{y\in \Omega_i : \Pr_{z\sim \Omega}[h(z) = i \mid z_i = y] \ge 1 - \delta_i\}$$

Let $\theta_i := \Pr_{y\sim \Omega_i}[y\in \Theta_i]$. Then, there exists $i^*\in [\Delta]$ such that

$$\sum_{i\in [\Delta], i\ne i^*} \theta_i \le F_1(\delta_{i^*}) = 6e^4\delta_{i^*}$$
\end{prop}

\begin{proof}[Proof of Lemma \ref{lem:r-1-neighborhood-low-common} given Proposition \ref{prop:distrib}]
Apply Proposition \ref{prop:distrib} with the following parameter choices:

\begin{enumerate}
    \item For any $i\in [\Delta]$, $\Omega_i\gets $ the set of all $y\in \mc F_{r-1}$ with the property that $\text{end}_A(y) = \sigma_i(x)$.
    \item For any $z\in \Omega$, define $h(z) \gets \text{dir}(f,\text{glue}(z_1,z_2,\hdots,z_{\Delta}))$.
    \item For any $i\in [\Delta]$, define $\delta_i\gets\delta_i\text{  }(\in [0,1/2])$.
\end{enumerate}

For each $i\in [\Delta]$, applying Fact \ref{fact:same-dir} (second equality below) and Proposition \ref{prop:r-to-r-sample} with $\zeta\gets y$ and $X(w)\gets \textbf{1}_{\text{dir}(f,w) = i}$ (third equality below) shows that

\begin{align*}
P_i(f,x,\delta_i) &= \Pr_{y\sim \mc F_{r-1}}[y\in \mc X_{r-1}(f,\delta_i) \mid \text{end}_A(y) = \sigma_i(x)]\\
&\le \Pr_{y\sim \mc F_{r-1}}\left[\Pr_{z\sim \mc R_r}[\text{dir}(f,z) = 1 \mid \text{res}_1(z) = y] \ge 1-\delta_i \mid \text{end}_A(y) = \sigma_i(x)\right]\\
&= \Pr_{y\sim \mc F_{r-1}}\left[\Pr_{z\sim \mc R_r}[\text{dir}(f,z) = i \mid \text{res}_i(z) = y] \ge 1-\delta_i \mid \text{end}_A(y) = \sigma_i(x)\right]\\
&= \Pr_{y\sim \Omega_i}\left[\Pr_{z\sim \Omega}[h(z) = i \mid z_i = y] \ge 1-\delta_i\right]\\
&= \theta_i
\end{align*}

Thus, Proposition \ref{prop:distrib} implies that there is an $i^*\in [\Delta]$ for which

$$\sum_{i\in [\Delta], i\ne i^*} P_i(f,x,\delta_i) \le \sum_{i\in [\Delta], i\ne i^*} \theta_i \le F_1(\delta_{i^*})$$

as desired.
\end{proof}

\subsubsection{A Crude Bound (Proposition \ref{prop:crude-distrib})}\label{sec:DistribCrude}

It will be helpful to have the following crude upper bound on the $\theta_i$s to be bootstrapped later. We prove this by lower bounding the expected number of directions that $h$ must point in:

\begin{prop}\label{prop:crude-distrib}
Let $\Omega_1,\Omega_2,\hdots,\Omega_{\Delta}$ be arbitrary probability spaces, $h:\Omega\rightarrow[[\Delta]]$, and $\delta_1,\delta_2,\hdots,\delta_{\Delta}\in [0,1/2]$. Define $\Omega$, $\Theta_1,\hdots,\Theta_{\Delta}$ and $\theta_1,\hdots,\theta_{\Delta}$ using the $\Omega_i$s, $h$, and $\delta_i$s as in Proposition \ref{prop:distrib}. Then

$$\sum_{i=1}^{\Delta} \theta_i \le \frac{1}{1-\displaystyle\max_{i\in[\Delta]}\delta_i} \le 2$$
\end{prop}

\begin{proof}
For $y\in \Omega$, let $S_y := \{i\in [\Delta]: i = h(y)\}$. Since $h$ is a function, $|S_y| \le 1$ for all $y\in \Omega$. Thus,

$$\E_{y\sim \Omega}[|S_y|] \le 1$$

However, we can also write this expectation as a sum over coordinates by linearity of expectation:

$$\sum_{i=1}^{\Delta} \Pr_{y\sim \Omega}[h(y) = i] = \sum_{i=1}^{\Delta} \Pr_{y\sim \Omega}[i\in S_y] = \E_{y\sim \Omega}\left[\sum_{i=1}^{\Delta} \textbf{1}_{i\in S_y}\right] = \E_{y\sim \Omega}[|S_y|] \le 1$$

For any fixed $i\in [\Delta]$, we can sample $y$ by first sampling $x\sim \Omega_i$ and then sampling $y\sim \Omega$ conditioned on $y_i = x$. This means that

\begin{align*}
\Pr_{y\sim \Omega}[h(y) = i] &= \E_{x\sim \Omega_i}\left[\Pr_{y\sim \Omega}[h(y) = i \mid y_i = x]\right]\\
&\ge \E_{x\sim \Omega_i}\left[\Pr_{y\sim \Omega}[h(y) = i \mid y_i = x]\textbf{1}_{x\in \Theta_i}\right]\\
&\ge (1 - \delta_i)\Pr_{x\sim \Omega_i}[x\in \Theta_i]\\
&= (1 - \delta_i)\theta_i
\end{align*}

Plugging this inequality in shows that

$$\left(1-\max_{i\in [\Delta]}\delta_i\right)\sum_{i=1}^{\Delta} \theta_i \le \sum_{i=1}^{\Delta} (1-\delta_i)\theta_i \le 1$$

as desired.
\end{proof}

\subsubsection{Final Bound (Proof of Proposition \ref{prop:distrib})}\label{sec:DistribFinal}

\begin{prop}[Folklore]\label{prop:eapx}
For any $x\in [0,1)$,

$$e^{-x/(1-x)} \le 1-x \le e^{-x}$$
\end{prop}

\begin{proof}
For the upper bound, for any $k\ge 2$, $\frac{x^k}{k!} \ge \frac{x^{k+1}}{(k+1)!}$, so $1-x \le \sum_{k=0}^{\infty} \frac{(-1)^kx^k}{k!} = e^{-x}$. For the lower bound, $x/(1-x)$ is positive, so $e^{x/(1-x)} \ge 1 + \frac{x}{1-x} = \frac{1}{1-x}$, as desired.
\end{proof}

We now prove the desired result using the intuition stated at the beginning of this section:

\begin{proof}[Proof of Proposition \ref{prop:distrib}]
Assume, for the sake of contradiction, that

$$\sum_{j\in [\Delta], j\ne i} \theta_j > F_1(\delta_i)$$

for all $i\in [\Delta]$. For $z\in \Omega$, let $I_z\subseteq [\Delta]$ denote the set of all $i\in [\Delta]$ for which $z_i\in \Theta_i$. For $z\in \Omega$, let $J_z\subseteq I_z$ denote the set of all $i\in I_z$ for which $h(z) = i$. We will obtain a contradiction by sampling $z$ uniformly from $\Omega$ conditioned on $|I_z| \ge 2$ and showing that there is a nonzero probability that $J_z = I_z$, contradicting the fact that $h$ has one output.

We start by lower bounding the denominator of the relevant conditional probability. Let $\Gamma \subseteq [\Delta]$ be the set of all $i\in [\Delta]$ for which $\theta_i > 1/2$. Proposition \ref{prop:crude-distrib} implies that $|\Gamma| \le 3$, so the list of sets of the form $\{i,j\}\cup \Gamma$ for all $1\le i < j\le \Delta$ repeats a set at most 3 times. For all $i\in \Gamma$, $\theta_i \ge 1-\theta_i$. Thus,

\begin{align*}
\Pr_{z\sim \Omega}[|I_z| \ge 2] &\ge \sum_{i,j\in [\Delta], i < j} \frac{1}{3}\Pr[I_z = \{i,j\} \cup \Gamma]\\
&= \frac{1}{3}\sum_{i,j\in [\Delta], i < j} \theta_i \theta_j \prod_{k\in \Gamma\setminus \{i,j\}} \theta_k \prod_{k\in [\Delta]\setminus \Gamma, k\ne i,j} (1 - \theta_k)\\
&\ge \frac{1}{3}\sum_{i,j\in [\Delta], i < j} \theta_i \theta_j \prod_{k\in \Gamma\setminus \{i,j\}} e^{-(1-\theta_k)/\theta_k} \prod_{k\in [\Delta]\setminus \Gamma, k\ne i,j} e^{-\theta_k/(1-\theta_k)}\\
&\ge \frac{1}{3}\sum_{i,j\in [\Delta], i < j} \theta_i \theta_j \prod_{k\in \Gamma\setminus \{i,j\}} e^{-2(1-\theta_k)} \prod_{k\in [\Delta]\setminus \Gamma, k\ne i,j} e^{-2\theta_k}\\
&\ge \frac{1}{3}\sum_{i,j\in [\Delta], i < j} \theta_i \theta_j \prod_{k\in [\Delta], k\ne i,j} e^{-2\theta_k}\\
&\ge \frac{1}{3}e^{-4} \sum_{i,j\in [\Delta], i < j} \theta_i \theta_j\\
&= \frac{1}{6}e^{-4} \sum_{i,j\in [\Delta], i\ne j} \theta_i \theta_j\\
&= \frac{1}{6}e^{-4} \sum_{i=1}^{\Delta} \theta_i\sum_{j\in [\Delta], j\ne i} \theta_j\\
&\ge \frac{1}{6}e^{-4} \sum_{i=1}^{\Delta} \theta_i F_1(\delta_i)\\
\end{align*}

where the first inequality follows from Proposition \ref{prop:eapx} and the second-to-last follows from Proposition \ref{prop:crude-distrib}. This is the desired lower bound for the denominator. Next, we upper bound the numerator of the relevant conditional probability using a union bound:

\begin{align*}
\Pr_{z\sim \Omega}[J_z \ne I_z] &\le \sum_{i=1}^{\Delta} \Pr_{z\sim \Omega}[i\notin J_z \text{ and } i\in I_z]\\
&= \sum_{i=1}^{\Delta} \E_{y\sim \Omega_i}\left[\Pr_{z\sim \Omega}[i\notin J_z \text{ and } i\in I_z \mid z_i = y]\right]\\
&= \sum_{i=1}^{\Delta} \E_{y\sim \Omega_i}\left[\Pr_{z\sim \Omega}[i\notin J_z \text{ and } z_i\in \Theta_i \mid z_i = y]\right]\\
&= \sum_{i=1}^{\Delta} \E_{y\sim \Omega_i}\left[\textbf{1}_{y\in \Theta_i}\Pr_{z\sim \Omega}[i\notin J_z \mid z_i = y]\right]\\
&\le \sum_{i=1}^{\Delta} \E_{y\sim \Omega_i}\left[\textbf{1}_{y\in \Theta_i}\delta_i\right]\\
&= \sum_{i=1}^{\Delta} \theta_i\delta_i
\end{align*}

Thus,

$$\Pr_{z\sim \Omega}[J_z = I_z \mid |I_z| \ge 2] \ge 1 - \frac{\Pr_{z\sim\Omega}[J_z \ne I_z]}{\Pr_{z\sim\Omega}[|I_z|\ge 2]} \ge 1 - \frac{\sum_{i=1}^{\Delta} \theta_i\delta_i}{(e^{-4}/6)\sum_{i=1}^{\Delta} \theta_iF_1(\delta_i)} > 0$$

which means that there exists a $z\in \Omega$ for which $J_z = I_z$ and $|I_z| \ge 2$. But this $z$ also has the property that $|J_z| \ge 2$, meaning that $h(z)$ must have two different values, a contradiction. Thus, the desired result holds.
\end{proof}

\subsection{Bound on Expected $\delta_{\text{dom}}$ (Proof of Proposition \ref{prop:complement-prob})}\label{sec:complement-prob}

To prove Proposition \ref{prop:complement-prob}, we need to show that most $(r-1)$-neighborhoods $x$ have a low value of $\delta_{\text{dom}}(f,x)$. As a warmup\footnote{This proposition is not used anywhere to prove Theorem \ref{thm:main}.} to illustrate the key ideas behind Proposition \ref{prop:complement-prob}, we show a tail bound on $\delta_{\text{dom}}$, i.e. that it is $O(P_f)$ with constant probability:

\begin{prop}\label{prop:weak-complement-prob}
For any $\tau\in (0,1/2)$,

$$\Pr_{x\sim \mc R_{r-1}}[\delta_{\text{dom}}(f,x) > \tau] = \Pr_{x\sim \mc R_{r-1}}[x\notin \mc R_{r-1}(f,\tau)] \le 600e^4(P_f/\tau + \tau)$$

In particular, when $P_f \le 1/(60000e^4)^2$,

$$\Pr_{x\sim \mc R_{r-1}}[\delta_{\text{dom}}(f,x) > 60000e^4 P_f] \le 1/25$$
\end{prop}

\begin{proof}
The equality holds by definition and the last conclusion holds by substituting $\tau = 60000e^4 P_f$, so it suffices to show the inequality. By definition of $\mc R_{r-1}(f,\tau)$,

$$\Pr_{x\sim \mc R_{r-1}}\left[x \notin \mc R_{r-1}(f,\tau)\right] = \Pr_{x\sim \mc R_{r-1}}\left[P_{\max}(f,x,\tau) < C_4\right]$$

We can upper bound this probability using Markov's Inequality on the random variable $1 - P_{\max}(f,x,\tau)$, using the fact that $P_{\max}(f,x,\tau) \le 1$ to show that the relevant random variable is nonnegative:

\begin{align*}
\Pr_{x\sim \mc R_{r-1}}\left[P_{\max}(f,x,\tau) < C_4\right] &= \frac{\displaystyle\E_{x\sim \mc R_{r-1}}\left[1 - P_{\max}(f,x,\tau)\right]}{1 - C_4}\\
&= \frac{1}{1-C_4}\left(\E_{x\sim \mc R_{r-1}}\left[1 - P(f,x,\tau)\right] + \E_{x\sim \mc R_{r-1}}\left[P_{\text{comp}}(f,x,\tau)\right]\right)
\end{align*}

We bound the two terms within the expectation separately. Intuitively, $1 - P_{\max}$ is very close to $1 - P$ because $P_{\text{comp}}$ is always small, and the expectation of $1 - P$ is very close to 0 because of the low vertex survival probability of $f$. We now formalize this, starting with a bound on the left term. Applying Proposition \ref{prop:r-to-f-sample} for the random variable $\textbf{1}_{y\in \mc X_{r-1}(f,\tau)}$ yields

\begin{align*}
\Pr_{y\sim \mc F_{r-1}}[y\in \mc X_{r-1}(f,\tau)] &= \E_{x\sim \mc R_{r-1}}\left[\E_{i\sim [\Delta]}\left[\Pr_{y\sim \mc F_{r-1}}[y\in \mc X_{r-1}(f,\tau) \mid \text{end}_A(y) = \sigma_i(x)]\right]\right]\\
&= \E_{x\sim \mc R_{r-1}}\left[\E_{i\sim [\Delta]}\left[P_i(f,x,\tau)\right]\right]\\
&= \frac{\displaystyle\E_{x\sim \mc R_{r-1}}[P(f,x,\tau)]}{\Delta}
\end{align*}

Rearranging and applying Proposition \ref{prop:r-1-flower-high-weak} shows that

$$\E_{x\sim \mc R_{r-1}}[1 - P(f,x,\tau)] = 1 - \Delta \Pr_{y\sim \mc F_{r-1}}[y\in \mc X_{r-1}(f,\tau)] \le 2P_f/\tau$$

We now proceed to bounding the second term. We bound this using Proposition \ref{prop:r-1-neighborhood-low-weak}:

$$\E_{x\sim \mc R_{r-1}}[P_{\text{comp}}(f,x,\tau)] \le \E_{x\sim \mc R_{r-1}}[F_1(\tau)] \le 6e^4\tau$$

Therefore,

$$\Pr_{x\sim \mc R_{r-1}}[x\notin \mc R_{r-1}(f,\tau)] \le \frac{6e^4}{1 - C_4}(P_f/\tau + \tau) \le 600e^4(P_f/\tau + \tau)$$

as desired.
\end{proof}

As discussed in Section \ref{sec:TORefinement}, this tail bound is not enough. We also need an expectation bound. To obtain the desired expectation bound, we improve a few aspects of this proof:

\begin{enumerate}
    \item For any random variable $X$ that takes values in $[0,1]$, $$\E[X] = \E_{\tau\sim [0,1]}[\Pr[X > \tau]]$$Applying this with $X := \delta_{\text{dom}}(f,x)$ allows us to write the desired expectation as an expectation over the tail probabilities bounded by Proposition \ref{prop:weak-complement-prob}.
    \item Substituting the bound given by Proposition \ref{prop:weak-complement-prob} would yield (up to a constant factor) a bound of $$\E_{\tau\sim [0,1]}[P_f/\tau + \tau] = P_f\E[1/\tau] + \E[\tau] \gg P_f$$on the desired expectation. Note that both terms are too large.
    \item Improving the first term:
    \begin{enumerate}
        \item To improve the first term, it suffices to only consider $\tau > P_f$, and when we do this $\displaystyle\E_{\tau\sim [P_f, 1]}[1/\tau] \approx \ln (1/P_f)$ so the first term is off from the desired value by a logarithmic factor.
        \item This logarithmic factor is intuitively coming from overcounting of $(r-1)$-flowers that are $\tau$-good for very low values of $\tau$. For more details, see the discussion in Section \ref{sec:FlowerStrong}.
        \item Formally, we correct this logarithmic factor loss simply by replacing the application of Proposition \ref{prop:r-1-flower-high-weak} with Lemma \ref{lem:r-1-flower-high-strong}.
    \end{enumerate}
    \item Improving the second term:
    \begin{enumerate}
        \item The uniform bound of $\tau$ can be replaced with the neighborhood-dependent bound of $\delta_{\text{dom}}(f,x)$.
        \item While this term may look initially look complicated to bound on its own, it is simple in our case, because it is the same as the thing we are trying to bound! So we can simply rearrange terms to get the desired bound. Roughly speaking, we show that $$\E[\delta_{\text{dom}}(f,x)] \le O(P_f) + (1/2)\E[\delta_{\text{dom}}(f,x)]$$and get the desired result by rearranging.\footnote{We actually show this for $\min(\delta_{\text{dom}}(f,x), C)$ rather than $\delta_{\text{dom}}(f,x)$ for some constant $C$.} For details, see the proof.
        \item This rearrangement trick may look opaque. Intuitively, the rearrangement is really just bootstrapping better tail bounds from worse ones. In more detail, notice that Proposition \ref{prop:weak-complement-prob} gives good probability bounds (where the first term $P_f/\tau$ is greater than the second term $\tau$) for all $\tau < \sqrt{P_f}$. We can bootstrap this bound to obtain good bounds up to $\tau < P_f^{1/3}$ as follows. Split $(r-1)$-neighborhoods into ones in $\mc R_{r-1}(f,\kappa)$ and ones not in $\mc R_{r-1}(f,\kappa)$ for the optimized choice $\kappa := \sqrt{P_f\tau}$. We can replace the relevant part of the Proposition \ref{prop:weak-complement-prob} proof as follows: \begin{align*}\E_{x\sim \mc R_{r-1}}[P_{\text{comp}}(f,x,\tau)] &= \E_{x\sim \mc R_{r-1}}[P_{\text{comp}}(f,x,\tau) \mid x\in \mc R_{r-1}(f,\kappa)] \Pr_{x\sim \mc R_{r-1}}[x\in \mc R_{r-1}(f,\kappa)]\\ &+ \E_{x\sim \mc R_{r-1}}[P_{\text{comp}}(f,x,\tau) \mid x\notin \mc R_{r-1}(f,\kappa)] \Pr_{x\sim \mc R_{r-1}}[x\notin \mc R_{r-1}(f,\kappa)]\\
        &\le 6e^4\E_{x\sim \mc R_{r-1}}[\delta_{\text{dom}}(f,x) \mid x\in \mc R_{r-1}(f,\kappa)] \Pr_{x\sim \mc R_{r-1}}[x\in \mc R_{r-1}(f,\kappa)]\\ &+ 6e^4\E_{x\sim \mc R_{r-1}}[\tau \mid x\notin \mc R_{r-1}(f,\kappa)] \Pr_{x\sim \mc R_{r-1}}[x\notin \mc R_{r-1}(f,\kappa)]\\
        &\le 6e^4\kappa\\ &+ 3600e^8(\tau)(P_f/\kappa + \kappa)\\
        &= O(\kappa + \tau P_f/\kappa)
        \end{align*}This is a better bound than $O(\tau)$ for many choices of $\kappa$. Combining with the first term leads to an improved tail bound of $$\Pr_{x\sim \mc R_{r-1}}[\delta_{\text{dom}}(f,x) > \tau] \le O(P_f/\tau + \kappa + \tau P_f/\kappa) = O(P_f/\tau + \sqrt{P_f\tau})$$ and the first term dominates up to $\tau < P_f^{1/3}$ as desired. The rearrangement trick does this bootstrapping in a continuous way.
    \end{enumerate}
\end{enumerate}

We now restate and formally prove Proposition \ref{prop:complement-prob}:

\propcomplementprob*

\begin{proof}
Let $C_{10} := \min(C_5, \frac{1-C_4}{12e^4}) = \frac{1}{1200e^4}$. We instead prove the stronger bound of

$$\E_{x\sim \mc R_{r-1}}[\min(\delta_{\text{dom}}(f,x), C_{10})] \le C_{10}C_{11}P_f$$

which is stronger due to the fact that $\delta_{\text{dom}}(f,x)\le 1$ and in turn $\delta_{\text{dom}}(f,x) \le \frac{1}{C_{10}}\min(\delta_{\text{dom}}(f,x), C_{10})$ for all $x\in \mc R_{r-1}$. We start by rewriting the expectation as an expectation over tail probabilities and apply Markov's Inequality to bound each tail probability individually:

\begin{align*}
\E_{x\sim \mc R_{r-1}}[\min(\delta_{\text{dom}}(f,x),C_{10})] &= \E_{x\sim \mc R_{r-1}}\left[\Pr_{\tau\sim [0,1]}\left[\delta_{\text{dom}}(f,x) > \tau \text{ and } C_{10} > \tau\right]\right]\\
&= \E_{x\sim \mc R_{r-1}}\left[\Pr_{\tau\sim [0,1]}\left[x \notin \mc R_{r-1}(f,\tau) \text{ and } C_{10} > \tau\right]\right]\\
&= \E_{\tau\sim [0,1]}\left[\Pr_{x\sim \mc R_{r-1}}\left[x \notin \mc R_{r-1}(f,\tau)\right] \textbf{1}_{C_{10} > \tau}\right]\\
&= \E_{\tau\sim [0,1]}\left[\Pr_{x\sim \mc R_{r-1}}\left[P_{\max}(f,x,\tau) < C_4\right] \textbf{1}_{C_{10} > \tau}\right]\\
&\le \E_{\tau\sim [0,1]}\left[\frac{\displaystyle\E_{x\sim \mc R_{r-1}}\left[1-P_{\max}(f,x,\tau)\right]}{1-C_4} \textbf{1}_{C_{10} > \tau}\right]\\
&= \frac{C_{10}}{1-C_4}\E_{\tau\sim [0,1]}\left[\E_{x\sim \mc R_{r-1}}\left[(1-P(f,x,\tau)) + P_{\text{comp}}(f,x,\tau)\right] \middle| \tau \le C_{10}\right]\\
\end{align*}

We bound the two terms within the expectation separately. We start with the left term. Applying Proposition \ref{prop:r-to-f-sample} for the random variable $\textbf{1}_{y\in \mc X_{r-1}(f,\tau)}$ yields

\begin{align*}
\Pr_{y\sim \mc F_{r-1}}[y\in \mc X_{r-1}(f,\tau)] &= \E_{x\sim \mc R_{r-1}}\left[\E_{i\sim [\Delta]}\left[\Pr_{y\sim \mc F_{r-1}}[y\in \mc X_{r-1}(f,\tau) \mid \text{end}_A(y) = \sigma_i(x)]\right]\right]\\
&= \E_{x\sim \mc R_{r-1}}\left[\E_{i\sim [\Delta]}\left[P_i(f,x,\tau)\right]\right]\\
&= \frac{\displaystyle\E_{x\sim \mc R_{r-1}}[P(f,x,\tau)]}{\Delta}
\end{align*}

Applying Lemma \ref{lem:r-1-flower-high-strong} with $\xi\gets C_{10}$ yields

\begin{align*}
\E_{\tau\sim [0,1]}\left[\E_{x\sim \mc R_{r-1}}[1 - P(f,x,\tau)]\middle|\tau\le C_{10}\right] &= 1 - \Delta\E_{\tau\sim [0,1]}\left[\Pr_{y\sim \mc F_{r-1}}[y \in \mc X_{r-1}(f,\tau)]\middle|\tau\le C_{10}\right]\\
&\le \frac{7 P_f}{C_{10}}
\end{align*}

This completes the bound of the left term. Now, we proceed to the right term. By Propositions \ref{prop:r-1-neighborhood-low-weak} and \ref{prop:comp-monotone}, for any $x\in \mc R_{r-1}$ and any $\tau\le C_{10}$

$$P_{\text{comp}}(f,x,\tau) \le P_{\text{comp}}(f,x,C_{10}) \le 6e^4 C_{10}$$

By definition of $\delta_{\text{dom}}(f,x)$, Proposition \ref{prop:comp-monotone}, Lemma \ref{lem:r-1-neighborhood-low-strong}, and the fact that $C_{10} \le C_5$,

$$P_{\text{comp}}(f,x,\tau) \le P_{\text{comp}}(f,x,C_5) \le 6e^4 \delta_{\text{dom}}(f,x)$$

This completes the bound of the right term. Substitution shows that

\begin{align*}
\E_{x\sim \mc R_{r-1}}[\min(\delta_{\text{dom}}(f,x),C_{10})] &\le \frac{7P_f}{1-C_4} + \frac{6e^4 C_{10}}{1-C_4}\E_{\tau\sim [0,1]}\left[\E_{x\sim\mc R_{r-1}}\left[\min(\delta_{\text{dom}}(f,x),C_{10})\right]\middle| \tau\le C_{10}\right]\\
&= \frac{7P_f}{1-C_4} + \frac{6e^4 C_{10}}{1-C_4}\E_{x\sim\mc R_{r-1}}\left[\min(\delta_{\text{dom}}(f,x),C_{10})\right]\\
&\le \frac{7P_f}{1-C_4} + \frac{1}{2}\E_{x\sim\mc R_{r-1}}\left[\min(\delta_{\text{dom}}(f,x),C_{10})\right]\\
\end{align*}

Rearranging as discussed in the intuition prior to this proof shows that

$$\E_{x\sim \mc R_{r-1}}[\min(\delta_{\text{dom}}(f,x),C_{10})] \le \frac{14P_f}{1-C_4} = C_{10}C_{11}P_f$$

as desired.
\end{proof}
\paragraph{Acknowledgment:} We would like to thank Amir Abboud, Sepehr Assadi, and Shafi Goldwasser for several helpful discussions and comments.
\bibliographystyle{plain}
\bibliography{refs.bib}

\begin{thebibliography}{10}

\bibitem{AlonBI86}
Noga Alon, L{\'{a}}szl{\'{o}} Babai, and Alon Itai.
\newblock A fast and simple randomized parallel algorithm for the maximal independent set problem.
\newblock {\em J. Algorithms}, 7(4):567--583, 1986.

\bibitem{BalliuGKO23}
Alkida~Balliu andf Mohsen~Ghaffari, Fabian Kuhn, and Dennis Olivetti.
\newblock Node and edge averaged complexities of local graph problems.
\newblock {\em Distributed Comput.}, 36(4):451--473, 2023.

\bibitem{AravindaMM21}
Heshan Aravinda, Arnaud Marsiglietti, and James Melbourne.
\newblock Concentration inequalities for ultra log-concave distributions, 2021.

\bibitem{BalliuB0O24}
Alkida Balliu, Thomas Boudier, Sebastian Brandt, and Dennis Olivetti.
\newblock Tight lower bounds in the supported {LOCAL} model.
\newblock In {\em PODC}, pages 95--105. {ACM}, 2024.

\bibitem{BalliuBHORS21}
Alkida Balliu, Sebastian Brandt, Juho Hirvonen, Dennis Olivetti, Mika{\"{e}}l Rabie, and Jukka Suomela.
\newblock Lower bounds for maximal matchings and maximal independent sets.
\newblock {\em J. {ACM}}, 68(5):39:1--39:30, 2021.

\bibitem{Balliu0KO21}
Alkida Balliu, Sebastian Brandt, Fabian Kuhn, and Dennis Olivetti.
\newblock Improved distributed lower bounds for {MIS} and bounded (out-)degree dominating sets in trees.
\newblock In Avery Miller, Keren Censor{-}Hillel, and Janne~H. Korhonen, editors, {\em {PODC} '21: {ACM} Symposium on Principles of Distributed Computing, Virtual Event, Italy, July 26-30, 2021}, pages 283--293. {ACM}, 2021.

\bibitem{Balliu0KO22}
Alkida Balliu, Sebastian Brandt, Fabian Kuhn, and Dennis Olivetti.
\newblock Distributed {$\Delta$}-coloring plays hide-and-seek.
\newblock In Stefano Leonardi and Anupam Gupta, editors, {\em {STOC} '22: 54th Annual {ACM} {SIGACT} Symposium on Theory of Computing, Rome, Italy, June 20 - 24, 2022}, pages 464--477. {ACM}, 2022.

\bibitem{Balliu0KO23}
Alkida Balliu, Sebastian Brandt, Fabian Kuhn, and Dennis Olivetti.
\newblock Distributed maximal matching and maximal independent set on hypergraphs.
\newblock In {\em {SODA}}, pages 2632--2676. {SIAM}, 2023.

\bibitem{BalliuBO22}
Alkida Balliu, Sebastian Brandt, and Dennis Olivetti.
\newblock Distributed lower bounds for ruling sets.
\newblock {\em {SIAM} J. Comput.}, 51(1):70--115, 2022.

\bibitem{Bar-YehudaCS17}
Reuven Bar{-}Yehuda, Keren Censor{-}Hillel, and Gregory Schwartzman.
\newblock A distributed {(2} + {\(\epsilon\)})-approximation for vertex cover in o(log {\(\Delta\)} / {\(\epsilon\)} log log {\(\Delta\)}) rounds.
\newblock {\em J. {ACM}}, 64(3):23:1--23:11, 2017.

\bibitem{BarenboimE09}
Leonid Barenboim and Michael Elkin.
\newblock Distributed (delta+1)-coloring in linear (in delta) time.
\newblock In Michael Mitzenmacher, editor, {\em Proceedings of the 41st Annual {ACM} Symposium on Theory of Computing, {STOC} 2009, Bethesda, MD, USA, May 31 - June 2, 2009}, pages 111--120. {ACM}, 2009.

\bibitem{BarenboimE10}
Leonid Barenboim and Michael Elkin.
\newblock Sublogarithmic distributed {MIS} algorithm for sparse graphs using nash-williams decomposition.
\newblock {\em Distributed Comput.}, 22(5-6):363--379, 2010.

\bibitem{2013Barenboim}
Leonid Barenboim and Michael Elkin.
\newblock {\em Distributed Graph Coloring: Fundamentals and Recent Developments}.
\newblock Synthesis Lectures on Distributed Computing Theory. Morgan {\&} Claypool Publishers, 2013.

\bibitem{BarenboimEK14}
Leonid Barenboim, Michael Elkin, and Fabian Kuhn.
\newblock Distributed (delta+1)-coloring in linear (in delta) time.
\newblock {\em {SIAM} J. Comput.}, 43(1):72--95, 2014.

\bibitem{BarenboimEPS12}
Leonid Barenboim, Michael Elkin, Seth Pettie, and Johannes Schneider.
\newblock The locality of distributed symmetry breaking.
\newblock In {\em 53rd Annual {IEEE} Symposium on Foundations of Computer Science, {FOCS} 2012, New Brunswick, NJ, USA, October 20-23, 2012}, pages 321--330. {IEEE} Computer Society, 2012.

\bibitem{BarenboimEPS16}
Leonid Barenboim, Michael Elkin, Seth Pettie, and Johannes Schneider.
\newblock The locality of distributed symmetry breaking.
\newblock {\em J. {ACM}}, 63(3):20:1--20:45, 2016.

\bibitem{Bollobas80}
Béla Bollobás.
\newblock A probabilistic proof of an asymptotic formula for the number of labelled regular graphs.
\newblock {\em European Journal of Combinatorics}, 1(4):311--316, 1980.

\bibitem{Bollobas88}
Béla Bollobás.
\newblock The isoperimetric number of random regular graphs.
\newblock {\em European Journal of Combinatorics}, 9(3):241--244, 1988.

\bibitem{Brandt19}
Sebastian Brandt.
\newblock An automatic speedup theorem for distributed problems.
\newblock In {\em {PODC}}, pages 379--388. {ACM}, 2019.

\bibitem{BrandtCGGRV22}
Sebastian Brandt, Yi-Jun Chang, Jan Greb{\'\i}k, Christoph Grunau, V\'{a}clav Rozho\v{n}, and Zolt\'{a}n Vidny\'{a}nszky.
\newblock {Local Problems on Trees from the Perspectives of Distributed Algorithms, Finitary Factors, and Descriptive Combinatorics}.
\newblock In Mark Braverman, editor, {\em 13th Innovations in Theoretical Computer Science Conference (ITCS 2022)}, volume 215 of {\em Leibniz International Proceedings in Informatics (LIPIcs)}, pages 29:1--29:26, Dagstuhl, Germany, 2022. Schloss Dagstuhl -- Leibniz-Zentrum f{\"u}r Informatik.

\bibitem{BrandtFHKLRSU16}
Sebastian Brandt, Orr Fischer, Juho Hirvonen, Barbara Keller, Tuomo Lempi{\"{a}}inen, Joel Rybicki, Jukka Suomela, and Jara Uitto.
\newblock A lower bound for the distributed lov{\'{a}}sz local lemma.
\newblock In {\em STOC}, pages 479--488. {ACM}, 2016.

\bibitem{BO20}
Sebastian Brandt and Dennis Olivetti.
\newblock Truly tight-in-{\(\Delta\)} bounds for bipartite maximal matching and variants.
\newblock In Yuval Emek and Christian Cachin, editors, {\em {PODC}}, pages 69--78. {ACM}, 2020.

\bibitem{Ghaffari16}
Mohsen Ghaffari.
\newblock An improved distributed algorithm for maximal independent set.
\newblock In Robert Krauthgamer, editor, {\em Proceedings of the Twenty-Seventh Annual {ACM-SIAM} Symposium on Discrete Algorithms, {SODA} 2016, Arlington, VA, USA, January 10-12, 2016}, pages 270--277. {SIAM}, 2016.

\bibitem{0001G23}
Mohsen Ghaffari and Christoph Grunau.
\newblock Faster deterministic distributed {MIS} and approximate matching.
\newblock In Barna Saha and Rocco~A. Servedio, editors, {\em Proceedings of the 55th Annual {ACM} Symposium on Theory of Computing, {STOC} 2023, Orlando, FL, USA, June 20-23, 2023}, pages 1777--1790. {ACM}, 2023.

\bibitem{ecomposition}
Mohsen Ghaffari and Christoph Grunau.
\newblock Near-optimal deterministic network decomposition and ruling set, and improved {MIS}.
\newblock In {\em 65th {IEEE} Annual Symposium on Foundations of Computer Science, {FOCS} 2024, Chicago, IL, USA, October 27-30, 2024}, pages 2148--2179. {IEEE}, 2024.

\bibitem{0001GHIR23}
Mohsen Ghaffari, Christoph Grunau, Bernhard Haeupler, Saeed Ilchi, and V{\'{a}}clav Rozhon.
\newblock Improved distributed network decomposition, hitting sets, and spanners, via derandomization.
\newblock In Nikhil Bansal and Viswanath Nagarajan, editors, {\em Proceedings of the 2023 {ACM-SIAM} Symposium on Discrete Algorithms, {SODA} 2023, Florence, Italy, January 22-25, 2023}, pages 2532--2566. {SIAM}, 2023.

\bibitem{GhaffariP19}
Mohsen Ghaffari and Julian Portmann.
\newblock Improved network decompositions using small messages with applications on mis, neighborhood covers, and beyond.
\newblock In Jukka Suomela, editor, {\em 33rd International Symposium on Distributed Computing, {DISC} 2019, October 14-18, 2019, Budapest, Hungary}, volume 146 of {\em LIPIcs}, pages 18:1--18:16. Schloss Dagstuhl - Leibniz-Zentrum f{\"{u}}r Informatik, 2019.

\bibitem{HarrisSS18}
David~G. Harris, Johannes Schneider, and Hsin{-}Hao Su.
\newblock Distributed ({\(\Delta\)} +1)-coloring in sublogarithmic rounds.
\newblock {\em J. {ACM}}, 65(4):19:1--19:21, 2018.

\bibitem{Janson04}
Svante Janson.
\newblock Large deviations for sums of partly dependent random variables.
\newblock {\em Random Structures \& Algorithms}, 24(3):234--248, 2004.

\bibitem{UBMIS}
Seri Khoury and Aaron Schild.
\newblock Breaking barriers for distributed mis by faster degree reduction.
\newblock Under submission, 2025.

\bibitem{Kuhn09}
Fabian Kuhn.
\newblock Weak graph colorings: distributed algorithms and applications.
\newblock In Friedhelm~Meyer auf~der Heide and Michael~A. Bender, editors, {\em {SPAA} 2009: Proceedings of the 21st Annual {ACM} Symposium on Parallelism in Algorithms and Architectures, Calgary, Alberta, Canada, August 11-13, 2009}, pages 138--144. {ACM}, 2009.

\bibitem{KuhnMW04}
Fabian Kuhn, Thomas Moscibroda, and Roger Wattenhofer.
\newblock What cannot be computed locally!
\newblock In Soma Chaudhuri and Shay Kutten, editors, {\em Proceedings of the Twenty-Third Annual {ACM} Symposium on Principles of Distributed Computing, {PODC} 2004, St. John's, Newfoundland, Canada, July 25-28, 2004}, pages 300--309. {ACM}, 2004.

\bibitem{abs-1011-5470}
Fabian Kuhn, Thomas Moscibroda, and Roger Wattenhofer.
\newblock Local computation: Lower and upper bounds.
\newblock {\em CoRR}, abs/1011.5470, 2010.

\bibitem{KuhnMW16}
Fabian Kuhn, Thomas Moscibroda, and Roger Wattenhofer.
\newblock Local computation: Lower and upper bounds.
\newblock {\em J. {ACM}}, 63(2):17:1--17:44, 2016.

\bibitem{LampertRZ18}
Christoph~H. Lampert, Liva Ralaivola, and Alexander Zimin.
\newblock Dependency-dependent bounds for sums of dependent random variables, 2018.

\bibitem{LenzenW11}
Christoph Lenzen and Roger Wattenhofer.
\newblock {MIS} on trees.
\newblock In Cyril Gavoille and Pierre Fraigniaud, editors, {\em Proceedings of the 30th Annual {ACM} Symposium on Principles of Distributed Computing, {PODC} 2011, San Jose, CA, USA, June 6-8, 2011}, pages 41--48. {ACM}, 2011.

\bibitem{Linial92}
Nathan Linial.
\newblock Locality in distributed graph algorithms.
\newblock {\em {SIAM} J. Comput.}, 21(1):193--201, 1992.

\bibitem{Luby86}
Michael Luby.
\newblock A simple parallel algorithm for the maximal independent set problem.
\newblock {\em {SIAM} J. Comput.}, 15(4):1036--1053, 1986.

\bibitem{McKayWW04}
Brendan~D. McKay, Nicholas~C. Wormald, and Beata Wysocka.
\newblock Short cycles in random regular graphs.
\newblock {\em Electronic Journal of Combinatorics}, 11(1):R66, 2004.

\bibitem{MetivierRSZ11}
Yves M{\'{e}}tivier, John~Michael Robson, Nasser Saheb{-}Djahromi, and Akka Zemmari.
\newblock An optimal bit complexity randomized distributed {MIS} algorithm.
\newblock {\em Distributed Comput.}, 23(5-6):331--340, 2011.

\bibitem{Naor91}
Moni Naor.
\newblock A lower bound on probabilistic algorithms for distributive ring coloring.
\newblock {\em {SIAM} J. Discret. Math.}, 4(3):409--412, 1991.

\bibitem{Olivetti20}
Dennis Olivetti.
\newblock Brief announcement: Round eliminator: a tool for automatic speedup simulation.
\newblock In Yuval Emek and Christian Cachin, editors, {\em {PODC} '20: {ACM} Symposium on Principles of Distributed Computing, Virtual Event, Italy, August 3-7, 2020}, pages 352--354. {ACM}, 2020.

\bibitem{abs-2406-19430}
V{\'{a}}clav Rozhon.
\newblock Invitation to local algorithms.
\newblock {\em CoRR}, abs/2406.19430, 2024.

\bibitem{RozhonG20}
V{\'{a}}clav Rozhon and Mohsen Ghaffari.
\newblock Polylogarithmic-time deterministic network decomposition and distributed derandomization.
\newblock In Konstantin Makarychev, Yury Makarychev, Madhur Tulsiani, Gautam Kamath, and Julia Chuzhoy, editors, {\em Proceedings of the 52nd Annual {ACM} {SIGACT} Symposium on Theory of Computing, {STOC} 2020, Chicago, IL, USA, June 22-26, 2020}, pages 350--363. {ACM}, 2020.

\bibitem{SchneiderW08}
Johannes Schneider and Roger Wattenhofer.
\newblock A log-star distributed maximal independent set algorithm for growth-bounded graphs.
\newblock In Rida~A. Bazzi and Boaz Patt{-}Shamir, editors, {\em Proceedings of the Twenty-Seventh Annual {ACM} Symposium on Principles of Distributed Computing, {PODC} 2008, Toronto, Canada, August 18-21, 2008}, pages 35--44. {ACM}, 2008.

\bibitem{Suomela13}
Jukka Suomela.
\newblock Survey of local algorithms.
\newblock {\em {ACM} Comput. Surv.}, 45(2):24:1--24:40, 2013.

\bibitem{wattenhofer2020mastering}
R.~Wattenhofer.
\newblock {\em Mastering Distributed Algorithms}.
\newblock Independently Published, 2020.

\bibitem{wattenhofer_podc}
Roger Wattenhofer.
\newblock Lecture notes on principles of distributed computing, chapter 4.
\newblock \url{https://disco.ethz.ch/courses/podc/lecturenotes/chapter4.pdf}, 2023.
\newblock ETH Zurich, Principles of Distributed Computing course.

\end{thebibliography}

\appendix

\section{ID Graph Construction}

We now construct the graph $G_{n,\Delta}$. This was done by prior work but we include this to prove additional properties of the resulting graph.

\begin{defn}[Ultra Log-Concave Random Variables \cite{AravindaMM21}]
A nonnegative-integer-valued random variable $X$ with probability mass function $p$ is called \emph{ultra log-concave} if it has contiguous support and satisfies the following property for any integer $n\ge 1$:

$$p(n)^2 \ge \frac{n+1}{n}p(n-1)p(n+1)$$
\end{defn}

\begin{thm}[Corollary 1.1 of \cite{AravindaMM21}]\label{thm:ultra-logconcave}
Let $X$ be an ultra log-concave random variable. Then, for all $t\ge 0$,

$$\Pr[X - \E[X] \ge t] \le e^{-\frac{t^2}{2(t+\E[X])}}$$

and

$$\Pr[X - \E[X] \le -t] \le e^{-\frac{t^2}{2\E[X]}}$$
\end{thm}

\begin{prop}\label{prop:intersection-concentration}
Consider an even positive integer $n$, $G = K_n$, and a set $S\subseteq V(G)$ with $k := |S|$ and $n := |V(G)|$. Let $\mu := \frac{k(k-1)}{2(n-1)}$. Let $\mc P$ denote the set of perfect matchings of $G$. For any $\delta > 0$,

$$\Pr_{M\sim \mc P}[|M\cap E(G[S])| \ge (1 + \delta)\mu] \le e^{-\frac{\delta^2\mu}{2+2\delta}}$$

For any $\delta\in (0,1)$,

$$\Pr_{M\sim \mc P}[|M\cap E(G[S])| \le (1 - \delta)\mu] \le e^{-\frac{\delta^2\mu}{2}}$$

\end{prop}

\begin{proof}
We start by showing that $|M\cap E(G[S])|$ is an ultra log-concave random variable. Let $t$ be a nonnegative integer for which $0\le k-2t \le \min(k,n-k)$. Recall the double factorial notation $\ell!! := (\ell-1)(\ell-3)\hdots (3)(1)$ for an even nonnegative integer $\ell$. The number of perfect matchings in $G$ is exactly $n!!$. The number of perfect matchings $M$ in $G$ for which $|M\cap E(G[S])| = t$ is exactly

$$\binom{k}{k-2t}\binom{n-k}{k-2t}(k-2t)!(2t)!!(n-2k+2t)!!$$

obtained by counting the number of options for unmatched vertex subsets of $S$ and $V\setminus S$ and selecting a matching within each set. Define

$$f(t) := \Pr_{M\sim \mc P}[|M\cap E(G[S])| = t]$$

$f(t)$ has the following form:

\begin{align*}
f(t) &= \frac{\binom{k}{k-2t}\binom{n-k}{k-2t}(k-2t)!(2t)!!(n-2k+2t)!!}{n!!}\\
&= \frac{k!(n-k)!(k-2t)!(2t)!(n-2k+2t)!2^{n/2}(n/2)!}{((k-2t)!)^2(2t)!(n-2k+2t)!2^t t! 2^{n/2-k+t}(n/2-k+t)!n!}\\
&= \frac{k!(n-k)!(n/2)!2^{k-2t}}{(k-2t)!t! (n/2-k+t)!n!}\\
\end{align*}

As a result,

$$f(t+1) = \frac{(k-2t)(k-2t-1)}{4(t+1)(n/2-k+t+1)}f(t)$$

and

$$f(t-1) = \frac{4t(n/2-k+t)}{(k-2t+2)(k-2t+1)}f(t)$$

so

$$f(t-1)f(t+1) = \frac{(k-2t)(k-2t-1)}{4(t+1)(n/2-k+t+1)}\frac{4t(n/2-k+t)}{(k-2t+2)(k-2t+1)}f(t)^2 \le \frac{t}{t+1}f(t)^2$$

so $|M\cap E(G[S])|$ is indeed ultra log-concave. Furthermore,

$$\E_{M\sim \mc P}[|M\cap E(G[S])|] = \mu$$

so applying Theorem \ref{thm:ultra-logconcave} yields the desired result.
\end{proof}

We construct the lower bound graph via the probabilistic method, using the \emph{configuration model} of regular random graphs \cite{Bollobas80}:

\begin{defn}[Configuration Model \cite{Bollobas80}]
For an even positive integer $n$ and a positive integer $\Delta$, let $\mc C(n,\Delta)$ denote a distribution over graphs, called the \emph{configuration model}, generated as follows:

\begin{enumerate}
    \item Define $n$ disjoint clusters $U_1, U_2, \hdots, U_n$, where $U_i = \{u_{i,1},u_{i,2},\hdots,u_{i,\Delta}\}$ is a set of size $\Delta$ for all $i\in [n]$.
    \item Let $M$ be a uniformly random perfect matching on $U_1\cup U_2\cup\hdots\cup U_n$.
    \item Define the sampled graph $G_M$ deterministically using $M$ as follows:
    \begin{enumerate}
        \item Let $V(G) := \{v_1,v_2,\hdots,v_n\}$.
        \item For any $\{u_{i,a}, u_{j,b}\} \in M$, add the edge $\{v_i, v_j\}$ to $E(G)$ if $i\ne j$.
    \end{enumerate}
\end{enumerate}
\end{defn}

Note that a graph $G_M$ sampled from $\mc C(n,\Delta)$ may not be regular, as $M$ can contain both (a) intracluster edges (cycles of length 1) and (b) multiple edges between the same pair of clusters (cycles of length 2). It has been shown that graphs sampled from $\mc C(n,\Delta)$ have a substantial probability of not only being $\Delta$-regular, but also having high girth:

\begin{thm}[Corollary 1.1 of \cite{McKayWW04}]\label{thm:config-model-girth}
For $(\Delta-1)^{2g-1} = o(n)$, the probability that $G_M\sim \mc C(n,\Delta)$ is $\Delta$-regular and has girth greater than $g$ is at least\footnote{In their work, they sum from $r=3$ rather than $r=1$; our definition of $\mc C(n,\Delta)$ slight differs in that they are sampling uniformly from the set of $\Delta$-regular graphs.}

$$\exp\left(-\sum_{r=1}^g \frac{(\Delta-1)^r}{2r}\right)$$
\end{thm}

\lemconfigmodel*

\begin{proof}
We argue that sampling $G_M \sim \mc C(n,\Delta)$ produces a graph with the desired properties with positive probability. Fix a set $I\subseteq [n]$. Let $S_I := \{v_i : \text{ for all } i\in I\}$ and $U_I := \{u_{i,a} \text{ for all } i\in I, a\in [\Delta]$.

First, consider the case when $|I| \le \frac{10^6n\ln \Delta}{\Delta}$. Apply Proposition \ref{prop:intersection-concentration} with $n\gets n\Delta$, $S\gets U_I$, and $\delta \gets \frac{10^{12}(\ln (n/|I|))|I|}{(|I|\Delta)^2/(2n\Delta)} = \frac{2(10)^{12}(\ln(n/|I|))n}{|I|\Delta} > 1$ to show that 

$$\Pr_{M\sim \mc P}\left[|M\cap E(G[U_I])| \ge 10^8(\ln(n/|I|))|I|\right] \le e^{-10^8(\ln(n/|I|))|I|/3} \le (|I|/n)^{30|I|}$$

Next, consider the case when $|I| > \frac{10^6n\ln \Delta}{\Delta}$. Apply Proposition \ref{prop:intersection-concentration} with $n\gets n\Delta$, $S\gets U_I$, and $\delta\gets 1/25$ to show that

$$\Pr_{M\sim \mc P}\left[\left||M\cap E(G[U_I])| - \frac{(|I|\Delta)^2}{2n\Delta}\right| \ge \frac{(|I|\Delta)^2}{50n\Delta}\right] \le 2e^{-(1/25)^2\frac{(|I|\Delta)^2}{6n\Delta}} \le e^{-10^3|I|\ln\Delta} = \Delta^{-10^3|I|}$$

Now, we do a union bound over all sets to get the desired guarantee for all sets. Specifically, there exists a set that does not satisfy the necessary guarantee with probability at most

\begin{align*}
&\sum_{k=\sqrt{n}}^{10^6n(\log\Delta)/\Delta} \binom{n}{k}(k/n)^{30k} + \sum_{k=10^6n(\log\Delta)/\Delta}^n \binom{n}{k}\Delta^{-10^3k}\\ \le& \sum_{k=\sqrt{n}}^{10^6n(\log\Delta)/\Delta} (k/n)^{29k} + \sum_{k=10^6n(\log\Delta)/\Delta}^n \Delta^{-10^3k+k}\\ \le& \exp(-\sqrt{n})
\end{align*}

But Theorem \ref{thm:config-model-girth} implies that the first two conditions are satisfied with $g\gets (\log_{\Delta} n)/1000$ with probability at least $\exp(-n^{1/100}) > \exp(-\sqrt{n})$. Thus, by a union bound, all properties are satisfied with positive probability, so a graph exists, as desired.
\end{proof}

\section{Details about Distributions}\label{sec:AppDistributions}

As mentioned earlier, conditioning on measure 0 sets needs to be defined due to the Borel-Kolmogorov paradox. This paradox says that if one conditions on two different great circles in a sphere, one gets conditional distributions with different probability density functions. One resolves the paradox by defining a coordinate system. We do this here. Furthermore, our conditionings in this paper are simple (just coordinate substitutions in the standard basis). As a result, we can define conditional distributions in our paper based on coordinate substitutions rather than limits. We do this in this section and prove the necessary properties here.

\subsection{Definitions of Relevant Conditional Distributions}\label{sec:ConditionalDefns}

To define the desired conditional probabilities, we need to discuss probability measures on $r$-neighborhoods and $r$-flowers.

\begin{defn}[Probability Measures on Flowers and Neighborhoods]
Let $\mu_0: 2^{[0,1]}\rightarrow [0,1]$ denote the standard uniform Lesbesgue measure on $\mc S_0 = [0,1]$. For each $r\ge 1$, let $\mu_r := \mu_{r-1}^{\Delta-1}$ be the $(\Delta-1)$-th power (product measure) of $\mu_{r-1}$. $\mu_r$ is the uniform measure over the continuous set $\mc S_r$. For $r\ge 0$, let $\phi_r$ and $\rho_r$ be the uniform measures over $\mc F_r$ and $\mc R_r$ respectively; more formally the product measures $\phi_r := \mu_0\times\mu_1^2\times\mu_2^2\times\hdots\times\mu_{r-1}^2\times\mu_r^2$ and $\rho_r := \mu_0^{\Delta}\times\mu_1^{\Delta}\times\hdots\times\mu_{r-1}^{\Delta}$.
\end{defn}

\begin{defn}[Expectation of a Random Variable according to a Measure]
For a probability space $(\Omega,\mc F,\mu)$\footnote{$\Omega$ is a set, $\mc F$ is a $\sigma$-algebra on that set, and $\mu:\mc F\rightarrow [0,1]$ is a probability measure on the $\sigma$-algebra.}, and a random variable $X:\Omega\rightarrow\mathbb R$ let $\displaystyle\E_{x\sim \mu}[X(x)]$ denote the expectation of the random variable $X$ against the measure $\mu$.
\end{defn}

Note that the expectations of a random variable with respect to two equal measures are equal.

\begin{defn}[Conditional Probability Measures on Flowers and Neighborhoods]
For any $r\ge 1$, $i\in [\Delta]$, and $x\in \mc F_{r-1}$, let $\mc R_r(i,x)$ denote the set of all $y\in \mc R_r$ for which $\text{res}_i(y) = x$ and let $\rho_{r,i}(x) := (\textbf{1}^{\Delta})^{r-1}\times(\mu_{r-1}^{i-1}\times\textbf{1}\times\mu_{r-1}^{\Delta-i})$ be the uniform measure on $\mc R_r(i,x)$, with $\textbf{1}$ denoting the uniform measure over the corresponding (singleton) set. For $r\ge 0$, $v\in \{A,B\}$, and $x\in \mc R_r$, let $\mc F_r(v,x)$ denote the set of all $y\in \mc F_r$ for which $\text{end}_v(y) = x$, let $\phi_{r,A}(x) := \textbf{1}\times(\textbf{1}^2)^{r-1}\times(\textbf{1}\times\mu_r)$ (uniform measure over $\mc F_r(A,x)$), and let $\phi_{r,B}(x) := \textbf{1}\times(\textbf{1}^2)^{r-1}\times(\mu_r\times\textbf{1})$ (uniform measure over $\mc F_r(B,x)$).
\end{defn}

\propcondresdefn*

\begin{proof}
Define it as

$$\E_{y\sim \mc R_r}[X(y) \mid \text{res}_i(y) = x] := \E_{y\sim \rho_{r,i}(x)}[X(y)]$$
\end{proof}

\propcondenddefn*

\begin{proof}
Define it as

$$\E_{y\sim \mc F_r}[X(y) \mid \text{end}_v(y) = x] := \E_{y\sim \phi_{r,v}(x)}[X(y)]$$
\end{proof}

\subsection{Proofs of Distributional Equivalences}

We now prove the desired distributional equivalences:

\propedgeflip*

\begin{proof}
It suffices to prove this result when $X$ is the indicator function of a set that can be written as a direct product; i.e. $X = \textbf{1}_S$ for some $S\subseteq \mc F_r$ for which there exist sets $S_0\subseteq \mc S_0$ and $S_{sA}, S_{sB}\subseteq \mc S_s$ for all $s\in [r]$ with the property that $S = S_0 \times \prod_{i=1}^r (S_{sA}\times S_{sB})$.\footnote{This is because product sets generate the product $\sigma$-algebra, so the result then can be shown for arbitrary set indicators (by the Caratheodory extension theorem). Then, for any nonnegative random variable $X$, $\E[X] = \int_0^{\infty} \Pr[X\ge \phi]d\phi = \int_0^{\infty} \E[\textbf{1}_{X\ge \phi}]d\phi$, so the result holding on indicators of sets leads to the result holding on all nonnegative random variables. Arbitrary random variables can be handled by writing as a difference of two nonnegative variables.} Then relabeling yields

\begin{align*}
\E_{y\sim \mc F_r}[\textbf{1}_S(\ol{y})] &= \Pr_{y\sim \mc F_r}[\ol{y}\in S]\\
&= \Pr_{y_0\sim \mu_0}[y_0\in S_0]\prod_{s=1}^r\left(\Pr_{y_{sB}\sim \mu_s}[y_{sB}\in S_{sA}]\Pr_{y_{sA}\sim \mu_s}[y_{sA}\in S_{sB}]\right)\\
&= \Pr_{y_0\sim \mu_0}[y_0\in S_0]\prod_{s=1}^r\left(\Pr_{y_{sA}\sim \mu_s}[y_{sA}\in S_{sA}]\Pr_{y_{sB}\sim \mu_s}[y_{sB}\in S_{sB}]\right)\\
&= \Pr_{y\sim \mc F_r}[y\in S] = \E_{y\sim \mc F_r}[\textbf{1}_S(y)]
\end{align*}

as desired.
\end{proof}

\propvertpermute*

\begin{proof}
It suffices to check this result when $\mc D$ is a deterministic distribution over a single permutation $\sigma$. Furthermore, as in the previous proposition, it suffices to prove the result when $X$ is the indicator function of a product set; i.e. $X = \textbf{1}_S$ for some $S\subseteq\mc R_r$ for which there exist sets $S_{sj}\subseteq\mc S_{s-1}$ for all $s\in [r]$ and $j\in [\Delta]$ with the property that $S = \prod_{s=1}^r (S_{s1}\times S_{s2}\times\hdots\times S_{s\Delta})$. Then relabeling yields

\begin{align*}
\E_{y\sim \mc R_r}[\textbf{1}_S(\sigma(y))] &= \Pr_{y\sim \mc R_r}[\sigma(y) \in S]\\
&= \prod_{s=1}^r\prod_{j=1}^{\Delta}\Pr_{y\sim \rho_r}[\sigma(y)_{sj} \in S_{sj}]\\
&= \prod_{s=1}^r\prod_{j=1}^{\Delta}\Pr_{y\sim \rho_r}[y_{s\sigma^{-1}(j)} \in S_{sj}]\\
&= \prod_{s=1}^r\prod_{j=1}^{\Delta}\Pr_{y_{s\sigma^{-1}(j)}\sim \mu_{r-1}}[y_{s\sigma^{-1}(j)} \in S_{sj}]\\
&= \prod_{s=1}^r\prod_{j=1}^{\Delta}\Pr_{y_{sj}\sim \mu_{r-1}}[y_{sj} \in S_{sj}]\\
&= \E_{y\sim \mc R_r}[\textbf{1}_S(y)]
\end{align*}

as desired.
\end{proof}

\propcondvertpermute*

\begin{proof}
As in the proof of the previous result, it suffices to prove this result when $X = \textbf{1}_S$ for some $S \subseteq \mc R_r$ for which there exist sets $S_{sj}\subseteq \mc S_{s-1}$ for all $s\in [r]$ and $j\in [\Delta]$ with the property that $S = \prod_{s=1}^r (S_{s1}\times S_{s2}\times\hdots\times S_{s\Delta})$. Note that the definition of $\text{res}_i$ implies the following constraints:

\begin{enumerate}
    \item $x_0 = \text{res}_i(y)_0 = \text{res}_1(\sigma_i(y))_0 = \sigma_i(y)_{11} = y_{1i}$
    \item $x_{sB} = \text{res}_i(y)_{sB} = \text{res}_1(\sigma_i(y))_{sB} = \sigma_i(y)_{(s+1)1} = y_{(s+1)i}$ for all $s\in [r-1]$
    \item $(x_{sA})_j = (\text{res}_i(y)_{sA})_j = (\text{res}_1(\sigma_i(y))_{sA})_j = \sigma_i(y)_{s(j+1)} = y_{s\sigma_i(j+1)}$ for all $s\in [r-1]$, $j\in [\Delta-1]$
\end{enumerate}

Re-notating leads to

\begin{enumerate}
    \item $y_{1i} = x_0$
    \item $y_{si} = x_{(s-1)B}$ for all $s\in \{2,3,\hdots,r\}$
    \item $y_{sj} = (x_{sA})_{\sigma_i(j)-1}$ for all $s\in [r-1]$, $j\in [\Delta]$ with $j\ne i$
\end{enumerate}

Similarly, for the right side, we get the following constraints on $z$:

\begin{enumerate}
    \item $z_{11} = x_0$
    \item $z_{s1} = x_{(s-1)B}$ for all $s\in \{2,3,\hdots,r\}$
    \item $z_{sj} = (x_{sA})_{j-1}$ for all $s\in [r-1]$, $j\in \{2,3,\hdots,\Delta\}$.
\end{enumerate}

Thus,

\begin{align*}
\E_{y\sim \mc R_r}[\textbf{1}_S(\sigma_i(y)) \mid \text{res}_i(y) = x] &= \Pr_{y\sim \rho_{r,i}(x)}[\sigma_i(y) \in S]\\
&= \prod_{s=1}^r \prod_{j=1}^{\Delta} \Pr_{y\sim \rho_{r,i}(x)}[\sigma_i(y)_{sj} \in S_{sj}]\\
&= \prod_{s=1}^r \prod_{j=1}^{\Delta} \Pr_{y\sim \rho_{r,i}(x)}[y_{s\sigma_i(j)} \in S_{sj}]\\
&= \prod_{s=1}^r \prod_{j=1}^{\Delta} \Pr_{y\sim \rho_{r,i}(x)}[y_{sj} \in S_{s\sigma_i(j)}]\\
&= \Pr_{y\sim \rho_{r,i}(x)}[y_{1i} \in S_{11}] \left(\prod_{s=2}^r \Pr_{y\sim \rho_{r,i}(x)}[y_{si} \in S_{s1}]\right)\\
&\left(\prod_{s=1}^{r-1}\prod_{j\in [\Delta], j\ne i} \Pr_{y\sim \rho_{r,i}(x)}[y_{sj} \in S_{s\sigma_i(j)}]\right)\prod_{j\in [\Delta], j\ne i}\Pr_{y\sim \rho_{r,i}(x)}[y_{rj} \in S_{r\sigma_i(j)}]\\
&= \textbf{1}_{x_0 \in S_{11}} \left(\prod_{s=2}^r \textbf{1}_{x_{(s-1)B} \in S_{s1}}\right)\\
&\left(\prod_{s=1}^{r-1}\prod_{j\in [\Delta], j\ne i} \textbf{1}_{(x_{sA})_{\sigma_i(j)-1} \in S_{s\sigma_i(j)}}\right)\prod_{j\in [\Delta], j\ne i}\mu_{r-1}(S_{r\sigma_i(j)})\\
&= \textbf{1}_{x_0 \in S_{11}} \left(\prod_{s=2}^r \textbf{1}_{x_{(s-1)B} \in S_{s1}}\right)\\
&\left(\prod_{s=1}^{r-1}\prod_{j\in [\Delta], j\ne 1} \textbf{1}_{(x_{sA})_{j-1} \in S_{sj}}\right)\prod_{j\in [\Delta], j\ne 1}\mu_{r-1}(S_{rj})\\
&= \Pr_{z\sim \rho_{r,1}(x)}[z\in S]\\
&= \E_{y\sim \mc R_r}[\textbf{1}_S(z) \mid \text{res}_1(z) = x]
\end{align*}

as desired.
\end{proof}

We now prove the desired distributional equivalences. Note that they would all be immediate consequences of the tower law of conditional expectations if conditioning were done on positive measure sets, but are not immediate given that we have defined conditioning on measure 0 sets using coordinate substitution.

\propftorsample*

\begin{proof}
It suffices to prove this result when $X$ is the indicator function $\textbf{1}_S$ of some product set $S := \prod_{s=1}^r (S_{s1}\times S_{s2}\times\hdots\times S_{s\Delta})$ for some sets $S_{sj}\subseteq \mc S_{s-1}$ with positive measure. By the definition of conditional expectation in the proof of Proposition \ref{prop:cond-res-defn},

$$\E_{y\sim \mc R_r}[\textbf{1}_S(y) \mid \text{res}_i(y) = x] = \Pr_{y\sim \rho_{r,i}(x)}[y\in S]$$

Let $\text{res}_i(S) = S_{1i}\times \prod_{s=1}^{r-1}(S_{(s+1)i}\times (S_{s2}\times S_{s3}\times\hdots\times S_{s(i-1)}\times S_{s1}\times S_{s(i+1)}\times\hdots\times S_{s\Delta}))$ denote the image of $S$ under the function $\text{res}_i$. Then

$$\Pr_{y\sim \rho_{r,i}(x)}[y\in S] = \textbf{1}_{x\in \text{res}_i(S)} \frac{1}{\mu_{r-1}(S_{ri})} \prod_{j=1}^{\Delta}\mu_{r-1}(S_{rj})$$

and

\begin{align*}
&\E_{x\sim \mc F_{r-1}}\left[\E_{y\sim \mc R_r}[\textbf{1}_S(y) \mid \text{res}_i(y) = x]\right]\\
&= \frac{1}{\mu_{r-1}(S_{ri})} \prod_{j=1}^{\Delta}\mu_{r-1}(S_{rj}) \Pr_{x\sim \mc F_{r-1}}[x\in \text{res}_i(S)]\\
&= \left(\frac{1}{\mu_{r-1}(S_{ri})} \prod_{j=1}^{\Delta}\mu_{r-1}(S_{rj})\right)\left(\mu_0(S_{1i})\prod_{s=1}^{r-1}\left(\mu_s(S_{(s+1)i})\left(\frac{\prod_{j=1}^{\Delta}\mu_{s-1}(S_{sj})}{\mu_{s-1}(S_{si})}\right)\right)\right)\\
&= \prod_{s=1}^r\prod_{j=1}^{\Delta}\mu_{s-1}(S_{sj})\\
&= \rho_r(S) = \E_{y\sim \mc R_r}[\textbf{1}_S(y)]
\end{align*}

as desired.
\end{proof}

\begin{prop}\label{prop:r-to-f-sample-no-perm}
For any positive integer $r$ and random variable $X:\mc F_{r-1}\rightarrow \mathbb R$,

$$\E_{x\sim \mc R_{r-1}}\left[\E_{y\sim \mc F_{r-1}}[X(y)\mid \text{end}_A(y) = x]\right] = \E_{y\sim \mc F_{r-1}}[X(y)]$$
\end{prop}

\begin{proof}
This fact is immediate when $r = 1$, so suppose that $r\ge 2$. Furthermore, it suffices to show this when $X = \textbf{1}_S$ for some product set $S\subseteq \mc F_{r-1}$; i.e. a set $S$ for which there exist sets $S_0\subseteq \mc S_0$, $S_{sB}\subseteq \mc S_s$ for all $s\ge 1$, and $S_{sAj}\subseteq \mc S_{s-1}$ for all $j\in \{2,3,\hdots,\Delta\}$ and $s \ge 1$ with the property that $S := S_0\times \prod_{s=1}^{r-1}\left(\left(S_{sA2}\times S_{sA3}\times\hdots\times S_{sA(\Delta-1)}\times S_{sA\Delta}\right)\times S_{sB}\right)$.
By the definition of conditional expectation in the proof of Proposition \ref{prop:cond-end-defn},

$$\E_{y\sim \mc F_{r-1}}[\textbf{1}_S(y) \mid \text{end}_A(y) = x] = \Pr_{y\sim \phi_{r,A}(x)}[y\in S]$$

Let $\text{end}_A(S) = S_0\times S_{1A2}\times S_{1A3}\times\hdots\times S_{1A\Delta}\times \prod_{s=2}^{r-1}\left(S_{(s-1)B}\times S_{sA2}\times S_{sA3}\times\hdots\times S_{sA\Delta}\right)$ denote the image of $S$ under $\text{end}_A$. Then

$$\Pr_{y\sim \phi_{r,A}(x)}[y\in S] = \textbf{1}_{x\in \text{end}_A(S)}\mu_{r-1}(S_{(r-1)B})$$

and

\begin{align*}
&\E_{x\sim \mc R_{r-1}}\left[\E_{y\sim \mc F_{r-1}}[\textbf{1}_S(y) \mid \text{end}_A(y) = x]\right]\\
&= \mu_{r-1}(S_{(r-1)B})\Pr_{x\sim \mc R_{r-1}}[x\in \text{end}_A(S)]\\
&= \mu_{r-1}(S_{(r-1)B})\left(\mu_0(S_0)\prod_{j=2}^{\Delta}\mu_0(S_{1Aj}) \prod_{s=2}^{r-1}\left(\mu_{s-1}(S_{(s-1)B})\prod_{j=2}^{\Delta}\mu_{s-1}(S_{sAj})\right)\right)\\
&= \mu_0(S_0)\prod_{s=1}^{r-1}\left(\mu_s(S_{sB})\prod_{j=2}^{\Delta}\mu_{s-1}(S_{sAj})\right)\\
&= \phi_{r-1}(S) = \E_{y\sim F_{r-1}}[\textbf{1}_S(y)]
\end{align*}

as desired.
\end{proof}

\proprtofsample*

\begin{proof}
By Proposition \ref{prop:vert-permute} with $\mc D$ being the uniform distribution over the permutations $\sigma_i$ for $i\in [\Delta]$ and applied to the random variable $Z(w) := \displaystyle\E_{y\sim \mc F_{r-1}}[X(y)\mid \text{end}_A(y) = w]$ for $w\in \mc R_{r-1}$, 

$$\E_{x\sim \mc R_{r-1}}\left[\E_{i\sim [\Delta]}\left[\E_{y\sim \mc F_{r-1}}[X(y)\mid \text{end}_A(y) = \sigma_i(x)]\right]\right] = \E_{x\sim \mc R_{r-1}}\left[\E_{y\sim \mc F_{r-1}}[X(y)\mid \text{end}_A(y) = x]\right]$$

Thus, applying Proposition \ref{prop:r-to-f-sample-no-perm} yields the desired result.
\end{proof}

\propftofsample*

\begin{proof}
For any $z\in \mc R_r$, define

$$X(z) = X_A(z)\E_{y\sim \mc F_r}[X_B(\text{end}_B(y)) \mid \text{end}_A(y) = z]$$

By Proposition \ref{prop:f-to-r-sample} for this choice of $X$ (first equality) and Proposition \ref{prop:r-to-f-sample-no-perm} for the random variable $Z(y) = X_A(\text{end}_A(y))X_B(\text{end}_B(y))$ for $y\in \mc F_r$ (last equality),

\begin{align*}
&\E_{x\sim \mc F_{r-1}}\left[\E_{z_A\sim \mc R_r}\left[X_A(z_A)\E_{y\sim \mc F_r}\left[X_B(\text{end}_B(y)) \mid \text{end}_A(y) = z_A\right]\middle| \text{ }\text{res}_1(z_A) = x\right]\right]\\
&= \E_{z_A\sim \mc R_r}\left[X_A(z_A)\E_{y\sim \mc F_r}\left[X_B(\text{end}_B(y)) \mid \text{end}_A(y) = z_A\right]\right]\\
&= \E_{z_A\sim \mc R_r}\left[\E_{y\sim \mc F_r}\left[X_A(\text{end}_A(y))X_B(\text{end}_B(y)) \mid \text{end}_A(y) = z_A\right]\right]\\
&= \E_{y\sim \mc F_r}\left[X_A(\text{end}_A(y))X_B(\text{end}_B(y))\right]
\end{align*}

Thus, to prove the desired result, it suffices to show that

$$\E_{y\sim \mc F_r}\left[X_B(\text{end}_B(y)) \mid \text{end}_A(y) = z\right] = \E_{w\sim \mc R_r}\left[X_B(w) \mid \text{res}_1(w) = \ol{x}\right]$$

for any $z\in \mc R_r$ for which $\text{res}_1(z) = x$. Intuitively, this holds because the part of $X_B$ only being a function of $\text{end}_B(y)$ means that all edges at distance $(r-1)$ to $r$ from $A$ as part of the conditioning $\text{end}_A(y) = z$ are irrelevant. To show the desired equality formally, it suffices to show this when $X_B = \textbf{1}_S$ for some product set $S = \prod_{s=1}^r\prod_{j=1}^{\Delta}S_{sj}$, where $S_{sj}\in \mc S_{s-1}$ for all $s\in [r]$ and $j\in [\Delta]$. Furthermore, for all $s\ge 2$, we may only consider $S_{sj}$s for which $S_{sj} = S_{sj1}\times S_{sj2}\times\hdots\times S_{sj(\Delta-1)}$ for some sets $S_{sjk} \in \mc S_{s-2}$. By definition,

$$\E_{y\sim \mc F_r}\left[\textbf{1}_S(\text{end}_B(y)) \mid \text{end}_A(y) = z\right] = \Pr_{y\sim \phi_{r,A}(z)}[\text{end}_B(y)\in S]$$

By definition of $\text{end}_B(y)$ (second equality), $\phi_{r,A}(z)$ (fourth equality), $\text{res}_1(z)$ (fifth equality), and $\rho_{1,A}(\ol{x})$ (seventh equality)

\begin{align*}
&\Pr_{y\sim \phi_{r,A}(z)}[\text{end}_B(y)\in S]\\
&= \prod_{s=1}^r\prod_{j=1}^{\Delta} \Pr_{y\sim \phi_{r,A}(z)}[\text{end}_B(y)_{sj} \in S_{sj}]\\
&= \Pr_{y\sim \phi_{r,A}(z)}[y_0 \in S_{11}]\left(\prod_{s=2}^r \Pr_{y\sim \phi_{r,A}(z)}[y_{(s-1)A} \in S_{s1}]\right)\left(\prod_{s=1}^r\prod_{j=2}^{\Delta}\Pr_{y\sim \phi_{r,A}(z)}[(y_{sB})_{j-1} \in S_{sj}]\right)\\
&= \Pr_{y\sim \phi_{r,A}(z)}[y_0 \in S_{11}]\left(\prod_{s=2}^r\prod_{k=1}^{\Delta-1} \Pr_{y\sim \phi_{r,A}(z)}[(y_{(s-1)A})_k \in S_{s1k}]\right)\left(\prod_{s=1}^r\prod_{j=2}^{\Delta}\Pr_{y\sim \phi_{r,A}(z)}[(y_{sB})_{j-1} \in S_{sj}]\right)\\
&= \textbf{1}_{z_{11}\in S_{11}}\left(\prod_{s=2}^r\prod_{k=1}^{\Delta-1}\textbf{1}_{z_{(s-1)(k+1)}\in S_{s1k}}\right)\left(\prod_{s=1}^{r-1}\prod_{j=2}^{\Delta}\textbf{1}_{(z_{(s+1)1})_{j-1} \in S_{sj}}\right)\left(\prod_{j=2}^{\Delta}\mu_{r-1}(S_{rj})\right)\\
&= \textbf{1}_{x_0\in S_{11}}\left(\prod_{s=2}^r\prod_{k=1}^{\Delta-1}\textbf{1}_{(x_{(s-1)A})_k\in S_{s1k}}\right)\left(\prod_{s=1}^{r-1}\prod_{j=2}^{\Delta}\textbf{1}_{(x_{sB})_{j-1} \in S_{sj}}\right)\left(\prod_{j=2}^{\Delta}\mu_{r-1}(S_{rj})\right)\\
&= \textbf{1}_{\ol{x}_0\in S_{11}}\left(\prod_{s=2}^r\prod_{k=1}^{\Delta-1}\textbf{1}_{(\ol{x}_{(s-1)B})_k\in S_{s1k}}\right)\left(\prod_{s=1}^{r-1}\prod_{j=2}^{\Delta}\textbf{1}_{(\ol{x}_{sA})_{j-1} \in S_{sj}}\right)\left(\prod_{j=2}^{\Delta}\mu_{r-1}(S_{rj})\right)\\
&= \Pr_{w\sim \rho_{r,1}(\ol{x})}[w_{11}\in S_{11}]\left(\prod_{s=2}^r\prod_{k=1}^{\Delta-1}\Pr_{w\sim \rho_{r,1}(\ol{x})}[(w_{s1})_k\in S_{s1k}]\right)\times\\
&\left(\prod_{s=1}^{r-1}\prod_{j=2}^{\Delta}\Pr_{w\sim \rho_{r,1}(\ol{x})}[w_{sj} \in S_{sj}]\right)\left(\prod_{j=2}^{\Delta}\Pr_{w\sim \rho_{r,1}(\ol{x})}[w_{rj} \in S_{rj}]\right)\\
&= \Pr_{w\sim \rho_{r,1}(\ol{x})}[w\in S] = \E_{w\sim \mc R_r}[\textbf{1}_S(w) \mid \text{res}_1(w) = \ol{x}]
\end{align*}

This is the desired equality.
\end{proof}

\proprtorsample*

\begin{proof}
\textbf{Definition of glue:} Define $\text{glue}(z_1,z_2,\hdots,z_{\Delta})$ as follows:

\begin{enumerate}
    \item $\text{glue}(z_1,z_2,\hdots,z_{\Delta})_{ri} := (z_i)_{(r-1)B}$ if $r \ge 2$ or $(z_i)_0$ if $r = 1$ for all $i\in [\Delta]$.
    \item $\text{glue}(z_1,z_2,\hdots,z_{\Delta})_{si} := x_{si}$ for all $s\in [r-1]$, $i\in [\Delta]$.
\end{enumerate}

We now check that this choice has the desired property coordinatewise. Let $y := \text{glue}(z_1,z_2,\hdots,z_{\Delta})$. If $r = 1$, then for any $i\in [\Delta]$,

$$\text{res}_i(y)_0 = y_{1i} = (z_i)_0$$

so $\text{res}_i(y) = z_i$ as desired. Now, consider $r > 1$. In this case, for all $i\in [\Delta]$,

$$(\text{res}_i(y))_0 = \sigma_i(y)_{11} = y_{1i} = x_{1i} = \sigma_i(\text{end}_A(z_i))_{1i} = \text{end}_A(z_i)_{11} = (z_i)_0$$

for all $s\in [r-1]$, $i\in [\Delta]$, and $j\in [\Delta-1]$,

$$((\text{res}_i(y))_{sA})_j = \sigma_i(y)_{s(j+1)} = y_{s\sigma_i(j+1)} = x_{s\sigma_i(j+1)} = \sigma_i(\text{end}_A(z_i))_{s\sigma_i(j+1)} = \text{end}_A(z_i)_{s(j+1)} = ((z_i)_{sA})_j$$

for all $s\in [r-2]$ and $i\in [\Delta]$,

$$(\text{res}_i(y))_{sB} = \sigma_i(y)_{(s+1)1} = y_{(s+1)i} = x_{(s+1)i} = \sigma_i(\text{end}_A(z_i))_{(s+1)i} = \text{end}_A(z_i)_{(s+1)1} = (z_i)_{sB}$$

and for all $i\in [\Delta]$,

$$(\text{res}_i(y))_{(r-1)B} = \sigma_i(y)_{r1} = y_{ri} = (z_i)_{(r-1)B}$$

This covers all coordinates, so $\text{res}_i(y) = z_i$ for all $i\in [\Delta]$ as desired.\\

\textbf{Uniqueness of glue:} Suppose, for the sake of contradiction, that there exist $x\in \mc R_{r-1}$, $z_1,z_2,\hdots,z_{\Delta}\in \mc F_{r-1}$ with $\text{end}_A(z_i) = \sigma_i(x)$ for all $i\in [\Delta]$, and $y\ne y'\in \mc R_r$ for which $\text{res}_i(y) = \text{res}_i(y') = z_i$ for all $i\in [\Delta]$. Note that for any $s\in \{2,3,\hdots,r\}$ and $i\in [\Delta]$,

\begin{align*}
y_{si} &= (\sigma_i(y))_{s1} = \text{res}_1(\sigma_i(y))_{(s-1)B} = \text{res}_i(y)_{(s-1)B}\\
&= \text{res}_i(y')_{(s-1)B} = \text{res}_1(\sigma_i(y'))_{(s-1)B} = (\sigma_i(y'))_{s1} = y_{si}'\\
\end{align*}

and

\begin{align*}
y_{1i} &= (\sigma_i(y))_{11} = \text{res}_1(\sigma_i(y))_0 = \text{res}_i(y)_0\\
&= \text{res}_i(y')_0 = \text{res}_1(\sigma_i(y'))_0 = (\sigma_i(y'))_{11} = y_{1i}'\\
\end{align*}

This covers all coordinates of $y$ and $y'$, so $y = y'$, a contradiction.\\

\textbf{Expectation property:} It suffices to prove the desired result when $X = \textbf{1}_S$ for a set $S := \prod_{s=1}^r\prod_{j=1}^{\Delta} S_{sj}$ for $S_{sj}\in \mc S_{s-1}$. By the definition of conditional expectation in the proof of Proposition \ref{prop:cond-end-defn} and the definition of $\text{glue}$\footnote{This reasoning works when $r\ge 2$; when $r=1$, replace $\zeta_{(r-1)B}$ with $\zeta_0$.},

\begin{align*}
&\E_{z_j\sim \mc F_{r-1}\text{ }\forall j\in [\Delta]}\left[\textbf{1}_S(\text{glue}(z_1,\hdots,z_{\Delta})) \mid (z_i = \zeta)\text{ and }(\text{end}_A(z_j) = \sigma_j(x)\text{ for all }j\in [\Delta])\right]\\
&= \Pr_{z_i = \zeta,\text{ }z_j\sim \phi_{r-1,A}(\sigma_j(x))\text{ }\forall j\ne i}[\text{glue}(z_1,\hdots,z_{\Delta}) \in S]\\
&= \prod_{s=1}^r\prod_{k=1}^{\Delta}\Pr_{z_i = \zeta,\text{ }z_j\sim \phi_{r-1,A}(\sigma_j(x))\text{ }\forall j\ne i}[\text{glue}(z_1,\hdots,z_{\Delta})_{sk} \in S_{sk}]\\
&= \left(\prod_{s=1}^{r-1}\prod_{k=1}^{\Delta}\textbf{1}_{x_{sk} \in S_{sk}}\right)\textbf{1}_{\zeta_{(r-1)B} \in S_{ri}}\left(\prod_{k\in [\Delta], k\ne i}\mu_{r-1}(S_{rk})\right)\\
\end{align*}

Note that $\sigma_i(\text{proj}(z)) = \text{proj}(\sigma_i(z)) = \text{end}_A(\text{res}_i(z)) = \text{end}_A(\zeta) = \sigma_i(x)$, so $\text{proj}(z) = x$. Thus, $z_{sk} = x_{sk}$ for all $s\in [r-1]$ and $k\in [\Delta]$. By the definition of conditional expectation in the proof of Proposition \ref{prop:cond-res-defn},

\begin{align*}
\E_{z\sim \mc R_r}[\textbf{1}_S(z) \mid \text{res}_i(z) = \zeta] &= \Pr_{z\sim \rho_{r,i}(\zeta)}[z\in S]\\
&= \prod_{s=1}^r\prod_{k=1}^{\Delta}\Pr_{z\sim \rho_{r,i}(\zeta)}[z_{sk}\in S_{sk}]\\
&= \left(\prod_{s=1}^{r-1}\prod_{k=1}^{\Delta}\Pr_{z\sim \rho_{r,i}(\zeta)}[z_{sk}\in S_{sk}]\right)\left(\prod_{k=1}^{\Delta}\Pr_{z\sim \rho_{r,i}(\zeta)}[z_{rk}\in S_{rk}]\right)\\
&= \left(\prod_{s=1}^{r-1}\prod_{k=1}^{\Delta}\textbf{1}_{x_{sk}\in S_{sk}}\right)\left(\prod_{k=1}^{\Delta}\Pr_{z\sim \rho_{r,i}(\zeta)}[z_{rk}\in S_{rk}]\right)\\
&= \left(\prod_{s=1}^{r-1}\prod_{k=1}^{\Delta}\textbf{1}_{x_{sk}\in S_{sk}}\right)\textbf{1}_{\zeta_{(r-1)B}\in S_{ri}}\left(\prod_{k\in [\Delta], k\ne i}\mu_{r-1}(S_{rk})\right)\\
\end{align*}

so the two expressions are equal, as desired.
\end{proof}

\end{document}